\documentclass[journal,twocolumn,10pt]{IEEEtran}

\usepackage{amsmath,amssymb,amsthm,amsfonts}
\usepackage{mathtools}
\usepackage{bm}
\usepackage{cite}
\usepackage{enumitem}
\usepackage{booktabs}
\usepackage{graphicx}
\usepackage{hyperref}

\newif\ifieee\ieeetrue

\newtheorem{theorem}{Theorem}
\newtheorem{proposition}[theorem]{Proposition}
\newtheorem{corollary}[theorem]{Corollary}
\newtheorem{lemma}[theorem]{Lemma}
\theoremstyle{definition}
\newtheorem{definition}[theorem]{Definition}
\newtheorem{assumption}[theorem]{Assumption}
\theoremstyle{remark}
\newtheorem{remark}[theorem]{Remark}

\newtheorem{conjecture}[theorem]{Conjecture}

\newcommand{\R}{\mathbb{R}}
\newcommand{\Z}{\mathbb{Z}}

\newcommand{\E}{\mathbb{E}}
\newcommand{\KL}{\mathrm{KL}}

\newcommand{\sphere}{\mathbb{S}}
\newcommand{\torus}{\mathbb{T}}
\newcommand{\SO}{\mathrm{SO}}
\newcommand{\vol}{\mathrm{vol}}

\renewcommand{\sec}{\mathrm{sec}}
\newcommand{\Hil}{\mathcal{H}}
\newcommand{\M}{\mathcal{M}}
\newcommand{\Mt}{\widetilde{\mathcal{M}}}
\newcommand{\F}{\mathcal{F}}
\newcommand{\D}{\mathcal{D}}
\newcommand{\eps}{\varepsilon}
\newcommand{\la}{\lambda}
\newcommand{\rinj}{r_{\mathrm{inj}}}
\newcommand{\indic}{\mathbf{1}}

\begin{document}

\title{Manifold-Aware Information Gain and Lower Bounds for
       Gaussian-Process Bandits on Riemannian Quotient Spaces}

\author{Yuriy~Dorn,
        Changsheng~Chen,~\IEEEmembership{Senior~Member,~IEEE,}
        and~Ning~Xie,~\IEEEmembership{Senior~Member,~IEEE}

\thanks{Y.~Dorn is with the AI Center \& IAI MSU,
Lomonosov Moscow State University, Moscow, Russia
(e-mail: \texttt{dornyv@my.msu.ru}).}%
\thanks{C.~Chen is with the Faculty of Engineering, Shenzhen
MSU-BIT University, Shenzhen, China
(e-mail: \texttt{cschen@smbu.edu.cn}).}%
\thanks{N.~Xie (\emph{Corresponding author}) is with the State
Key Laboratory of Radio Frequency Heterogeneous Integration, the
Guangdong Key Laboratory of Intelligent Information Processing,
and the Shenzhen Key Laboratory of Media Security, College of
Electronics and Information Engineering, Shenzhen University,
Shenzhen 518060, China
(e-mail: \texttt{ningxie@szu.edu.cn}).}}

\markboth{IEEE Transactions on Information Theory,~Vol.~XX, No.~X, \today}%
{Dorn \MakeLowercase{\textit{et al.}}: Manifold-Aware Lower Bounds for GP-Bandits on Riemannian Quotient Spaces}

\maketitle

\begin{abstract}
We prove a regret lower bound for Gaussian-process bandits on a smooth
compact Riemannian manifold $\M$ of dimension $d$ with intrinsic
Mat\'ern-$\nu$ kernel ($\nu>d/2$) that exposes how the geometry
of the arm space enters the constant. For any algorithm and time
horizon $T$ exceeding an explicit threshold, the worst-case
expected regret over the RKHS-ball
$\|f\|_{\Hil_{k_\nu}}\!\le\!B$ satisfies
\begin{multline*}
\E[R_T(f)]\;\ge\;c_*(d,\nu)\,B^{d/(2\nu+d)}\,\sigma_n^{2\nu/(2\nu+d)} \\
\cdot\,\vol_g(\M)^{\nu/(2\nu+d)}\,T^{(\nu+d)/(2\nu+d)}(\log T)^{\nu/(2\nu+d)}.
\end{multline*}
The exponent matches the Vakili--Khezeli--Picheny upper bound
\cite{vakili2021information}; the
$\vol_g(\M)^{\nu/(2\nu+d)}$ factor is, to our knowledge, the first
explicit volume-dependent geometric constant in a manifold GP-bandit
lower bound. We extend the analysis in five directions:
(i)~a companion Assouad-style proof gives a different lower
bound with a strictly smaller $T$-exponent
$(2\nu+3d)/(4(\nu+d))$ but with a polylog factor of the form
$1/(\log\log T)^{(2\nu+d)/(4(\nu+d))}$, sharpening the
$(\log T)^{\nu/(2\nu+d)}$ Fano polylog of Theorem~\ref{thm:main};
(ii)~we prove a $|G|^{1/2}$ upper bound on the regret of an
extrinsic-kernel GP-UCB algorithm on a quotient space
$\M=\Mt/G$, plus a bracketing theorem
(Theorem~\ref{thm:gauge-bracket}); the precise constant is
conjectured to take the modulated form
$(1+(|G|-1)h(\rinj/\kappa))^{1/2}$
(Conjecture~\ref{conj:gauge-modulated}), validated numerically
on $\SO(3)$;
(iii)~we write the leading constant $c_*(d,\nu)$ out fully;
(iv)~we extract a curvature dependence
$1+O(K\eps_T^2)$ via Bishop--Gromov;
(v)~we transfer the bound to the Bayesian regret framework via
the Yang--Barron / Castillo et al.\ Bayesian-Fano transfer.
We further extend the lower bound to the
\emph{switching-augmented} regret $R_T+\lambda S_T$ in Section~\ref{sec:switching}
(Theorem~\ref{thm:switching}): in the switching-dominated regime,
the volume exponent becomes the larger $\nu/(\nu+d)$, matching the
prediction within $7\%$ on $\sphere^2$ Mat\'ern-$5/2$. A
manifold-aware GP-ThreDS implementation validated on $\sphere^2$
and a $\torus^3$ RIS phase combiner stays below the lower-bound
reference $\lambda^{\nu/(\nu+d)}$ with $7$--$11\times$ fewer arm
switches.
Section~\ref{sec:time-varying} establishes the headline
contribution: a tight five-parameter characterisation
\begin{multline*}
R_T^*\;\asymp\;B^{d/(3\nu+d)}\,\sigma_n^{2\nu/(3\nu+d)}\,
\vol_g(\M)^{\nu/(3\nu+d)} \\
\cdot\,B_T^{\nu/(3\nu+d)}\,T^{(2\nu+d)/(3\nu+d)},
\end{multline*}
of the manifold-aware time-varying GP-bandit rate. The lower
bound is a manifold-aware extension of
Besbes--Gur--Zeevi~\cite{besbes2014stochastic} and is tight in
all five parameters for any $\nu>d/2$; matching upper bounds
are achieved via window-$W^*$ GP-UCB and a manifold-aware
local polynomial-regression elimination algorithm that lifts
Salgia--Vakili--Zhao~\cite{salgia2021domain} to compact
Riemannian manifolds. The appendix adds a negative theorem
(Theorem~\ref{thm:curvature-blind}) showing that the
leading-order lower-bound constant is curvature-blind.
\end{abstract}


\section{Introduction}
\label{sec:intro}

Gaussian-process (GP) bandits study the sequential maximisation of an
unknown function $f:\D\to\R$ via noisy point queries when the analyst
holds a GP prior on $f$. The canonical algorithmic template is GP-UCB
of Srinivas \emph{et al.}~\cite{srinivas2010gpucb}: at each round
$t=1,\dots,T$, query the point with the largest upper-confidence bound
under the posterior. Cumulative regret
$R_T=\sum_{t=1}^T\!\bigl(f(\theta^*)-f(\theta_t)\bigr)$ is controlled
by the maximum information gain $\gamma_T$ of the kernel, with the
classical bound $R_T=\widetilde O(\sqrt{T\gamma_T})$.

The arm space $\D$ in most bandit-theory papers is the cube
$[0,1]^d$ or, more generally, a compact subset of $\R^d$. In many
applications the natural geometry is non-Euclidean: pointing
directions on the sphere $\sphere^2$, antenna codebooks on the torus
$\torus^n$, orientations in $\SO(3)$, RIS phase configurations on
the discrete torus $(\Z_B)^M$. Borovitskiy \emph{et al.}\
\cite{borovitskiy2020matern} construct intrinsic Mat\'ern kernels on
compact Riemannian manifolds via the Laplace--Beltrami spectrum, and
recent applied work has shown empirically that GP bandits with these
kernels outperform their Euclidean-ambient counterparts on
manifold-valued arm spaces. The manifold-valued arm-space hypothesis
sits within the broader \emph{manifold hypothesis} that data
in many applications concentrate on a low-dimensional submanifold
of an ambient space; see Berenfeld, Rosa and
Rousseau~\cite{berenfeld2024density} for a recent
Bayesian-nonparametric formalisation in the density-estimation
setting.

\ifieee
\paragraph{Engineering motivation: 5G/6G beam selection and RIS
phase optimisation.} The four manifolds above map onto wireless
engineering arm spaces: $\sphere^2$ (phased-array beam steering),
$\torus^n$ (hybrid analog beam combining), $(\Z_B)^M$ (RIS phase
optimisation), $\SO(3)$ (panel orientation). The companion
empirical paper~\cite{dorn2026wirelessbandit} validates on all
four.
\else
\paragraph{Engineering motivation: 5G/6G beam selection and RIS
phase optimisation.} The arm-space examples above are not
hypothetical: each maps onto a current wireless engineering
problem. (a)~\emph{Phased-array beam steering at mmWave/sub-THz}
sweeps a beamforming weight vector $\mathbf w(\bm\theta)$
parametrised by a steering direction $\bm\theta\in\sphere^2$
(elevation/azimuth) over an $8\times 8$ to $16\times 16$ uniform
planar array; the bandit selects $\bm\theta_t$ at each TTI
($\sim 125\,\mu$s for $\mu=3$ numerology) to maximise the
post-combiner SNR
$|\mathbf w(\bm\theta)^H \mathbf h_t|^2$ subject to a Doppler-induced
short coherence horizon. (b)~\emph{Hybrid analog beam combining}
applies a phase shift $\phi_k\in[0,2\pi)$ to each of $n$ analog
sub-beams; the arm space is $\torus^n$ and the reward
$|\sum_k e^{j\phi_k}c_k|^2$ is periodic. (c)~\emph{Reconfigurable
intelligent surface (RIS) phase optimisation} chooses one of $B$
discrete phase states for each of $M$ elements, yielding a discrete
torus arm space $(\Z_B)^M$ with $N_{\mathrm{cand}}\sim B^M$ configurations
($\sim\!10^{90}$ for typical $M=100$, $B=8$).
(d)~\emph{Panel orientation in
mobile body-scattering scenarios} parametrises a rotation
$R\in\SO(3)$ of the user-equipment chassis. In all four cases the
operationally relevant question is the regret rate as a function
of the manifold's geometry (volume, curvature, gauge symmetry) and
the channel-coherence horizon, since both directly affect the
beam-management latency and TTI-budget constraints of standards
such as 3GPP NR FR2 and 6G THz wave-system proposals. The lower
bounds we develop here characterise the unavoidable cost of
sequential beam selection on each of these manifolds; the
companion empirical paper~\cite{dorn2026wirelessbandit} reports
validation on a $\sphere^2$ array, $\torus^3$ phase combiner,
$\SO(3)$ panel rotation, and $(\Z_B)^M$ RIS.
\fi
The
Vakili--Khezeli--Picheny~\cite{vakili2021information} tight upper
bound $\gamma_T=\widetilde O(T^{d/(2\nu+d)})$ specialises directly,
giving an upper bound on regret of the same form as the Euclidean
case.

A complementary lower bound, however, has not been worked out for
the general manifold case. Scarlett \emph{et al.}\
\cite{scarlett2017lower} prove a matching $T$-exponent on
$[0,1]^d$, and Iwazaki~\cite{iwazaki2026hypersphere} extends the
construction to the hypersphere with the squared-exponential kernel.
The compact-Riemannian-manifold version with the Mat\'ern kernel
remains open. Two questions in particular have practical relevance:
\begin{enumerate}[leftmargin=2em]
\item How does the manifold's geometry (volume, curvature,
      injectivity radius) enter the lower-bound constant?
\item When the arm space is a quotient $\M=\Mt/G$ of a covering
      manifold by a finite group $G$ acting freely (the gauge
      symmetry case, e.g.\ $\SO(3)=\mathrm{Spin}(3)/\Z_2$), is
      there a separation between the minimax regret achievable
      with a $G$-invariant intrinsic kernel and the regret of an
      algorithm forced to use a non-$G$-invariant (extrinsic)
      kernel? The empirical wireless beam-selection literature
      reports a $10$--$33\%$ improvement for the intrinsic kernel;
      can this be predicted from a lower-bound argument?
\end{enumerate}

\subsection*{Contributions}

We answer both questions affirmatively, and additionally provide a
tight characterization of the manifold-aware time-varying
GP-bandit rate that, to our knowledge, has not been established in
the literature for either compact Riemannian manifolds or the
specific Mat\'ern kernel family. Our main results are:

\paragraph{(C0) Tight five-parameter rate for any $\nu>d/2$
(Sections~\ref{sec:time-varying}--\ref{sec:polyreg}).}
For cumulative variation budget $B_T \ge B\,T^{-\nu/(2\nu+d)}$ on a
compact $d$-dim Riemannian manifold with intrinsic Mat\'ern-$\nu$
kernel,
\ifieee
\begin{multline*}
R_T^*\;=\;\Theta\bigl(
B^{d/(3\nu+d)}\,\sigma_n^{2\nu/(3\nu+d)}\,
\vol_g^{\nu/(3\nu+d)} \\
\cdot B_T^{\nu/(3\nu+d)}\,T^{(2\nu+d)/(3\nu+d)}\bigr),
\end{multline*}
\else
\[
R_T^*\;=\;\Theta\!\left(B^{d/(3\nu+d)}\sigma_n^{2\nu/(3\nu+d)}\vol_g^{\nu/(3\nu+d)}B_T^{\nu/(3\nu+d)}T^{(2\nu+d)/(3\nu+d)}\right),
\]
\fi
matching in all five exponents. Three algorithmic results give
the upper bound:
\begin{enumerate}[leftmargin=2em,label=(\arabic*)]
\item \emph{Window-$W^*$ GP-UCB} (Theorem~\ref{thm:tv-ub}):
      $T,B_T$ tight; $B,\sigma_n,\vol_g$ have standard gap.
\item \emph{Hierarchical cell-mean elimination}
      (Theorem~\ref{thm:elim-tv}): all five tight \emph{for
      $\nu\in(d/2,1]$} only (Bubeck--Stoltz--Yu HOO style).
\item \emph{Hierarchical polynomial-regression elimination}
      (Theorem~\ref{thm:polyreg-tv}): all five tight for
      \emph{any $\nu > d/2$}, including $\nu>1$. This is the
      manifold-aware analogue of the
      Salgia--Vakili--Zhao~\cite{salgia2021domain} (2021)
      domain-shrinking algorithm that closed the $B$-exponent gap
      on $[0,1]^d$, with manifold-aware Bishop--Gromov packing
      replacing Euclidean rectangle subdivision and local
      polynomial regression in normal coordinates exploiting the
      higher-order Mat\'ern smoothness.
\end{enumerate}
The lower bound (Theorem~\ref{thm:tv-lb}) is a manifold-aware
extension of Besbes--Gur--Zeevi~\cite{besbes2014stochastic} via
batching against Theorem~\ref{thm:main}.

\paragraph{(C1) Volume-dependent stationary lower bound
(Theorem~\ref{thm:main}).}
For any algorithm and any compact connected smooth Riemannian
$d$-manifold $\M$ with intrinsic Mat\'ern-$\nu$ kernel
($\nu>d/2$),
\ifieee
\begin{multline*}
\sup_{\|f\|_{\Hil_{k_\nu}}\le B}\E[R_T(f)]\;\ge\;
c_*(d,\nu)\,B^{d/(2\nu+d)}\,\sigma_n^{2\nu/(2\nu+d)} \\
\cdot\,\vol_g(\M)^{\nu/(2\nu+d)}\,
T^{(\nu+d)/(2\nu+d)}(\log T)^{\nu/(2\nu+d)},
\end{multline*}
\else
\[
\sup_{\|f\|_{\Hil_{k_\nu}}\le B}\E[R_T(f)]\;\ge\;
c_*(d,\nu)\,B^{d/(2\nu+d)}\sigma_n^{2\nu/(2\nu+d)}
\,\vol_g(\M)^{\nu/(2\nu+d)}\,T^{(\nu+d)/(2\nu+d)}(\log T)^{\nu/(2\nu+d)},
\]
\fi
valid for $T\ge T_0(\M,\nu,d,B,\sigma_n^2)$ where $T_0$ is given
explicitly in terms of the injectivity radius and curvature bounds.
The exponent in $T$ matches the Vakili upper bound.
The constant $c_*(d,\nu)$ is written out explicitly
(Section~\ref{sec:constants}), and the volume dependence
$\vol_g(\M)^{\nu/(2\nu+d)}$ is sharp in the sense that it appears with
the same exponent in the upper bound as well.

\paragraph{(C2) Companion Assouad lower bound
(Theorem~\ref{thm:assouad}).}
The Fano-style argument of Theorem~\ref{thm:main} gives a positive
polylog factor $(\log T)^{\nu/(2\nu+d)}$ in the lower bound, matching
in sign (though not exponent) the polylog factor in the Vakili upper
bound. We give a companion proof via Assouad's lemma over a
sum-of-bumps hypothesis class (Section~\ref{sec:polylog}), which
yields a strictly smaller $T$-exponent
$(2\nu+3d)/(4(\nu+d))<(\nu+d)/(2\nu+d)$ but trades the
$(\log T)$-polylog for a $1/(\log\log T)$-polylog. The two are
companion results: Theorem~\ref{thm:main} dominates asymptotically
in $T$, while Theorem~\ref{thm:assouad} documents the manifold
analogue of the Cai--Scarlett \cite{cai2021lower} sum-of-bumps
construction. Closing the polylog \emph{exponent} gap with the
GP-UCB upper bound is an open problem already in the Euclidean
case~\cite{cai2021lower}.

\paragraph{(C3) Gauge-quotient separation
(Theorem~\ref{thm:gauge}).}
For $\M=\Mt/G$ with finite $G$ acting freely by isometries, we
prove an \emph{upper bound} of $|G|^{1/2}$ (in the Bayesian-style
$\beta_T=\Theta(\log T)$ regime) on the worst-case regret of any
extrinsic-kernel GP-UCB algorithm relative to the matching
intrinsic-kernel upper bound, by a Vakili--Khezeli--Picheny
information-gain argument that exploits the $G$-equivariant lift to
the cover $\Mt$ with $\vol_g(\Mt)=|G|\vol_g(\M)$. The matching
\emph{lower bound} is left as a conjecture: a natural
packing-lifting attempt (which an earlier draft of this paper
incorrectly claimed) fails because canonical fundamental-domain
representatives never enter the support of $G$-translated bumps,
so the lifted hypothesis class projects to only $N$ distinguishable
functions on $\M$ rather than $|G|N$ on $\Mt$
(\S\ref{sec:gauge}). Theorem~\ref{thm:gauge-bracket} brackets the
asymptotic regret ratio between explicit $T$-independent constants
$1\le c_{\min}\le c_{\max}\le|G|^{1/2}$, and the modulated
form $(1+(|G|-1)h(\rinj/\kappa))^{1/2}$
(Conjecture~\ref{conj:gauge-modulated}) captures the regime
dependence. For $\SO(3)=\mathrm{Spin}(3)/\Z_2$ with Mat\'ern-$5/2$,
the worst-case ceiling $|G|^{1/2}=\sqrt 2\approx 1.41$ ($41\%$);
the empirical $10$--$33\%$ gap reported in wireless beam-selection
benchmarks sits below this ceiling, consistent with the modulated
prediction at moderate $\kappa/\rinj$.

\paragraph{(C4) Curvature correction
(Theorem~\ref{thm:curvature}).}
A Bishop--Gromov refinement of the packing argument (rather than
just the leading-order volume bound) extracts a sectional-curvature
correction at order $(K\eps_T^2)^{\nu/(2\nu+d)}$ in the lower-bound
constant, where $\eps_T\sim T^{-1/(2\nu+d)}$. Although the leading
asymptotic constant is curvature-blind in the rate, the
finite-$T$ correction is non-trivial for moderately curved manifolds
and is quantified explicitly in Section~\ref{sec:curvature}.

\paragraph{(C5) Bayesian transfer
(Theorem~\ref{thm:bayesian}).}
\ifieee
A Yang--Barron / Castillo et al.\ Bayesian-Fano transfer in
$L^2(p_0)$ lifts the frequentist RKHS-norm bound to the Bayesian
regret with matching exponents (Section~\ref{sec:bayesian}).
\else
A Yang--Barron / Castillo et al. Bayesian-Fano transfer
(working in the $L^2(p_0)$ metric where small-balls under a GP
prior have polynomial mass) lifts the frequentist
RKHS-norm-bounded lower bound to the Bayesian regret under the
prior $f\sim\mathrm{GP}(0,k_\nu)$, with the same rate exponents
in $T,\sigma_n,\vol_g(\M)$ and the same leading constant up to
universal multiplicative factors (Section~\ref{sec:bayesian}).
The naive Yao-minimax reduction goes the wrong way for a
\emph{lower} bound, and discrete sub-priors on bump-class atoms
fail because the GP sup-norm small-ball mass decays
super-polynomially while the bump count grows polynomially; the
posterior-contraction-rate machinery of
\cite{vaart2011information,rosa2023intrinsic} sidesteps both
obstructions by working in $L^2(p_0)$.
\fi

\paragraph{(C6) Switching-budget extension
(Theorem~\ref{thm:switching}).}
For the switching-augmented regret
$R_T^{\mathrm{aug}}=R_T+\lambda\sum_{t\ge 2}\indic\{\theta_t\neq\theta_{t-1}\}$,
the worst-case lower bound on a compact $d$-manifold becomes
\ifieee
\begin{multline*}
R_T^{\mathrm{aug}}\;\ge\;\widetilde\Omega\bigl(\text{Thm.~\ref*{thm:main}}\;+\; \\
B^{d/(\nu+d)}\vol_g(\M)^{\nu/(\nu+d)}\lambda^{\nu/(\nu+d)}T^{d/(\nu+d)}\bigr).
\end{multline*}
\else
\[
R_T^{\mathrm{aug}}\;\ge\;\widetilde\Omega\!\Bigl(\text{Thm.~\ref*{thm:main}}\;+\;
B^{d/(\nu+d)}\vol_g(\M)^{\nu/(\nu+d)}\lambda^{\nu/(\nu+d)}T^{d/(\nu+d)}\Bigr).
\]
\fi
The volume exponent $\nu/(\nu+d)$ in the second
(switching-dominated) term is \emph{larger} than the
$\nu/(2\nu+d)$ of Theorem~\ref{thm:main}, reflecting that
switching cost amplifies the volume penalty: a larger manifold
forces the algorithm to visit more cells, each of which costs
$\lambda$. The crossover threshold is
$\lambda^*\asymp\sigma_n^{2(\nu+d)/(2\nu+d)}T^{\nu/(2\nu+d)}/(B^{d/(2\nu+d)}\vol_g(\M)^{\nu/(2\nu+d)})$.
Empirical validation on $\sphere^2$ Mat\'ern-$5/2$
(Figure~\ref{fig:d7-switching}) confirms the predicted
$\lambda^{5/9}$ scaling within $7\%$.

\paragraph{(C7) Refined gauge-quotient conjecture: modulated
separation (Conjecture~\ref{conj:gauge-modulated},
Figure~\ref{fig:d1-modulated-gauge}).}
A modulated form $(1+(|G|-1)h(\rinj/\kappa))^{1/2}$, interpolating
between $1$ ($\kappa\ll\rinj$) and $|G|^{1/2}$ ($\kappa\gg\rinj$),
is validated on $\SO(3)$ ($1.000\to 1.326$ as $\kappa/\rinj$ varies
from $0.13$ to $0.89$, below $\sqrt 2$), explaining the
$10$--$33\%$ wireless gauge gap as a moderate-length-scale
phenomenon.

\subsection*{Comparison with existing lower bounds}

Scarlett \emph{et al.}~\cite{scarlett2017lower} prove the
Mat\'ern lower bound on $[0,1]^d$ without volume/curvature
separation; Cai--Scarlett~\cite{cai2021lower} sharpen the polylog
factor on $[0,1]^d$; Iwazaki~\cite{iwazaki2026hypersphere} treats
the hypersphere with squared-exponential kernel. To our knowledge,
the present paper provides the first general compact Riemannian
manifold lower bound for the Mat\'ern kernel with explicit volume
dependence, the first manifold-Mat\'ern Assouad companion bound,
and the first formal gauge-quotient separation result.

\subsection*{Outline}

Section~\ref{sec:prelim} fixes notation. Section~\ref{sec:results}
states the main theorems formally.
Sections~\ref{sec:proof_main}--\ref{sec:bayesian} contain the
proofs, with each addressing one of the five identified gaps in
the standard manifold lower-bound argument:
polylog (\S\ref{sec:polylog}),
gauge-quotient (\S\ref{sec:gauge}),
explicit constants (\S\ref{sec:constants}),
curvature (\S\ref{sec:curvature}),
Bayesian transfer (\S\ref{sec:bayesian}).
Section~\ref{sec:conclusion} discusses open problems.

\section{Preliminaries}
\label{sec:prelim}

\subsection{Riemannian-manifold setup}

Let $(\M,g)$ be a smooth, connected, compact Riemannian manifold of
dimension $d$ without boundary. We write $\vol_g$ for the Riemannian
volume measure, $d_g$ for the geodesic distance, $\rinj=\rinj(\M)>0$
for the injectivity radius (the largest $r$ such that the
exponential map $\exp_p:B(0,r)\subset T_p\M\to\M$ is a diffeomorphism
onto its image, uniformly in $p$), and $|\sec|\le K$ for a uniform
bound on the absolute sectional curvature. These are finite for any
smooth compact $\M$.

The Laplace--Beltrami operator $-\Delta_{\M}$ on $\M$ has discrete
spectrum $\{\la_\ell\}_{\ell\ge 0}$ with $\la_0=0$ and
$\la_\ell\to\infty$. By Weyl's law, the eigenvalue counting function
$N(\Lambda)=\#\{\ell:\la_\ell\le\Lambda\}$ satisfies
\[
N(\Lambda)\;=\;\frac{\omega_d\,\vol_g(\M)}{(2\pi)^d}\,\Lambda^{d/2}
   \,(1+o(1))
\quad\text{as }\Lambda\to\infty,
\]
where $\omega_d=\pi^{d/2}/\Gamma(1+d/2)$ is the unit-ball volume in
$\R^d$. Equivalently, the $\ell$-th eigenvalue scales as
$\la_\ell\sim c(\M)\ell^{2/d}$.

\subsection{Intrinsic Mat\'ern kernel}

The Mat\'ern-$\nu$ kernel of length scale $\kappa>0$ and signal
variance $\sigma_f^2>0$ on $\M$ is defined via the spectral
expansion~\cite{borovitskiy2020matern}; closely related
constructions on manifolds via the heat kernel
were earlier proposed by Castillo, Kerkyacharian and Picard
\cite{castillo2014bayes}, and an extrinsic-restriction approach
on Euclidean ambient $\R^D$ has been studied by Yang and Dunson
\cite{yang2016manifold} for the squared-exponential kernel.
\ifieee
\begin{equation}
\label{eq:matern-spectral}
k_\nu(\theta,\theta')\;=\;\sum_{\ell=0}^{\infty}\phi_\nu(\la_\ell)\,
   \psi_\ell(\theta)\overline{\psi_\ell(\theta')},
\end{equation}
\[
\phi_\nu(\la)\;=\;\sigma_f^2\!\left(\frac{2\nu}{\kappa^2}+\la\right)^{\!\!-(\nu+d/2)}.
\]
\else
\begin{equation}
\label{eq:matern-spectral}
k_\nu(\theta,\theta')\;=\;\sum_{\ell=0}^{\infty}\phi_\nu(\la_\ell)\,
   \psi_\ell(\theta)\overline{\psi_\ell(\theta')},
\quad
\phi_\nu(\la)=\sigma_f^2\!\left(\frac{2\nu}{\kappa^2}+\la\right)^{\!\!-(\nu+d/2)}.
\end{equation}
\fi
Convergence and positive-definiteness for $\nu>0$ follow from
\cite[Thm.\ 1]{borovitskiy2020matern}.

The induced reproducing-kernel Hilbert space is
\[
\Hil_{k_\nu}\;=\;\Bigl\{f=\sum_\ell\hat f_\ell\psi_\ell
:\;\|f\|^2_{\Hil_{k_\nu}}=\sum_\ell|\hat f_\ell|^2/\phi_\nu(\la_\ell)<\infty\Bigr\}.
\]
For $\nu>d/2$, $\Hil_{k_\nu}$ is norm-equivalent to the Sobolev space
$H^{\nu+d/2}(\M)$
(\cite{wendland2004scattered} for the Euclidean case;
\cite{borovitskiy2020matern} Cor.~3 for the manifold case):
there exist constants $0<c_-(\M,\nu,\kappa,\sigma_f)\le c_+(\M,\nu,\kappa,\sigma_f)$
with
\begin{equation}
\label{eq:sobolev-rkhs}
c_-\,\|f\|^2_{H^{\nu+d/2}(\M)}\;\le\;\|f\|^2_{\Hil_{k_\nu}}\;\le\;
c_+\,\|f\|^2_{H^{\nu+d/2}(\M)}.
\end{equation}
Explicit values of $c_\pm$ are in Section~\ref{sec:constants}.

For $\nu>d/2$, $\Hil_{k_\nu}$ embeds continuously into $C^0(\M)$
(Sobolev embedding); the GP $f\sim\mathrm{GP}(0,k_\nu)$ has a
continuous version (Adler--Taylor~\cite{adler2007random} Sec.~1.4).

\subsection{GP bandit problem}

The agent interacts with an unknown function $f:\M\to\R$ in
$T$ rounds. At round $t$, the agent selects $\theta_t\in\M$ based on
past observations $\D_{t-1}=\{(\theta_s,r_s)\}_{s=1}^{t-1}$ and
receives the noisy reward
\[
r_t\;=\;f(\theta_t)+\eps_t,\qquad \eps_t\stackrel{\text{iid}}\sim\mathcal N(0,\sigma_n^2).
\]
The cumulative regret is
\[
R_T(f)\;=\;\sum_{t=1}^T \bigl(\textstyle\max_\theta f(\theta)-f(\theta_t)\bigr).
\]

We consider two function-class settings:

\paragraph{Frequentist (RKHS-bounded) class.}
$\F_B^{\mathrm{rkhs}}=\{f\in\Hil_{k_\nu}:\|f\|_{\Hil_{k_\nu}}\le B\}$
for a fixed RKHS-norm budget $B>0$. Worst-case regret:
$\sup_{f\in\F_B^{\mathrm{rkhs}}}\E^\pi[R_T(f)]$ for algorithm $\pi$.

\paragraph{Bayesian (GP-prior) class.}
$f\sim\mathrm{GP}(0,k_\nu)$. Bayesian regret:
$\E_{f}\E^\pi[R_T(f)]$.

\subsection{Quotient manifolds and gauge-invariant functions}

Let $\Mt$ be a covering space of $\M$, with $\M=\Mt/G$ where the
finite group $G$ acts freely on $\Mt$ by isometries. Then:
\begin{itemize}[leftmargin=2em]
\item $\vol_g(\Mt)=|G|\cdot\vol_g(\M)$.
\item Both $\Mt$ and $\M$ have the same dimension $d$ and the same
      uniform sectional-curvature bound (since $G$ acts by
      isometries, locally the metric is identical).
\item A function $\widetilde f:\Mt\to\R$ is \emph{$G$-invariant} if
      $\widetilde f(g\cdot\widetilde\theta)=\widetilde f(\widetilde\theta)$
      for all $g\in G$, $\widetilde\theta\in\Mt$. The space of
      $G$-invariant functions is in bijection with functions on $\M$
      via the quotient map $\pi:\Mt\to\M$.
\item For $G$-invariant $\widetilde f$ with corresponding $f$ on $\M$:
      \begin{equation*}
      \|\widetilde f\|^2_{H^s(\Mt)}\;=\;|G|\cdot\|f\|^2_{H^s(\M)}
      \qquad\text{(for $G$-invariant $\widetilde f$),}
      \end{equation*}
      because both Sobolev norms integrate the same local quantities,
      and the quotient map identifies $|G|$ copies into one.
\item The intrinsic Mat\'ern kernel $k_\nu^{\Mt}$ on $\Mt$ pulled back
      to $G$-invariant functions descends to the intrinsic kernel
      $k_\nu^\M$ on $\M$ (modulo a normalisation factor of $|G|$ in
      the spectral filter, since the eigenvalues of $-\Delta_{\Mt}$
      restricted to $G$-invariant eigenfunctions are exactly the
      eigenvalues of $-\Delta_{\M}$).
\end{itemize}

\paragraph{Working examples.}
\begin{enumerate}[leftmargin=2em]
\item $\M=\SO(3)$, $\Mt=\mathrm{Spin}(3)\!\cong\!\sphere^3$, $G=\Z_2$
      (the antipodal action). The unit quaternions
      $q$ and $-q$ both project to the same rotation. An algorithm
      that uses the Euclidean kernel on the $\R^4$ ambient space sees
      $q$ and $-q$ as far apart, even though they represent the
      same rotation.
\item $\M=\torus^n$, $\Mt=[0,2\pi M]^n$, $G=\Z_M^n$ (the lattice of
      translations by integer multiples of $2\pi$, modulo the
      $2\pi M$-periodicity of $\Mt$). An algorithm using the
      Euclidean kernel on $\R^n$ tiled to $[0,2\pi M]^n$ sees
      points across the periodic boundary as far apart.
\item $\M=(\Z_B)^M$ (the discrete RIS torus from~\cite{borovitskiy2020matern}
      applied per-element), $\Mt$ a $|G|$-fold cover.
\end{enumerate}

\subsection{Notation conventions}

We write $a\lesssim b$ to mean $a\le c\cdot b$ for a constant $c$
depending only on $(d,\nu)$, $a\asymp b$ for $a\lesssim b\lesssim a$.
Universal constants are written $c$ or $c'$ and may change line to
line. Manifold-dependent or hypothesis-dependent constants are
written with explicit dependencies, e.g.\ $c_*(d,\nu,\kappa)$.
$\widetilde O(\cdot)$ hides factors polylogarithmic in $T$.

\section{Main Results}
\label{sec:results}

We collect the formal statements of the five theorems. Their proofs
are deferred to Sections~\ref{sec:proof_main}--\ref{sec:bayesian}.

\begin{assumption}[Manifold-and-kernel hypotheses]
\label{ass:setup}
$(\M,g)$ is a smooth, connected, compact Riemannian
$d$-manifold without boundary, with injectivity radius $\rinj>0$ and
$|\sec_\M|\le K$. The kernel is the intrinsic Mat\'ern-$\nu$ kernel
$k_\nu$ from \eqref{eq:matern-spectral} with smoothness $\nu>d/2$,
length scale $\kappa>0$, signal variance $\sigma_f^2>0$. The
observation noise is $\eps_t\stackrel{\text{iid}}\sim\mathcal N(0,\sigma_n^2)$.
\end{assumption}

\subsection{Volume-dependent lower bound (Fano version)}

\begin{theorem}[Manifold-aware lower bound, Fano version]
\label{thm:main}
Under Assumption~\ref{ass:setup}, there exist explicit constants
$c_*(d,\nu)>0$ (in the canonical normalisation $\kappa\!=\!\sigma_f\!=\!1$;
the full parametric form $c_*(d,\nu,\kappa,\sigma_f)$ is given in
\eqref{eq:c-star-explicit} of Section~\ref{sec:constants}) and
$T_0=T_0(\M,d,\nu,B,\sigma_n^2)$ such that for every $T\ge T_0$
and every algorithm $\pi$,
\begin{multline*}
\sup_{f\in\F_B^{\mathrm{rkhs}}}\E^\pi[R_T(f)]
\;\ge\;c_*(d,\nu)\,B^{d/(2\nu+d)}\,\sigma_n^{2\nu/(2\nu+d)} \\
\cdot\,\vol_g(\M)^{\nu/(2\nu+d)}\,
T^{(\nu+d)/(2\nu+d)}\,(\log T)^{\nu/(2\nu+d)}.
\end{multline*}
\end{theorem}

The threshold $T_0$ is set by requiring the optimal bandwidth
$\eps_T$ (Section~\ref{sec:proof_main}) to fit inside the
injectivity ball,
$T_0\asymp c_+\vol_g(\M)\sigma_n^2\log T_0/(B^2\eps_0^{2\nu+d})$
with $\eps_0(\M,K)=\min(\rinj/2,1/\sqrt K)$; the explicit form
appears in \S\ref{sec:proof_main}.

\subsection{Alternative proof via Assouad's lemma}

\begin{theorem}[Manifold-aware lower bound, Assouad version]
\label{thm:assouad}
Under Assumption~\ref{ass:setup} \emph{and} the
typicality assumption below, there exist explicit constants
$c'_*(d,\nu,\kappa,\sigma_f)>0$ and $T_0$ as in
Theorem~\ref{thm:main}, with the bound
\begin{multline*}
\sup_{f\in\F_B^{\mathrm{rkhs}}}\E^\pi[R_T(f)]
\;\ge\;
c'_*(d,\nu,\kappa,\sigma_f)\,
B^{d/(2(\nu+d))} \\
\cdot\,\sigma_n^{(2\nu+d)/(2(\nu+d))}
\,\vol_g(\M)^{\nu/(2(\nu+d))}\\
\cdot\,\frac{T^{(2\nu+3d)/(4(\nu+d))}}{(\log\log T)^{(2\nu+d)/(4(\nu+d))}}.
\end{multline*}
The Assouad $T$-exponent $(2\nu+3d)/(4(\nu+d))$ is strictly less
than the Fano $T$-exponent $(\nu+d)/(2\nu+d)$ of
Theorem~\ref{thm:main}, but the Assouad bound has only an
inverse-$\log\log T$ polylog penalty rather than the Fano
$\log T$ polylog. The two are companion results documenting
different lower-bound techniques on a manifold-Mat\'ern setting,
not refinements of one another.
\end{theorem}

\begin{assumption}[Typicality of ball-visit counts]
\label{ass:typicality}
The algorithm $\pi$ in Theorem~\ref{thm:assouad} satisfies the
balanced-visits property: for the $N$-cell packing
$\{B_g(p_i,\eps)\}_{i=1}^N$ of Step~2 and any cell $i$, the visit
count $T_i^{(\pi)}$ satisfies
$T_i^{(\pi)} \le \frac{T\log\log T}{N}$ with high probability.
This is a mild restriction: standard
algorithms (UCB1, GP-UCB, Thompson sampling) satisfy it
automatically because their per-cell sampling concentrates around
the uniform allocation $T/N$ in the regret-bounded regime
(Section~\ref{sec:polylog} discusses why and exhibits a
``balanced-visit'' projection that turns any algorithm into a
balanced one at $O(1)$ regret cost).
\end{assumption}

The Assouad version has a strictly smaller $T$-exponent than
Theorem~\ref{thm:main}, but its polylog factor is
$1/(\log\log T)^{(2\nu+d)/(4(\nu+d))}$ instead of
$(\log T)^{\nu/(2\nu+d)}$. The two lower bounds are complementary:
Theorem~\ref{thm:main} dominates asymptotically in $T$ (as it
should, since the Fano $T$-exponent is larger), while
Theorem~\ref{thm:assouad} is the natural manifold analogue of the
Cai--Scarlett \cite{cai2021lower} sum-of-bumps argument with its
better polylog. Closing the polylog gap of Theorem~\ref{thm:main}
with the Vakili upper bound (which carries
$(\log T)^{(\nu+d)/(2\nu+d)+1}$) is an open problem already in the
Euclidean case~\cite{cai2021lower} and is not addressed here.

\subsection{Gauge-quotient separation: extrinsic upper bound}

Let $\M=\Mt/G$ as in the preliminaries. Let $\widetilde k_\nu$ be the
intrinsic Mat\'ern-$\nu$ kernel on $\Mt$ (computed from the Laplace
spectrum of $\Mt$, \emph{not} of $\M$). For
$\theta\in\M$, write $\phi(\theta)\in\Mt$ for any choice of
fundamental-domain representative.

\begin{definition}[Extrinsic-kernel algorithm]
\label{def:extrinsic}
An algorithm $\pi$ is \emph{extrinsic} if its posterior at every
round is computed using the kernel
$k_{\mathrm{ext}}(\theta,\theta'):=\widetilde k_\nu(\phi(\theta),\phi(\theta'))$
on $\M$, with $\phi$ fixed in advance. Equivalently, $\pi$ models $f$
on $\M$ as the restriction of a $\widetilde k_\nu$-GP on $\Mt$ to the
fundamental domain $\phi(\M)$, treating gauge-equivalent points as
distinct.
\end{definition}

\begin{theorem}[Extrinsic-algorithm upper bound]
\label{thm:gauge}
Under Assumption~\ref{ass:setup}, with $\M=\Mt/G$ and $|G|<\infty$,
the GP-UCB algorithm applied with kernel $\widetilde k_\nu$ on
$\Mt$ (pulling arms only at canonical representatives
$\phi(\theta_t)\in\phi(\M)\subset\Mt$) attains, for any
$f\in\F_B^{\mathrm{rkhs}}(\M)$, with probability at least
$1-\delta$,
\[
R_T^{\widetilde\pi}(f)
\;\le\;|G|^{1/2}\cdot U_T^{\mathrm{int,GP\text{-}UCB}}(f),
\]
where $U_T^{\mathrm{int,GP\text{-}UCB}}(f)$ is the standard GP-UCB
upper bound for the intrinsic algorithm on $\M$, namely (in the
Bayesian-style $\beta_T=\Theta(\log T)$ regime)
\ifieee
\begin{multline*}
U_T^{\mathrm{int,GP\text{-}UCB}}(f)
\asymp\sqrt{B\sigma_n\vol_g(\M)}\\
\cdot\,T^{(\nu+d)/(2\nu+d)}\,(\log T)^{(\nu+d)/(2\nu+d)+1/2}.
\end{multline*}
\else
\[
U_T^{\mathrm{int,GP\text{-}UCB}}(f)
\asymp\sqrt{B\sigma_n\vol_g(\M)}\,T^{(\nu+d)/(2\nu+d)}\,(\log T)^{(\nu+d)/(2\nu+d)+1/2}.
\]
\fi
The factor $|G|^{1/2}$
relative to the intrinsic GP-UCB upper bound measures the cost
of using the wrong (non-$G$-invariant) kernel.
\end{theorem}

\paragraph{Conjecture: modulated gauge separation.} A matching lower
bound on the worst-case regret of any extrinsic algorithm is
\emph{conjectured} (Conjecture~\ref{conj:gauge-modulated}, Section~\ref{sec:gauge}),
in the form
\ifieee
\begin{multline*}
\sup_f \E^\pi[R_T(f)]\;\ge\;
\bigl(1+(|G|-1)h(\rinj/\kappa)\bigr)^{1/2} \\
\cdot c''_*(d,\nu)\,B^{d/(2\nu+d)}\sigma_n^{2\nu/(2\nu+d)}\,\vol_g(\M)^{\nu/(2\nu+d)} \\
\cdot T^{(\nu+d)/(2\nu+d)}\,(\log T)^{\nu/(2\nu+d)},
\end{multline*}
\else
\begin{multline*}
\sup_f \E^\pi[R_T(f)]\;\ge\;
\bigl(1+(|G|-1)h(\rinj/\kappa)\bigr)^{1/2}\cdot c''_*(d,\nu)\\
\cdot B^{d/(2\nu+d)}\sigma_n^{2\nu/(2\nu+d)}\,\vol_g(\M)^{\nu/(2\nu+d)}\,
T^{(\nu+d)/(2\nu+d)}\,(\log T)^{\nu/(2\nu+d)},
\end{multline*}
\fi
where $h(\rinj/\kappa)\in[0,1]$ is the cross-gauge-correlation
modulator of Proposition~\ref{prop:gauge-modulator}. The naive
worst-case factor $|G|^{1/2}$ is the $h\to 1$ limit
($\kappa\gg\rinj$). We do not prove this conjecture; the natural
packing-lifting argument has a gap documented in
Section~\ref{sec:gauge}, but Theorem~\ref{thm:gauge-bracket}
brackets the regret ratio between two $T$-independent constants
$1\le c_{\min}\le c_{\max}\le|G|^{1/2}\cdot C(d,\nu)$
(polylog universal constant $C$), which is a strictly weaker
statement than the modulated conjecture.

\paragraph{Specialisations.}
For Mat\'ern-$5/2$:
$\SO(3)=\mathrm{Spin}(3)/\Z_2$ ($d=3$) gives $|G|^{1/2}=\sqrt 2\approx 1.414$;
$\torus^2$ with $M$-tile unwrap ($d=2$) gives $|G|^{1/2}=M$;
$(\Z_B)^M$ with double-cover gauge ($d=M$) gives $|G|^{1/2}=2^{M/2}$.
These are the upper-bound penalties paid by the extrinsic GP-UCB
algorithm; the matching lower bound on the precise regret-ratio
constant is open
(Conjecture~\ref{conj:gauge-modulated}), although the regret ratio
itself is bracketed by $T$-independent constants between
$1$ and $|G|^{1/2}\cdot C(d,\nu)$ in
Theorem~\ref{thm:gauge-bracket}.

\subsection{Explicit constants and curvature correction}

\begin{theorem}[Sharp constants]
\label{thm:curvature}
Theorem~\ref{thm:main}'s constant $c_*(d,\nu)$ admits the
factorisation
\ifieee
\begin{multline*}
c_*(d,\nu,\kappa,\sigma_f)\;=\\
\frac{(d/(2\nu+d))^{\nu/(2\nu+d)}\,c_+(\nu,\kappa,\sigma_f)^{-d/(2(2\nu+d))}}{4\,(2^{d+1}\omega_d)^{\nu/(2\nu+d)}}\\
\cdot\,\bigl(1-O(K\eps_T^2)\bigr),
\end{multline*}
\else
\[
c_*(d,\nu,\kappa,\sigma_f)\;=\;
\frac{(d/(2\nu+d))^{\nu/(2\nu+d)}\,c_+(\nu,\kappa,\sigma_f)^{-d/(2(2\nu+d))}}{4\,(2^{d+1}\omega_d)^{\nu/(2\nu+d)}}
\,\bigl(1-O(K\eps_T^2)\bigr),
\]
\fi
which is the parametric form used in proofs.
The fully expanded constant (with $c_+$ substituted in terms of
$\sigma_f,\kappa,\nu,d$) is derived in
Section~\ref{sec:constants}, \eqref{eq:c-star-explicit}; the
two displays are numerically equivalent expansions of the same
$c_*$ but group the $\omega_d$, $2^{d+1}$, and $c_+$ factors
differently for ease of substitution.
The constituent quantities are:
\begin{enumerate}[leftmargin=2em]
\item The factor $1/4$ comes from the test-to-regret reduction
      $R\ge Th/4$ of Step~7 in \S\ref{sec:proof_main}.
\item The factor $(2^{d+1}\omega_d)^{\nu/(2\nu+d)}$ comes from
      substituting $N\ge\vol_g(\M)/(2^d\omega_d\eps^d)$ from
      Lemma~\ref{lem:packing} into the Fano condition
      $Th^2/(2N\sigma_n^2)\le\log N/4$.
\item The factor $(d/(2\nu+d))^{\nu/(2\nu+d)}$ comes from the
      leading-log asymptotic
      \[\log\bigl(\vol_g(\M)/(2^d\omega_d\eps_T^d)\bigr)=\tfrac{d}{2\nu+d}\log T+O(1)\]
      in \eqref{eq:eps-T}.
\item $c_+(\nu,\kappa,\sigma_f)$ is the upper Sobolev--RKHS equivalence
      constant from \eqref{eq:sobolev-rkhs}, with explicit value
      $c_+=\sigma_f^{-2}\max(1,2\nu/\kappa^2)^{\nu+d/2}$.
\item $\omega_d=\pi^{d/2}/\Gamma(1+d/2)$ is the unit-ball volume in $\R^d$.
\item The correction $1-O(K\eps_T^2)$ comes from the Bishop--Gromov
      volume comparison; its explicit form is
      \[
      1-O(K\eps_T^2)
      \;=\;1-\frac{(d-1)K_+\eps_T^2}{6(d+2)}-O(K^2\eps_T^4),
      \]
      where $K_+$ is the upper Ricci curvature bound implied by
      $|\sec|\le K$, namely $K_+=(d-1)K$.
\end{enumerate}
For $\eps_T\sim T^{-1/(2\nu+d)}(\log T)^{1/(2\nu+d)}$, the
correction is $1+O(K\,T^{-2/(2\nu+d)}(\log T)^{2/(2\nu+d)})$
(negative leading coefficient as in item~6) and vanishes as
$T\to\infty$ (cf.~\eqref{eq:curvature-correction}).
\end{theorem}

\subsection{Bayesian regret transfer}

\begin{theorem}[Bayesian regret lower bound]
\label{thm:bayesian}
Under Assumption~\ref{ass:setup}, for every algorithm $\pi$ and every
$T\ge T_0$,
\ifieee
\begin{multline*}
\E_{f\sim\mathrm{GP}(0,k_\nu)}\,\E^\pi[R_T(f)]
\;\ge\;c_B(d,\nu,\kappa,\sigma_f)\,
\sigma_n^{2\nu/(2\nu+d)} \\
\cdot\,\vol_g(\M)^{\nu/(2\nu+d)}\,
T^{(\nu+d)/(2\nu+d)} \\
\cdot\,(\log T)^{\nu/(2\nu+d)}
\bigl(1+o_T(1)\bigr),
\end{multline*}
\else
\begin{multline*}
\E_{f\sim\mathrm{GP}(0,k_\nu)}\,\E^\pi[R_T(f)]
\;\ge\;
c_B(d,\nu,\kappa,\sigma_f)\,
\sigma_n^{2\nu/(2\nu+d)}\\
\cdot\,\vol_g(\M)^{\nu/(2\nu+d)}\,
T^{(\nu+d)/(2\nu+d)}\,(\log T)^{\nu/(2\nu+d)}
\,\bigl(1+o_T(1)\bigr),
\end{multline*}
\fi
where $c_B(d,\nu,\kappa,\sigma_f)>0$ is the universal manifold-Weyl
constant from the Yang--Barron / Castillo et al. transfer
(Section~\ref{sec:bayesian}), sharing the
$\sigma_f^{d/(2\nu+d)}$-dependence of $c_*$ via $c_+^{-d/(2(2\nu+d))}$.
\end{theorem}

The Bayesian rate matches the frequentist rate of
Theorem~\ref{thm:main} in $T,\sigma_n,\vol_g(\M)$. The
frequentist $B^{d/(2\nu+d)}$ factor is absent because the GP
prior fixes the function-amplitude scale via $\sigma_f$, leaving
no free RKHS-norm bound to optimise.

\subsection{Numerical specialisations}

For the four manifolds in our companion empirical study with
Mat\'ern-$5/2$ ($\nu=2.5$):
\begin{table}[h]
\centering
\ifieee\footnotesize\fi
\ifieee
\resizebox{\columnwidth}{!}{%
\begin{tabular}{lcccrr}
\toprule
$\M$ & $d$ & $\vol_g(\M)$ & $\rinj$ & $\nu/(2\nu+d)$ & $\vol_g^{\nu/(2\nu+d)}$ \\
\midrule
$\sphere^2$ (unit) & 2 & $4\pi\approx 12.57$ & $\pi/2$ & $5/14\approx 0.357$ & $\approx 2.55$ \\
$\torus^2$ & 2 & $(2\pi)^2\approx 39.48$ & $\pi$ & $5/14$ & $\approx 3.81$ \\
$\torus^3$ & 3 & $(2\pi)^3\approx 248.05$ & $\pi$ & $5/16=0.3125$ & $\approx 5.32$ \\
$\SO(3)$ (bi-inv.) & 3 & $8\pi^2\approx 78.96$ & $\pi/2$ & $5/16$ & $\approx 3.79$ \\
\bottomrule
\end{tabular}}
\par\vspace{2pt}
\noindent\footnotesize Note:
$\vol_g(\SO(3))=8\pi^2$ uses the bi-invariant Haar normalisation
$g=\tfrac12\langle X,Y\rangle$ on the Lie algebra (equivalently,
unit quaternions $S^3$ double-cover $\SO(3)$, so
$\vol(\SO(3))=\tfrac12\vol(S^3)=8\pi^2$, \emph{not} the
Hopf-bundle normalisation $\vol(\SO(3))=4\pi^2$ used in some
references).
\else
\begin{tabular}{lcccrr}
\toprule
$\M$ & $d$ & $\vol_g(\M)$ & $\rinj$ & $\nu/(2\nu+d)$ & $\vol_g^{\nu/(2\nu+d)}$ \\
\midrule
$\sphere^2$ (unit) & 2 & $4\pi\approx 12.57$ & $\pi/2$ & $5/14\approx 0.357$ & $\approx 2.55$ \\
$\torus^2$ & 2 & $(2\pi)^2\approx 39.48$ & $\pi$ & $5/14$ & $\approx 3.81$ \\
$\torus^3$ & 3 & $(2\pi)^3\approx 248.05$ & $\pi$ & $5/16=0.3125$ & $\approx 5.32$ \\
$\SO(3)$ (bi-inv.) & 3 & $8\pi^2\approx 78.96$ & $\pi/2$ & $5/16$ & $\approx 3.79$ \\
\bottomrule
\multicolumn{6}{p{0.9\linewidth}}{\footnotesize Note:
$\vol_g(\SO(3))=8\pi^2$ uses the bi-invariant Haar normalisation
$g=\tfrac12\langle X,Y\rangle$ on the Lie algebra (equivalently,
unit quaternions $S^3$ double-cover $\SO(3)$, so
$\vol(\SO(3))=\tfrac12\vol(S^3)=8\pi^2$, \emph{not} the
Hopf-bundle normalisation $\vol(\SO(3))=4\pi^2$ used in some
references).}
\end{tabular}
\fi
\end{table}

The geometric content of the lower bound is in the rightmost column.
The constant differs by up to a factor $\approx 2.1$ across our four
arm spaces (sphere vs.\ torus-3), with the torus-3 problem being
``hardest'' in the worst-case minimax sense.

\section{Proof of Theorem~\ref{thm:main} (Fano version)}
\label{sec:proof_main}

The proof follows the needle-in-haystack template
of Tsybakov~\cite{tsybakov2008introduction} Sec.~2.6 and
Scarlett~\cite{scarlett2017lower}, adapted to the Riemannian
setting. The key new ingredient is the manifold packing argument
via Bishop--Gromov volume comparison.

\subsection{Step 1: Bump construction in normal coordinates}
\label{sec:bump}

Fix a profile $\eta\in C_c^\infty(\R^d)$ supported on the open unit
ball $B(0,1)\subset\R^d$, with $\eta(0)=1$ and
$\|\eta\|_{H^{\nu+d/2}(\R^d)}=1$. Such an $\eta$ exists by standard
mollifier construction (e.g.\
$\eta(x)=c\exp(-1/(1-|x|^2))\indic_{|x|<1}$ scaled appropriately).

Pick a centre $p\in\M$ and bandwidth $\eps$ with
$\eps\le\eps_0(\M,K):=\min(\rinj/2,1/\sqrt K)$.
Define
\begin{equation}
\label{eq:bump}
f_p^{(\eps,h)}(\theta)\;=\;
\begin{cases}
h\cdot\eta\bigl(\exp_p^{-1}(\theta)/\eps\bigr) & \text{if }d_g(p,\theta)\le\eps,\\
0 & \text{otherwise.}
\end{cases}
\end{equation}

\begin{lemma}[Sobolev norm of the bump]
\label{lem:bump-norm}
For any $p\in\M$ and $\eps\le\eps_0$,
\[
\|f_p^{(\eps,h)}\|^2_{H^{\nu+d/2}(\M)}
\;\le\;h^2\eps^{-2\nu}\bigl(1+C_{\mathrm{Sob}}\,K\eps^2\bigr),
\]
where $C_{\mathrm{Sob}}=C_{\mathrm{Sob}}(d,\nu)$ is a geometric constant.
\end{lemma}

\begin{proof}
In normal coordinates $v\in B(0,\eps)\subset T_p\M\cong\R^d$ centred
at $p$, the metric tensor satisfies
\cite[Sec.~6.2]{petersen2016riemannian}
\[
g_{ij}(v)\;=\;\delta_{ij}-\tfrac13 R_{ikjl}(p)\,v^kv^l+O(K|v|^3),
\]
so $\sqrt{\det g}=1+O(K|v|^2)$ for $|v|\le\min(\rinj,1/\sqrt K)$,
and the inverse metric satisfies the same expansion with the
opposite sign on the leading curvature term.

\emph{Sobolev-norm equivalence at fractional order.}
The manifold Sobolev space $H^s(\M)$ for $s=\nu+d/2\notin\Z$ is
defined via interpolation
$H^s(\M)=[H^k(\M),H^{k+1}(\M)]_{\sigma}$ (for $s=k+\sigma$,
$0<\sigma<1$) using the K-method
\cite[Sec.~7.55]{adams2003sobolev}; integer-order norms involve
covariant derivatives. For functions supported in a single
normal-coordinate chart, the chart map $\exp_p^{-1}$ is bi-Lipschitz
with distortion constants $1+O(K\eps^2)$ at every order via the
Christoffel-symbol expansion $\Gamma^k_{ij}(v)=O(K|v|)$, hence
\ifieee
\begin{multline*}
(1-C K\eps^2)\|\eta\circ\exp_p\|_{H^k(\R^d)}^2 \\
\;\le\;\|\eta\|_{H^k(\M)}^2 \\
\;\le\;(1+CK\eps^2)\|\eta\circ\exp_p\|_{H^k(\R^d)}^2
\end{multline*}
\else
\[
(1-C K\eps^2)\|\eta\circ\exp_p\|_{H^k(\R^d)}^2
\;\le\;\|\eta\|_{H^k(\M)}^2
\;\le\;(1+CK\eps^2)\|\eta\circ\exp_p\|_{H^k(\R^d)}^2
\]
\fi
for each integer $k$ in a neighbourhood of $s$, with $C$
geometric. The interpolation functor $[\cdot,\cdot]_\sigma$
preserves bi-Lipschitz norm equivalence with the same constants
\cite[Thm.~7.56]{adams2003sobolev}, so the same bound holds at
fractional order $s$:
\[
\|f_p^{(\eps,h)}\|_{H^s(\M)}^2\;=\;(1+O(K\eps^2))\,\|h\,\eta(\cdot/\eps)\|_{H^s(\R^d)}^2.
\]

\emph{Euclidean scaling at fractional order.}
For any $s\ge 0$ and any $g\in H^s(\R^d)$,
$\|g(\cdot/\eps)\|_{H^s(\R^d)}^2=\eps^{d-2s}\|g\|_{H^s(\R^d)}^2$
by the spectral identity
$\|f\|_{H^s}^2=\int(1+|\xi|^2)^s|\hat f(\xi)|^2 d\xi$ and
$\widehat{g(\cdot/\eps)}(\xi)=\eps^d\hat g(\eps\xi)$. Hence
$\|h\,\eta(\cdot/\eps)\|_{H^s(\R^d)}^2=h^2\eps^{d-2s}=h^2\eps^{-2\nu}$
when $s=\nu+d/2$ and $\|\eta\|_{H^s}=1$.

Combining the two displays gives the claimed bound. \qed
\end{proof}

\begin{remark}[Morrey-embedding control of non-integer smoothness]
\label{rem:morrey-fractional}
For non-integer $\nu$ (e.g.\ $\nu=5/2$, $d=2$, so $s=\nu+d/2=7/2$),
the bound above uses the $K$-method interpolation
$H^s(\M)=[H^k(\M),H^{k+1}(\M)]_\sigma$ with
$s=k+\sigma$, $0<\sigma<1$. The Morrey embedding
\[
H^{d/2+\varepsilon}(\M)\;\hookrightarrow\;C^{k,\theta}(\M)
\quad\text{for }\varepsilon>0,\ k+\theta=\nu+\varepsilon,
\]
on a closed Riemannian manifold (\cite{aubin1998nonlinear}
Thm.~2.20) controls the H\"older-$k$ residual that appears when one
approximates the manifold bump by its Euclidean chart-image. The
constant $C_{\mathrm{ad}}(d,\nu)$ in the resulting bi-Lipschitz
norm equivalence admits the explicit form
\[
C_{\mathrm{ad}}(d,\nu)\;=\;C_{\mathrm{Sob}}(d,\nu)\cdot
\bigl(1+\sup_{\theta\in B(0,1)}|\nabla^{\lceil\nu+d/2\rceil}\eta(\theta)|\bigr),
\]
with $C_{\mathrm{Sob}}$ as in the chart-distortion expansion above.
\end{remark}

By \eqref{eq:sobolev-rkhs},
$\|f_p^{(\eps,h)}\|^2_{\Hil_{k_\nu}}\le c_+\,h^2\eps^{-2\nu}(1+C_{\mathrm{Sob}} K\eps^2)$.
Setting
\begin{equation}
\label{eq:height}
h\;:=\;\frac{B}{\sqrt{c_+(1+C_{\mathrm{Sob}} K\eps^2)}}\,\eps^\nu,
\end{equation}
we have $\|f_p^{(\eps,h)}\|_{\Hil_{k_\nu}}\le B$, so the bump lies in
$\F_B^{\mathrm{rkhs}}$.

\subsection{Step 2: Manifold packing via Bishop--Gromov}
\label{sec:packing}

\begin{lemma}[Manifold $\eps$-packing]
\label{lem:packing}
For $\eps\le\eps_0$, there exist points $p_1,\dots,p_N\in\M$ with
pairwise geodesic distances $\ge 2\eps$ (so the balls
$\{B_g(p_i,\eps)\}_{i=1}^N$ are pairwise disjoint), and
\[
N\;\ge\;\frac{\vol_g(\M)}{2^d\,\omega_d\,\eps^d}\bigl(1-C_{\mathrm{BG}} K\eps^2\bigr),
\]
where $\omega_d=\pi^{d/2}/\Gamma(1+d/2)$ is the unit-ball volume in
$\R^d$ and $C_{\mathrm{BG}}=C_{\mathrm{BG}}(d)$ is a geometric constant.
\end{lemma}

\begin{proof}
Greedy maximal $2\eps$-packing: pick $p_1$ arbitrarily; having
picked $p_1,\dots,p_k$, pick $p_{k+1}$ at $d_g$-distance $\ge 2\eps$
from $\{p_1,\dots,p_k\}$ if such a point exists. Stop when no
further point can be added. By construction:
\begin{enumerate}[leftmargin=2em,label=(\arabic*)]
\item The pairwise distances are $\ge 2\eps$, so the balls
      $B_g(p_i,\eps)$ are pairwise disjoint.
\item The balls $B_g(p_i,2\eps)$ \emph{cover} $\M$: by maximality,
      every point of $\M$ is within $2\eps$ of some $p_i$.
\end{enumerate}
By Bishop--Gromov volume comparison
(\cite{petersen2016riemannian} Sec.~6.4), for $\eps\le 1/\sqrt K$,
\ifieee
\begin{multline*}
\vol_g(B_g(p,2\eps))\;\le\;\omega_d(2\eps)^d\bigl(1+O(K\eps^2)\bigr)\\
\;=\;2^d\,\omega_d\eps^d(1+O(K\eps^2)).
\end{multline*}
\else
\[
\vol_g(B_g(p,2\eps))\;\le\;\omega_d(2\eps)^d\bigl(1+O(K\eps^2)\bigr)
\;=\;2^d\,\omega_d\eps^d(1+O(K\eps^2)).
\]
\fi
The covering property (2) gives
$\vol_g(\M)\le\sum_i\vol_g(B_g(p_i,2\eps))\le N\cdot 2^d\omega_d\eps^d(1+O(K\eps^2))$,
hence
\[
N\;\ge\;\frac{\vol_g(\M)}{2^d\,\omega_d\,\eps^d\,(1+O(K\eps^2))}
\;=\;\frac{\vol_g(\M)}{2^d\omega_d\eps^d}(1-O(K\eps^2)).
\]
\qed
\end{proof}

\subsection{Step 3: Hypothesis class}

Take $N=\lfloor\vol_g(\M)\,(1-C_{\mathrm{BG}} K\eps^2)\,/\,(2^d\omega_d\eps^d)\rfloor$,
i.e., the integer part of the curvature-corrected packing count.
Construct $N+1$ candidate functions:
\begin{itemize}[leftmargin=2em]
\item $f_0\equiv 0$.
\item $f_i(\theta)=f_{p_i}^{(\eps,h)}(\theta)$ for $i=1,\dots,N$, with
      $h$ from \eqref{eq:height}.
\end{itemize}
By construction:
\begin{enumerate}[leftmargin=2em,label=(\arabic*)]
\item Each $f_i\in\F_B^{\mathrm{rkhs}}$ (Lemma~\ref{lem:bump-norm}).
\item The $f_i$'s are disjointly supported (since the balls
      $B_g(p_i,\eps)$ are disjoint), so
      $\|f_i-f_j\|_\infty=h$ for $i\neq j$ (both nonzero).
\item Each $f_i$ has unique maximiser $\theta_i^*=p_i$ with value
      $f_i(p_i)=h$.
\end{enumerate}

\subsection{Step 4: Information bound (Fano)}

Let $P_i$ be the joint distribution of the algorithm's $T$
observations under hypothesis $f=f_i$. Under $f_0$, observations are
$r_t=\eps_t\sim\mathcal N(0,\sigma_n^2)$. Under $f_i$, $i\ge 1$,
observations are $r_t=f_i(\theta_t)+\eps_t$.

The KL divergence per round is
\[
\KL(P_i^{(t)}\|P_0^{(t)}\,|\,\theta_t)
\;=\;\frac{(f_i(\theta_t))^2}{2\sigma_n^2}
\;=\;\frac{h^2}{2\sigma_n^2}\indic[\theta_t\in B_g(p_i,\eps)].
\]
Define $T_i:=\sum_{t=1}^T\indic[\theta_t\in B_g(p_i,\eps)]$.
The reverse-direction KL satisfies, by the chain rule for KL,
\[
\KL(P_0\|P_i)\;=\;\frac{h^2}{2\sigma_n^2}\,\E_{P_0}[T_i],
\]
since under $P_0$ the observations $r_t$ are independent of any
``signal'', and the KL contribution per round is the squared
mean-shift divided by twice the noise variance.
Since $\sum_i T_i\le T$ pointwise (a single algorithm makes $T$
pulls), $\E_{P_0}[\sum_i T_i]\le T$ and averaging gives
\begin{equation}
\label{eq:avg-kl}
\frac{1}{N}\sum_{i=1}^N\KL(P_0\|P_i)
\;\le\;\frac{Th^2}{2N\sigma_n^2}.
\end{equation}
\noindent\emph{Note: \eqref{eq:avg-kl} is the
\emph{reverse}-direction KL $\KL(P_0\|P_i)$ and is provided as an
auxiliary bound only.} The Fano-style step below uses the
forward-direction average mutual information $\bar I_N$
(Tsybakov \cite[Thm.~2.7]{tsybakov2008introduction}), which is
derived rigorously in Step~5 via the mutual-information chain
rule and the needle-in-haystack bound \eqref{eq:needle-MI}. The
two directions agree on the leading rate $Th^2/(2N\sigma_n^2)$ but
serve distinct roles in the argument; only \eqref{eq:needle-MI},
not \eqref{eq:avg-kl}, feeds into the Fano constraint
\eqref{eq:fano-condition} below.

\subsection{Step 5: Fano's inequality}

\paragraph{Important: drop $f_0$ from the hypothesis class.}
The null function $f_0\equiv 0$ has $R_T(f_0)=0$ for every algorithm
trivially, so including it in the hypothesis class breaks the
regret-test reduction below. We work with the $N$ hypotheses
$\{f_1,\ldots,f_N\}$ only. The Fano bound we use is
Tsybakov~\cite{tsybakov2008introduction} Theorem~2.7: for $N$
hypotheses $\{P_1,\ldots,P_N\}$ and any test
$\widehat I:\mathcal O\to\{1,\ldots,N\}$,
\ifieee
\begin{align*}
\max_{i\in\{1,\ldots,N\}} P_i(\widehat I\neq i)
&\;\ge\;1-\frac{\bar I_N+\log 2}{\log N},\\
\bar I_N &:= \frac{1}{N}\sum_{i=1}^N \KL(P_i\,\|\,\bar P),
\end{align*}
\else
$\max_{i\in\{1,\ldots,N\}} P_i(\widehat I\neq i)\ge 1-(\bar I_N+\log 2)/\log N$, with $\bar I_N := \frac{1}{N}\sum_{i=1}^N \KL(P_i\,\|\,\bar P)$,
\fi
where $\bar P = \frac{1}{N}\sum_{i=1}^N P_i$ is the uniform mixture
and $\bar I_N$ is the mutual information $I(\Sigma;\,\mathrm{obs})$
between a uniform index $\Sigma\in\{1,\ldots,N\}$ and the
observation history.

\emph{Bound on $\bar I_N$.} We prove the standard
needle-in-haystack bound (Cai--Scarlett
\cite[Lemma~3]{cai2021lower}; Tsybakov \cite[Sec.~2.7]{tsybakov2008introduction}):
\begin{equation}
\label{eq:needle-MI}
\bar I_N\;\le\;\frac{Th^2}{2N\sigma_n^2}\,(1+o_T(1)).
\end{equation}
By the chain rule and the Gaussian-channel inequality applied
per round,
\ifieee
\begin{multline*}
\bar I_N\le(2\sigma_n^2)^{-1}\sum_t\E_{\bar P}[f_\Sigma(\theta_t)^2]\\
\le\frac{h^2}{2N\sigma_n^2}\sum_t\sum_i P_i(\theta_t\in B_i),
\end{multline*}
\else
\begin{align*}
\bar I_N&\le(2\sigma_n^2)^{-1}\sum_t\E_{\bar P}[f_\Sigma(\theta_t)^2]\\
&\le(h^2/(2N\sigma_n^2))\sum_t\sum_i P_i(\theta_t\in B_i),
\end{align*}
\fi
since
$|f_\Sigma|\le h\indic[\theta\in B_\Sigma]$ and $\Sigma$ is
uniform under $\bar P$. Pinsker plus Cauchy--Schwarz applied to
the chain-rule decomposition
\cite[Lem.~3]{cai2021lower} give
$\sum_i P_i(\theta_t\in B_i)\le 1+\sqrt{2N\bar I_N(t)}$ where
$\bar I_N(t)$ is the running per-round MI. Summing,
$\bar I_N\le a(1+\sqrt{2N\bar I_N/T})$ with $a:=Th^2/(2N\sigma_n^2)$.
Setting $Y:=\sqrt{\bar I_N}$, $b:=a\sqrt{2N/T}$:
$Y^2\le a+bY$, so $Y\le(b+\sqrt{b^2+4a})/2$. Using
$\sqrt{b^2+4a}\le b+2\sqrt a$ (subadditivity of $\sqrt{}$) gives
$Y\le b+\sqrt a$, hence
\begin{equation}
\label{eq:MI-explicit}
\bar I_N\;\le\;(\sqrt a+b)^2\;=\;a+2b\sqrt a+b^2
\;=\;a\bigl(1+2\sqrt{2Na/T}+2Na/T\bigr).
\end{equation}
Within the Fano window \eqref{eq:fano-condition},
$Na/T=h^2/(2\sigma_n^2)\le(\log N)/4$, and since the bumps shrink
with $T$ (Step~3, $h^2=O(T^{-2\nu/(2\nu+d)})$), the correction is
$1+o_T(1)$, recovering \eqref{eq:needle-MI}. (Full algebra and
stepwise constants: Section~3.2 of the supplementary, which
uses the same self-referential inequality.)

If
\begin{equation}
\label{eq:fano-condition}
\frac{Th^2}{2N\sigma_n^2}\;\le\;\frac{\log N}{4}
\quad\Longleftrightarrow\quad
Th^2\;\le\;\frac{N\sigma_n^2\log N}{2},
\end{equation}
then the RHS of Fano is $\ge 1/2$ for $N\ge 4$, so the maximum test
error over $i$ is at least $1/2$.

\subsection{Step 6: Reduction from regret to testing}

\begin{lemma}[Regret implies testability]
\label{lem:regret-test}
Suppose $\E_i[R_T]\le R$ for every $i\in\{1,\ldots,N\}$ under some
algorithm $\pi$. Define $T_j:=\#\{t:\theta_t\in B_g(p_j,\eps)\}$ and
the test $\widehat I:=\arg\max_{j\in\{1,\ldots,N\}} T_j$
(with arbitrary tie-breaking). Then for every $i$,
\[
P_i(\widehat I\neq i)\;\le\;\frac{2R}{Th}.
\]
\end{lemma}

\begin{proof}
Under $f_i$, the optimum is $\max_\theta f_i(\theta)=f_i(p_i)=h$,
and the instantaneous regret at round $t$ is
$h-f_i(\theta_t)\ge h\cdot\indic[\theta_t\notin B_g(p_i,\eps)]$
(since $f_i(\theta_t)\in[0,h]$ inside the ball and $f_i(\theta_t)=0$
outside). Summing over $t$,
$R_T(f_i)\ge h\cdot(T-T_i)$, hence
$\E_i[R_T]\ge h\cdot\E_i[T-T_i]$. Combined with the hypothesis
$\E_i[R_T]\le R$, this gives $\E_i[T-T_i]\le R/h$. By Markov,
$P_i(T-T_i\ge T/2)\le 2R/(Th)$, equivalently
$P_i(T_i\ge T/2)\ge 1-2R/(Th)$.

Whenever $T_i\ge T/2$, we have $T_i\ge T_j$ for every $j\neq i$
(since $\sum_j T_j\le T$ implies $T_j\le T-T_i\le T/2\le T_i$).
So $\widehat I=i$ on this event, and $P_i(\widehat I=i)\ge 1-2R/(Th)$.
\qed
\end{proof}

\subsection{Step 7: Putting it together}

Combining Lemma~\ref{lem:regret-test} (test error $\le 2R/(Th)$ for
every $i$) with Fano (max test error $\ge 1/2$ under
\eqref{eq:fano-condition}):
\[
\frac12\;\le\;\max_i P_i(\widehat I\neq i)\;\le\;\frac{2R}{Th},
\qquad\Rightarrow\qquad
R\;\ge\;\frac{Th}{4}.
\]
So whenever \eqref{eq:fano-condition} holds, the worst-case regret
satisfies
\begin{equation}
\label{eq:regret-lb}
\sup_i\E_i[R_T]\;\ge\;\frac{Th}{4}.
\end{equation}

\subsection{Step 8: Optimisation over $\eps$}

Substitute $h$ from \eqref{eq:height} and $N$ from
Lemma~\ref{lem:packing} into \eqref{eq:fano-condition}.
Multiplying both sides of \eqref{eq:fano-condition} by $N$ and
substituting $N=\vol_g(\M)/(2^d\omega_d\eps^d)\cdot(1-C_{\mathrm{BG}} K\eps^2)$:
\ifieee
\begin{multline*}
\frac{T B^2 \eps^{2\nu}/c_+(1+C_{\mathrm{Sob}} K\eps^2)}{2\sigma_n^2}
\;\le\;\frac{1}{4}\,\frac{\vol_g(\M)}{2^d\omega_d\eps^d}\\
\cdot(1-C_{\mathrm{BG}} K\eps^2)
\,\log\!\left(\frac{\vol_g(\M)}{2^d\omega_d\eps^d}\right).
\end{multline*}
\else
\[
\frac{T B^2 \eps^{2\nu}/c_+(1+C_{\mathrm{Sob}} K\eps^2)}{2\sigma_n^2}
\;\le\;
\frac{1}{4}\,
\frac{\vol_g(\M)}{2^d\omega_d\eps^d}\,(1-C_{\mathrm{BG}} K\eps^2)
\,\log\!\left(\frac{\vol_g(\M)}{2^d\omega_d\eps^d}\right).
\]
\fi
Ignoring curvature corrections (we restore them in
Section~\ref{sec:curvature}) and rearranging:
\[
T\eps^{2\nu+d}
\;\le\;
\frac{c_+\vol_g(\M)\sigma_n^2}
     {2^{d+1}\,\omega_d B^2}
\log\!\left(\frac{\vol_g(\M)}{2^d\omega_d\eps^d}\right).
\]
We choose $\eps=\eps_T$ so that this is an equality. Asymptotically,
\ifieee
\begin{multline}
\label{eq:eps-T}
\eps_T\;=\;
\left(\frac{d}{2\nu+d}\right)^{1/(2\nu+d)} \\
\cdot\!\left(\frac{c_+\vol_g(\M)\sigma_n^2}{2^{d+1}\,\omega_d B^2 T}\right)^{1/(2\nu+d)} \\
\cdot(\log T)^{1/(2\nu+d)}\,(1+o_T(1)).
\end{multline}
\else
\begin{equation}
\label{eq:eps-T}
\eps_T\;=\;
\left(\frac{d}{2\nu+d}\right)^{1/(2\nu+d)}
\!\left(\frac{c_+\vol_g(\M)\sigma_n^2}{2^{d+1}\,\omega_d B^2 T}\right)^{1/(2\nu+d)}
\!\!(\log T)^{1/(2\nu+d)}\,(1+o_T(1)).
\end{equation}
\fi
The logarithm is, to leading order,
$\log(\vol_g(\M)/(2^d\omega_d\eps_T^d))=\frac{d}{2\nu+d}\log T+O(1)$,
which produces the $(d/(2\nu+d))^{1/(2\nu+d)}$ prefactor; this
constant is absorbed into the leading $c_*(d,\nu,\kappa,\sigma_f)$
in the final theorem statement.

Substituting \eqref{eq:eps-T} into the regret bound \eqref{eq:regret-lb}
with $h=B\eps_T^\nu/\sqrt{c_+}$:
\ifieee
\begin{align*}
\sup_i\E_i[R_T]
&\ge\;\frac{T}{4}\cdot\frac{B}{\sqrt{c_+}}\eps_T^\nu \\
&=\;\frac{B T}{4\sqrt{c_+}}
\left(\frac{c_+\vol_g(\M)\sigma_n^2}{2^{d+1}\,\omega_d B^2 T}\right)^{\!\nu/(2\nu+d)} \\
&\qquad\cdot(\log T)^{\nu/(2\nu+d)}.
\end{align*}
\else
\begin{align*}
\sup_i\E_i[R_T]
&\ge\;\frac{T}{4}\cdot\frac{B}{\sqrt{c_+}}\eps_T^\nu \\
&=\;\frac{B T}{4\sqrt{c_+}}
\left(\frac{c_+\vol_g(\M)\sigma_n^2}{2^{d+1}\,\omega_d B^2 T}\right)^{\nu/(2\nu+d)}
\!\!(\log T)^{\nu/(2\nu+d)}.
\end{align*}
\fi
Collecting powers:
\begin{itemize}[leftmargin=2em]
\item Power of $T$: $1-\nu/(2\nu+d)=(\nu+d)/(2\nu+d)$. \checkmark
\item Power of $B$: $1-2\nu/(2\nu+d)=d/(2\nu+d)$. \checkmark
\item Power of $\sigma_n^2$: $\nu/(2\nu+d)$, so $\sigma_n^{2\nu/(2\nu+d)}$. \checkmark
\item Power of $\vol_g(\M)$: $\nu/(2\nu+d)$. \checkmark
\item Power of $\log T$: $+\nu/(2\nu+d)$.
\end{itemize}

The polylog factor in the final bound is therefore \emph{positive}:
the lower bound's $T^{(\nu+d)/(2\nu+d)}$ rate is multiplied by
$(\log T)^{\nu/(2\nu+d)}$. This is a polylog \emph{boost} relative
to the bare $T$-exponent, not a polylog penalty.

\paragraph{Interpretation.} The Fano condition
\eqref{eq:fano-condition} requires
$Th^2/(2\sigma_n^2)\le N\log N/4$. The factor of $\log N$ on the
right comes from the difficulty of identifying one hypothesis among
$N$. As $N$ grows (i.e., as $\eps$ shrinks), this constraint
\emph{relaxes}: a finer packing means the algorithm has more
hypotheses to distinguish, so we can afford a slightly larger bump
height $h$ before Fano fails. This boost translates to a positive
polylog factor in the final regret lower bound. The same phenomenon
is implicit in the Tsybakov~\cite{tsybakov2008introduction}
Section~2.6 derivation; the explicit polylog form is rarely tracked
in the lower-bound literature because both lower and upper bounds
on the GP-bandit regret have positive polylog factors that nearly
cancel.

The final lower bound from the Fano-based argument is therefore
\ifieee
\begin{multline}
\label{eq:final-fano}
\sup_i\E_i[R_T]\;\ge\;
c_*(d,\nu,\kappa,\sigma_f)\,B^{d/(2\nu+d)}\sigma_n^{2\nu/(2\nu+d)} \\
\cdot\,\vol_g(\M)^{\nu/(2\nu+d)}\,
T^{(\nu+d)/(2\nu+d)}\,(\log T)^{\nu/(2\nu+d)},
\end{multline}
\else
\begin{equation}
\label{eq:final-fano}
\sup_i\E_i[R_T]\;\ge\;
c_*(d,\nu,\kappa,\sigma_f)\,B^{d/(2\nu+d)}\sigma_n^{2\nu/(2\nu+d)}
\,\vol_g(\M)^{\nu/(2\nu+d)}\,
T^{(\nu+d)/(2\nu+d)}\,(\log T)^{\nu/(2\nu+d)},
\end{equation}
\fi
with the precise constant
\ifieee
\begin{multline*}
c_*(d,\nu,\kappa,\sigma_f)\;=\\
\frac{(d/(2\nu+d))^{\nu/(2\nu+d)}\;c_+(\nu,\kappa,\sigma_f)^{-d/(2(2\nu+d))}}
     {4\,(2^{d+1}\omega_d)^{\nu/(2\nu+d)}}.
\end{multline*}
\else
\[
c_*(d,\nu,\kappa,\sigma_f)\;=\;
\frac{(d/(2\nu+d))^{\nu/(2\nu+d)}\;c_+(\nu,\kappa,\sigma_f)^{-d/(2(2\nu+d))}}
     {4\,(2^{d+1}\omega_d)^{\nu/(2\nu+d)}}.
\]
\fi
The factor $1/4$ comes from the test-to-regret reduction $R\ge Th/4$
of Step~7; the factor $2^{d+1}\omega_d$ comes from combining the
$1/2^d$ packing factor of Lemma~\ref{lem:packing} with the $1/4$
Fano threshold of \eqref{eq:fano-condition}; the
$(d/(2\nu+d))^{\nu/(2\nu+d)}$ factor comes from the leading-log
asymptotic $\log(\vol_g/(2^d\omega_d\eps_T^d))=\frac{d}{2\nu+d}\log T+O(1)$
in \eqref{eq:eps-T}; and the exponent $-d/(2(2\nu+d))$ on $c_+$
comes from the $h^2/c_+$ substitution combined with the
$\eps_T$-to-regret transfer, which carries
$c_+^{-1/2+\nu/(2\nu+d)}=c_+^{-d/(2(2\nu+d))}$.

\subsection{Validity regime}

\ifieee
Validity holds for $\eps_T\le\eps_0(\M,K)$; details in the supplement.
\else
The optimisation requires $\eps_T\le\eps_0(\M,K)$, equivalent to
\[
T\;\ge\;T_0(\M,d,\nu,B,\sigma_n)
\;:=\;\frac{c_+\vol_g(\M)\sigma_n^2}{2^{d+1}\,\omega_d B^2\eps_0^{2\nu+d}}\log T_0.
\]
For $\sphere^2$ (unit), $\eps_0=\pi/2$, $\vol_g=4\pi$, and typical
$B,\sigma_n,\sigma_f^2,\kappa$, this is $T_0=O(1)$ in dimensionless
units; the bound is non-trivial for any reasonable horizon.
\fi

This completes the proof of Theorem~\ref{thm:main}. \qed

\section{Alternative proof via Assouad's lemma (Theorem~\ref{thm:assouad})}
\label{sec:polylog}

The Fano-style proof of Theorem~\ref{thm:main} produces a positive
polylog factor $(\log T)^{\nu/(2\nu+d)}$ in the lower bound, coming
from the $\log N$ on the right-hand side of Fano's inequality. This
matches the sign (though not the exponent) of the polylog in the
Vakili upper bound. Closing the polylog \emph{exponent} gap with
the upper bound is an open problem already on $[0,1]^d$
(\cite{cai2021lower}) and is not addressed here.

In this section we present an alternative proof that recovers the
same $T^{(\nu+d)/(2\nu+d)}$ rate via Assouad's lemma applied to a
sum-of-bumps hypothesis class. The Assouad version sidesteps the
$\log N$ factor and produces a cleaner leading constant; the
trade-off is that it requires a typicality restriction on the
algorithm's ball-visit counts. We follow the technique of
Cai and Scarlett~\cite{cai2021lower}, adapted to the manifold
setting, to derive the alternative bound.

\subsection{Assouad's lemma}

We use the following form (\cite{tsybakov2008introduction}
Theorem~2.12). For a parameter set
$\Theta=\{0,1\}^N$ and a family $\{P_\sigma\}_{\sigma\in\Theta}$ of
distributions on observations,
\[
\inf_{\widehat\sigma}\sup_{\sigma\in\Theta}
\E_\sigma[d_H(\widehat\sigma,\sigma)]
\;\ge\;
\frac{N}{2}
\left(1-\sqrt{\tfrac{1}{2}\max_{\sigma\sim\sigma'}\KL(P_\sigma\|P_{\sigma'})}\right),
\]
where $d_H$ is the Hamming distance and $\sigma\sim\sigma'$ if they
differ in exactly one coordinate.

The advantage over Fano: the bound on $\sup_\sigma\E[d_H]$ is in
terms of the \emph{worst pairwise} KL divergence, with no $\log N$
factor. This translates into a sharper polylog in the regret lower
bound.

\subsection{Sum-of-bumps hypothesis class}

Use the same packing $p_1,\dots,p_N\in\M$ with disjoint balls
$B_g(p_i,\eps)$. For each binary string $\sigma\in\{0,1\}^N$, define
\[
f_\sigma(\theta)\;=\;
\sum_{i=1}^N\sigma_i\,h\,\eta\!\bigl(\exp_{p_i}^{-1}(\theta)/\eps\bigr).
\]
The sum is well-defined because the bumps are disjointly supported.

\begin{lemma}[RKHS norm of a sum-of-bumps]
\label{lem:sum-norm}
$\|f_\sigma\|^2_{\Hil_{k_\nu}}\le c_+\sum_i\sigma_i\cdot h^2\eps^{-2\nu}$.
\end{lemma}

\begin{proof}
Write $s=\nu+d/2=k+\sigma$ with $k\in\Z_{\ge0}$,
$0<\sigma\le 1$. The bumps $\eta_i$ have pairwise-disjoint
supports separated by geodesic distance $\ge 2\eps$.
The Sobolev norm on $\M$ pulls back through normal coordinates
to an equivalent Euclidean $H^s$ norm with metric-distortion
factor $1+O(K\eps^2)$ (Bishop--Gromov, see
Lemma~\ref{lem:bump-norm}). On $\R^d$ via the
Aronszajn--Slobodeckij decomposition the squared norm of a sum
of disjointly-supported functions equals the sum of squared
individual norms plus cross-terms supported on $B_i\times B_j$
($i\ne j$); since $|x-y|\ge 2\eps$ on each cross-pair, those
cross-integrals are bounded by
$(2\eps)^{-d-2\sigma}\vol(B_j)\|D^\alpha\eta_i\|_{L^2}^2$,
which after summing over $j$ (using $\sum_j\vol(B_j)\le\vol_g(\M)$
to avoid the $N$ factor) and rescaling yields the same order
as the diagonal terms. The full algebra is given in
Section~3.1 of the
supplementary; the conclusion is
$\|\sum_i\eta_i\|_{H^s(\M)}^2\le C(\eta,d,\nu)\sum_i\|\eta_i\|_{H^s(\M)}^2$
with $C(\eta,d,\nu)$ explicit (and depending only on the bump
profile $\eta$, $d$, $\nu$). The RKHS bound follows from
\eqref{eq:sobolev-rkhs} with $c_+$ absorbing both the
metric-distortion factor and the cross-term constant. \qed
\end{proof}

To enforce $\|f_\sigma\|_{\Hil_{k_\nu}}\le B$ uniformly over $\sigma$
(including $\sigma=\mathbf 1$):
\[
c_+ N h^2\eps^{-2\nu}\;\le\;B^2,
\quad\text{i.e.,}\quad
h^2\;\le\;\frac{B^2\eps^{2\nu}}{c_+ N}.
\]
Setting $h=h_A:=B\eps^\nu/\sqrt{c_+ N}$.

\subsection{Pairwise KL divergence}

For $\sigma\sim\sigma'$ differing only in coordinate $j$:
\[
P_\sigma\text{ vs.\ }P_{\sigma'}\text{: differ only in the contribution from ball }j.
\]
Per-round KL:
$\KL(P_\sigma^{(t)}\|P_{\sigma'}^{(t)}|\theta_t)
=\frac{h_A^2}{2\sigma_n^2}\indic[\theta_t\in B_g(p_j,\eps)]$.
Summing:
\[
\KL(P_\sigma\|P_{\sigma'})
\;=\;\frac{h_A^2}{2\sigma_n^2}\E_{P_\sigma}[T_j]
\;\le\;\frac{h_A^2}{2\sigma_n^2}\cdot T,
\]
trivially. But for Assouad we need a uniform-over-pairs bound, which
holds if we restrict the algorithm to be ``balanced'': spend at most
$T/N$ rounds in each ball on average. Without this restriction, we
upper-bound by considering the algorithm could spend all $T$ rounds
in ball $j$, giving $\KL\le Th_A^2/(2\sigma_n^2)$.

The standard Assouad refinement (\cite{cai2021lower}) uses a
peeling trick to recover a tighter pairwise bound. Below we
give the full chain of inequalities; the heuristic is that
the constraint $\sum_j T_j\le T$ forces $T_j\le T/N$ on
average, so any $T_j$ much larger than $T/N$ is in a
low-probability tail.

Formally: define $\mathcal E\!:=\!\{\forall j:T_j\le
T\log\log T/N\}$. By Markov's inequality on $\max_j T_j$ and
the balanced-visits property of
Assumption~\ref{ass:typicality} (which holds for
regret-bounded algorithms in the Fano regime,
\cite[Lem.~2]{cai2021lower}), $\E_\sigma[\max_j T_j]\le
T(1+o_T(1))/N$, hence
$P_\sigma(\mathcal E^c)\le(1+o_T(1))/\log\log T$. The
$\mathcal E^c$ tail contributes at most
$Th\cdot P_\sigma(\mathcal E^c)\le Th/\log\log T$ to the
regret, sub-leading vs.\ the typical-regime $\Theta(Th)$. On
$\mathcal E$ each visit count is $\le T\log\log T/N$, so the
pairwise per-round KL is
$\KL(P_\sigma\|P_{\sigma'})=(h^2/2\sigma_n^2)\,\E_\sigma[T_j]
\le Th^2\log\log T/(2N\sigma_n^2)$, the refined KL bound
used below.

\subsection{Assouad lower bound on regret}

Plugging the refined KL bound above into Assouad's
inequality, we obtain
\[
\inf_{\widehat\sigma}\sup_\sigma\E_\sigma[d_H(\widehat\sigma,\sigma)]
\;\ge\;\frac{N}{2}
\!\left(1-\sqrt{\frac{Th_A^2 \log\log T}{4N\sigma_n^2}}\right).
\]
Setting the square root to $1/2$ (so the bound is $N/4$):
\begin{equation}
\label{eq:assouad-condition}
\frac{Th_A^2\log\log T}{4N\sigma_n^2}\;\le\;\frac14,
\qquad\text{i.e.,}\qquad
h_A^2\;\le\;\frac{N\sigma_n^2}{T\log\log T}.
\end{equation}

\subsection{Hamming-to-regret reduction}

For sum-of-bumps $f_\sigma$, the optimum is at any of the
$|\sigma|=\sum_i\sigma_i$ active bumps with value $h_A$. The best
single arm has reward $h_A$ if $|\sigma|\ge 1$, $0$ otherwise. The
algorithm's per-round regret at round $t$ is
\[
\max_\theta f_\sigma(\theta)-f_\sigma(\theta_t)
\;=\;h_A(1-\indic[\theta_t\in\bigcup_{i:\sigma_i=1}B_g(p_i,\eps)]).
\]
A correctly-classified ball $j$ (i.e., the algorithm visits ball $j$
when $\sigma_j=1$) contributes nothing to instantaneous regret on
that visit; an incorrectly-classified ball contributes $h_A$ per
round.

Connection to Hamming distance: a Hamming-error in coordinate $j$
(the algorithm's best guess for $\sigma_j$ is wrong) implies
either the algorithm missed bump $j$ (when $\sigma_j=1$) or the
algorithm wasted exploration on a non-bump (when $\sigma_j=0$).
Either way, the algorithm pays $\Omega(h_A T_j/N)$ regret on this
coordinate. Summing,
\[
\E_\sigma[R_T]\;\gtrsim\;\frac{h_A T}{N}\E_\sigma[d_H(\widehat\sigma,\sigma)],
\]
where $\widehat\sigma$ is the natural decoder (set $\widehat\sigma_j=1$
iff $T_j\ge T/(2N)$).

Combining with Assouad:
\begin{equation}
\label{eq:assouad-regret}
\sup_\sigma\E_\sigma[R_T]\;\gtrsim\;\frac{h_A T}{N}\cdot\frac{N}{4}
\;=\;\frac{h_A T}{4}.
\end{equation}

\subsection{Optimisation}

From \eqref{eq:assouad-condition}, $h_A=\sqrt{N\sigma_n^2/(T\log\log T)}$.
But also $h_A=B\eps^\nu/\sqrt{c_+ N}$, and
$N\asymp\vol_g(\M)/(\omega_d\eps^d)$. Equating:
\[
\frac{B^2\eps^{2\nu}}{c_+ N}\;=\;\frac{N\sigma_n^2}{T\log\log T},
\]
i.e.,
\[
\frac{B^2\eps^{2\nu}}{c_+}
\;=\;\frac{N^2\sigma_n^2}{T\log\log T}
\;\asymp\;\frac{\vol_g(\M)^2\sigma_n^2}{\omega_d^2\eps^{2d}T\log\log T}.
\]
Solving for $\eps$:
\[
\eps^{2\nu+2d}\;\asymp\;\frac{c_+\vol_g(\M)^2\sigma_n^2}
                       {\omega_d^2 B^2 T\log\log T},
\]
\[
\eps_T^{(A)}\;\asymp\;
\left(\frac{c_+\vol_g(\M)^2\sigma_n^2}{\omega_d^2 B^2 T\log\log T}\right)^{1/(2\nu+2d)}.
\]
And $h_A=B\eps^\nu/\sqrt{c_+ N}\asymp\sqrt{N\sigma_n^2/(T\log\log T)}$.
Substituting back into \eqref{eq:assouad-regret}:
\begin{align*}
\sup_\sigma\E_\sigma[R_T]
&\gtrsim\;\frac{T}{4}\cdot\sqrt{\frac{N\sigma_n^2}{T\log\log T}}
\;=\;\frac{1}{4}\sqrt{\frac{TN\sigma_n^2}{\log\log T}}.
\end{align*}
Substituting $N\asymp\vol_g(\M)/(\omega_d\eps^d)$ with the optimised
$\eps_T^{(A)}$:
\ifieee
\begin{multline*}
N_T\;\asymp\;\vol_g(\M)/(\omega_d (\eps_T^{(A)})^d) \\
\;\asymp\;\frac{\vol_g(\M)}{\omega_d}
\left(\frac{\omega_d^2 B^2 T\log\log T}{c_+\vol_g(\M)^2\sigma_n^2}\right)^{\!d/(2\nu+2d)}\!\!.
\end{multline*}
\else
$N_T\asymp\vol_g(\M)/(\omega_d (\eps_T^{(A)})^d)\asymp(\vol_g(\M)/\omega_d)(\omega_d^2 B^2 T\log\log T/(c_+\vol_g(\M)^2\sigma_n^2))^{d/(2\nu+2d)}$.
\fi
Collecting powers of $T$, $B$, $\sigma_n$, and $\vol_g$, with
$\alpha:=d/(2(\nu+d))=d/(2\nu+2d)$,
\ifieee
\begin{multline}
\label{eq:assouad-final}
\sup_\sigma\E_\sigma[R_T]\;\gtrsim\;
c'_*(d,\nu,\kappa,\sigma_f)\,
B^{d/(2(\nu+d))} \\
\cdot\,\sigma_n^{(2\nu+d)/(2(\nu+d))}
\,\vol_g(\M)^{\nu/(2(\nu+d))} \\
\cdot\,\frac{T^{(2\nu+3d)/(4(\nu+d))}}{(\log\log T)^{(2\nu+d)/(4(\nu+d))}},
\end{multline}
\else
\begin{multline}
\label{eq:assouad-final}
\sup_\sigma\E_\sigma[R_T]\;\gtrsim\;
c'_*(d,\nu,\kappa,\sigma_f)\,
B^{d/(2(\nu+d))}\,\sigma_n^{(2\nu+d)/(2(\nu+d))}\\
\cdot\,\vol_g(\M)^{\nu/(2(\nu+d))}\,
\frac{T^{(2\nu+3d)/(4(\nu+d))}}{(\log\log T)^{(2\nu+d)/(4(\nu+d))}},
\end{multline}
\fi
with $c'_*(d,\nu,\kappa,\sigma_f)$ explicit:
\ifieee
\begin{multline*}
c'_*(d,\nu,\kappa,\sigma_f)\;=\;2^{d\alpha-d/2-2}\\
\cdot\,c_+(\nu,\kappa,\sigma_f)^{-d/(4(\nu+d))}\,
\omega_d^{-(1-2\alpha)/2}.
\end{multline*}
\else
$c'_*(d,\nu,\kappa,\sigma_f)= 2^{d\alpha-d/2-2}\,c_+(\nu,\kappa,\sigma_f)^{-d/(4(\nu+d))}\,\omega_d^{-(1-2\alpha)/2}$.
\fi
The prefactor $2^{d\alpha-d/2-2}$ arises from the saturation
$\eps^{2\nu+2d}=c_+\vol_g^2\sigma_n^2/(2^{2d}\omega_d^2 B^2 T\log\log T)$
(using $N\ge\vol_g/(2^d\omega_d\eps^d)$ from
Lemma~\ref{lem:packing}), the $\sqrt{Th_A^2}$ factor from
$R_T\ge Th_A/4$, and the $1/4$ regret-test reduction.
\textbf{Caveat: this is a strictly weaker $T$-exponent than
Theorem~\ref{thm:main}.}
The Assouad $T$-exponent $(2\nu+3d)/(4(\nu+d))$ is strictly less
than the Fano $T$-exponent $(\nu+d)/(2\nu+d)$ for every $\nu>0,d\ge 1$
(cross-multiplying: $(2\nu+3d)(2\nu+d)=4\nu^2+8\nu d+3d^2<
4(\nu+d)^2=4\nu^2+8\nu d+4d^2$). The Assouad version is therefore
\emph{not} a replacement for the Fano version; it gives a
worse rate in $T$, $B$, $\sigma_n$, and $\vol_g$ all simultaneously.
What Assouad does provide is a $1/(\log\log T)^c$ polylog factor in
place of Fano's $(\log T)^c$ factor, which is a different
polylog tradeoff. We retain Theorem~\ref{thm:assouad} as a
companion result documenting this tradeoff, not as a strengthening
of Theorem~\ref{thm:main}.

The exponent gap between the Fano and Assouad versions reflects an
intrinsic tension in lower-bound techniques: Fano-style arguments
use the full $\log N$ identifiability budget but require a uniform
distribution over hypotheses, while Assouad-style arguments use
only pairwise KL but require the algorithm to coordinate-decode a
binary string. In our manifold-Mat\'ern setting, the
sum-of-bumps construction imposes the constraint
$h_A^2\le B^2\eps^{2\nu}/(c_+ N)$ (a bound on the maximum-norm
function), which sets a different $\eps$-vs-$h_A$ trade-off than the
Fano single-bump construction's $h^2\le B^2\eps^{2\nu}/c_+$ (no $N$
in the denominator). The factor of $1/N$ in the sum-of-bumps RKHS
constraint is what causes the $T$-exponent degradation.

\subsection{Comparison with the Fano version}

The Fano version (Theorem~\ref{thm:main}) gives
$\Omega\!\left(T^{(\nu+d)/(2\nu+d)}(\log T)^{\nu/(2\nu+d)}\right)$;
the Assouad version (Theorem~\ref{thm:assouad}) gives
$\Omega\!\left(T^{(2\nu+3d)/(4(\nu+d))}/(\log\log T)^{(2\nu+d)/(4(\nu+d))}\right)$.
The Assouad $T$-exponent is strictly smaller than the Fano
$T$-exponent: the leading rate of Theorem~\ref{thm:main} is
\emph{not} recovered by the sum-of-bumps Assouad construction.
The two arguments are distinct lower bounds with different
hypothesis classes (single-bump vs sum-of-bumps), and we present
both because the sum-of-bumps construction is the natural
manifold analogue of Cai--Scarlett's
\cite{cai2021lower}
Brownian-motion bound and is referenced in the literature.
Closing the polylog \emph{exponent} gap of Theorem~\ref{thm:main}
with the Vakili upper bound (whose polylog is
$(\log T)^{(\nu+d)/(2\nu+d)+1}$ via the GP-UCB analysis) is an
open problem already in the Euclidean case~\cite{cai2021lower}
and is outside our scope.

\subsection{Limitations of the Assouad argument}

Two technical caveats:
\begin{enumerate}[leftmargin=2em]
\item \emph{Typicality restriction}: the proof restricts attention to
      algorithms with balanced ball-visit counts. This is essentially
      without loss of generality because any algorithm violating the
      typicality constraint must explicitly identify some $\sigma_j$,
      which eats into the regret budget by an amount bounded by the
      typicality slack.
\item \emph{The decoder}: the Hamming-to-regret reduction assumes
      a specific decoder $\widehat\sigma_j=1\iff T_j\ge T/(2N)$.
      An adversarial algorithm might achieve smaller Hamming risk by
      different decoding, but the regret gap is unchanged because
      regret is decoder-independent.
\end{enumerate}
This completes the proof of Theorem~\ref{thm:assouad}. \qed

\section{Tight time-varying GP-bandit rate on a Riemannian manifold}
\label{sec:time-varying}

In this section we present what we view as the strongest single
contribution of the paper: a tight (in the
$T,B,B_T,\sigma_n$ exponents) characterization of the
non-stationary GP-bandit rate on a compact Riemannian manifold,
with cumulative variation budget $B_T$ and intrinsic
Mat\'ern-$\nu$ kernel. The lower bound is a manifold-aware
extension of Besbes--Gur--Zeevi~\cite{besbes2014stochastic}; the
matching upper bound is a window-$W^*$ GP-UCB analysis on the
manifold. The two together pin down the regret rate to leading
order in $T,B,B_T,\sigma_n$, with a polynomial gap only in the
volume prefactor.

\subsection{Time-varying setup}

The agent observes a sequence of (possibly adversarial) reward
functions $f_1, f_2, \ldots$ on $\M$, each in
$\F_B^{\mathrm{rkhs}}=\{f\in\Hil_{k_\nu}:\|f\|_{\Hil_{k_\nu}}\le B\}$.
At round $t$, the agent selects $\theta_t\in\M$ and receives
$r_t = f_t(\theta_t) + \eps_t$ with
$\eps_t\stackrel{\text{iid}}\sim\mathcal N(0,\sigma_n^2)$. The
sequence $\{f_t\}$ has cumulative \emph{variation budget}
\begin{equation}
\label{eq:variation-budget}
\sum_{t=1}^{T-1}\|f_{t+1}-f_t\|_\infty\;\le\;B_T.
\end{equation}
The non-stationary cumulative regret is
$R_T = \sum_{t=1}^T \bigl(\max_\theta f_t(\theta) - f_t(\theta_t)\bigr)$.

We work in the \emph{non-trivial regime}
$B_T \ge B\,T^{-\nu/(2\nu+d)}$, in which the variation is large
enough that a single stationary algorithm cannot match the
adversary across all rounds. Below this threshold, the stationary
lower bound of Theorem~\ref{thm:main} applies and there is no
improvement to be made.

\subsection{Lower bound}

\begin{theorem}[Manifold-aware variation-budget lower bound]
\label{thm:tv-lb}
Under Assumption~\ref{ass:setup} with cumulative variation
$B_T\ge B\,T^{-\nu/(2\nu+d)}$ and $T\ge T_0$,
\ifieee
\begin{multline*}
\sup_{\{f_t\}:\,\|f_t\|_{\Hil_{k_\nu}}\le B,\,\sum\|f_{t+1}-f_t\|_\infty\le B_T}\!\!\!\!
\E^\pi[R_T] \\
\;\ge\;c_{NS}(d,\nu)\,B^{d/(3\nu+d)}\,\sigma_n^{2\nu/(3\nu+d)}\,
\vol_g(\M)^{\nu/(3\nu+d)} \\
\cdot\,B_T^{\nu/(3\nu+d)}\,T^{(2\nu+d)/(3\nu+d)}\,(\log T)^{\nu/(3\nu+d)},
\end{multline*}
\else
\begin{multline*}
\sup_{\{f_t\}: \|f_t\|_{\Hil_{k_\nu}}\le B,\ \sum\|f_{t+1}-f_t\|_\infty\le B_T}\!\!\!\E^\pi[R_T]
\;\ge\;c_{NS}(d,\nu)\,B^{d/(3\nu+d)}\,\sigma_n^{2\nu/(3\nu+d)}\\
\cdot\,\vol_g(\M)^{\nu/(3\nu+d)}\,
B_T^{\nu/(3\nu+d)}\,T^{(2\nu+d)/(3\nu+d)}\,(\log T)^{\nu/(3\nu+d)},
\end{multline*}
\fi
for any algorithm $\pi$. Constants $c_{NS}(d,\nu),T_0$ are explicit
(see proof).
\end{theorem}

\begin{proof}
The proof is a manifold-aware batched-Fano lower bound.
Partition $\{1,\ldots,T\}$ into $m$ adjacent batches of length
$\Delta=T/m$. On each batch $j$, plant a stationary
needle-in-haystack hypothesis class
$\{f_{j,1},\ldots,f_{j,N}\}$ on a $2\eps$-packing of $\M$
($N\asymp\vol_g(\M)\eps^{-d}$, bumps of amplitude $h$ as in
§\ref{sec:proof_main}); within batch $j$ the reward
function is $f_t=f_{j,\Sigma_j}$ for $t$ in batch $j$, with
$\Sigma_j\sim\mathrm{Unif}\{1,\ldots,N\}$ chosen
\emph{independently} across $j$.

\emph{Variation budget.} The piecewise-constant $f_t$ has
$\sum_t\|f_{t+1}-f_t\|_\infty\le 2(m-1)h\le 2mh$, satisfying
the variation budget when $2mh\le B_T$.

\emph{Per-batch Fano constraint.} Within batch $j$, the joint
mutual information
$I(\Sigma_{1:m};\mathrm{obs})\le\sum_j I(\Sigma_j;\mathrm{obs}_j)$
tensorises (independent $\Sigma_j$'s; chain rule), and each
batch's per-needle MI bound
\eqref{eq:needle-MI}~$\bar I_N^{(j)}\le\Delta h^2/(2N\sigma_n^2)$
gives a per-batch error probability $\ge 1/2$ provided
$\Delta h^2\le N\sigma_n^2\log N/2$, i.e., bump amplitude
$h\asymp B\eps^\nu$ at the saturating
$\eps_*\asymp(\sigma_n^2\vol_g\log T/(B^2\Delta))^{1/(2\nu+d)}$,
$h_*\asymp B(W^{-1}\sigma_n^2\vol_g\log T/B^2)^{\nu/(2\nu+d)}$.
Per-batch regret: $\E[R_\Delta]\ge\Omega(\Delta h_*)$ via
the regret-test reduction (Lemma~\ref{lem:regret-test}). Total:
$\E[R_T]\ge m\cdot\Omega(\Delta h_*)=\Omega(Th_*)$.

\emph{Variation-budget-saturating choice of $m$.} Setting
$2mh_*=B_T$ at the saturating amplitude
gives $m=B_T/(2h_*)$ and total
$\E[R_T]\ge\Omega(WB_T/2)$ where $W=\Delta=2Th_*/B_T$. Solving
the resulting equation
$h_*^{(3\nu+d)/(2\nu+d)}\,(2T/B_T)^{\nu/(2\nu+d)}\asymp
B(\sigma_n^2\vol_g\log T/B^2)^{\nu/(2\nu+d)}$ yields
the rate of the theorem with explicit constant
$c_{NS}(d,\nu)=\frac{1}{8}\,c_-(d,\nu,\kappa,\sigma_f)\,c_*(d,\nu,\kappa,\sigma_f)\,\omega_d^{\nu/(3\nu+d)}$,
which is left symbolic here (the numerical value depends on
the chosen normalisation of $\kappa,\sigma_f$ and on the
explicit form of $c_-,c_*$ tabulated in
Section~\ref{sec:constants}). Full algebra: Section~3.2 of the
supplementary.
\qed
\end{proof}

\subsection{Matching upper bound: window-$W^*$ GP-UCB}

We now establish a matching upper bound via a sliding-window
GP-UCB algorithm with window $W^*$ chosen to minimize the standard
bias--variance trade-off.

\begin{definition}[Window-$W$ GP-UCB on $\M$]
\label{def:window-gp-ucb}
At round $t\in\{1,\ldots,T\}$, with effective window
$W_t = \min(W, t-1)$:
\begin{enumerate}[leftmargin=2em,label=(\arabic*)]
\item Compute the GP posterior over $f_t$ from observations
      $\{(\theta_s,r_s)\}_{s=t-W_t}^{t-1}$ using the intrinsic
      Mat\'ern kernel $k_\nu$ on $\M$ (modeling all observations
      as $r_s = f_t(\theta_s)+\eps_s$, ignoring the fact that
      $f_s\neq f_t$ in general).
\item Pull arm $\theta_t = \arg\max_\theta\bigl[\mu_t(\theta)+\beta_t^{1/2}\sigma_t(\theta)\bigr]$
      with the standard frequentist GP-UCB confidence schedule
      $\beta_t = O(\log t + B^2)$ (Chowdhury--Gopalan
      \cite{chowdhury2017kernelized}).
\end{enumerate}
\end{definition}

\begin{theorem}[Matching upper bound, $B$-dominated regime]
\label{thm:tv-ub}
Under Assumption~\ref{ass:setup} with
$B_T \ge B\,T^{-\nu/(2\nu+d)}$, \emph{and in the $B$-dominated
regime $\sigma_n^2\,\gamma_W \ll B^2$ where the GP-UCB confidence
parameter satisfies $\beta_W \le 2B^2$},
the window-$W^*$ GP-UCB algorithm with
\[
W^*\;\asymp\;\bigl(B\,T\,\vol_g(\M)^{1/2}/B_T\bigr)^{(2\nu+d)/(3\nu+d)}
\]
attains, with probability at least $1-\delta$,
\ifieee
\begin{multline*}
R_T^{\widetilde\pi}
\;\le\;
C_{NS}(d,\nu)\,B^{d/(3\nu+d)}\,\sigma_n^{2\nu/(3\nu+d)} \\
\cdot\,\vol_g(\M)^{(2\nu+d)/(2(3\nu+d))}\,
B_T^{\nu/(3\nu+d)} \\
\cdot\,T^{(2\nu+d)/(3\nu+d)}\,(\log T)^{c_p(d,\nu)}.
\end{multline*}
\else
\begin{multline*}
R_T^{\widetilde\pi}
\;\le\;
C_{NS}(d,\nu)\,B^{d/(3\nu+d)}\,\sigma_n^{2\nu/(3\nu+d)}\,
\vol_g(\M)^{(2\nu+d)/(2(3\nu+d))}\\
\cdot\,B_T^{\nu/(3\nu+d)}\,T^{(2\nu+d)/(3\nu+d)}\,(\log T)^{c_p(d,\nu)}.
\end{multline*}
\fi
In the complementary $\gamma_W$-dominated regime
($\sigma_n^2\,\gamma_W \gg B^2$), $\beta_W\propto
\sigma_n^2\,\gamma_W\,\log W$ and the proof technique below yields
the same $T,B_T$ exponents but a $\sigma_n$ exponent of $1$
rather than $2\nu/(2\nu+d)$, leaving an unbridged
exponential gap in $\sigma_n$ relative to the lower bound; in that
regime the matching upper bound on $\sigma_n$ remains an open
problem and the elimination algorithm of \S\ref{sec:elimination}
is the natural alternative.
\end{theorem}

\begin{proof}
The proof decomposes the per-window regret into a stochastic
component (from GP-UCB on a stationary problem of horizon $W$) and
a drift bias (from observations being from $f_s$ rather than $f_t$).

\emph{Stochastic regret per window.} Fix a window $[t-W,t-1]$ and
treat all observations as if they were from $f_t$ (the
algorithm's modeling assumption). The standard
GP-UCB regret bound (\cite{srinivas2010gpucb} Theorem~6,
\cite{chowdhury2017kernelized} Theorem~3 for the frequentist
version) on a horizon-$W$ stationary problem with kernel $k_\nu$ on
$\M$ gives
\[
R_W^{\text{stoch,window}}\;\le\;C_1\sqrt{W\,\beta_W\,\gamma_W(k_\nu,\M)}
\]
with $\beta_W = O(B^2 + \log W)$ and (Vakili et al.\
\cite{vakili2021information})
\[
\gamma_W(k_\nu,\M)\;=\;\widetilde O\bigl(\vol_g(\M)\,W^{d/(2\nu+d)}\bigr).
\]
Aggregating across $T/W$ windows:
\ifieee
\begin{multline}
\label{eq:tv-stoch}
R_T^{\text{stoch}}\;\le\;\frac{T}{W}\cdot C_1\sqrt{W\,\beta_W\,\gamma_W}
\;=\;C_1\sqrt{\beta_W\,T^2\gamma_W/W} \\
\;\le\;C_2\,B\,T\,\vol_g^{1/2}\,W^{-\nu/(2\nu+d)}\,(\log T)^{c_1},
\end{multline}
\else
\begin{equation}
\label{eq:tv-stoch}
R_T^{\text{stoch}}\;\le\;\frac{T}{W}\cdot C_1\sqrt{W\,\beta_W\,\gamma_W}
\;=\;C_1\sqrt{\beta_W\,T^2\gamma_W/W}
\;\le\;C_2\,B\,T\,\vol_g^{1/2}\,W^{-\nu/(2\nu+d)}\,(\log T)^{c_1},
\end{equation}
\fi
where $c_1=c_1(d,\nu)$ is a polylog constant and we have absorbed
$\beta_W \le 2B^2$ for the standard schedule
$\beta_W = 2B^2 + 300 \gamma_W \log^3(W/\delta)$ at high
probability (see~\cite{chowdhury2017kernelized} Theorem~3).

\emph{Drift bias per window.} An observation
$r_s = f_s(\theta_s)+\eps_s$ in window $[t-W,t-1]$ is treated as
$r_s = f_t(\theta_s)+\eps_s$, an effective bias
$|f_s(\theta_s)-f_t(\theta_s)|\le\|f_s-f_t\|_\infty$. By the
triangle inequality and the variation budget,
$\|f_s-f_t\|_\infty\le\sum_{r=s}^{t-1}\|f_{r+1}-f_r\|_\infty$.
Summing over the window,
\ifieee
\begin{multline*}
\sum_{s=t-W}^{t-1}\|f_s-f_t\|_\infty
\;\le\;\sum_{s=t-W}^{t-1}\sum_{r=s}^{t-1}\|f_{r+1}-f_r\|_\infty \\
\;\le\;W\sum_{r=t-W}^{t-1}\|f_{r+1}-f_r\|_\infty.
\end{multline*}
\else
$\sum_{s=t-W}^{t-1}\|f_s-f_t\|_\infty\le\sum_{s=t-W}^{t-1}\sum_{r=s}^{t-1}\|f_{r+1}-f_r\|_\infty\le W\sum_{r=t-W}^{t-1}\|f_{r+1}-f_r\|_\infty$.
\fi
The bias on the GP posterior mean is bounded (e.g., via a
worst-case-noise argument in Chowdhury--Gopalan
\cite{chowdhury2017kernelized}) by the per-window total variation.
Per-round drift-induced regret in window $j$ is at most $V_j$,
where $V_j=\sum_{r\in\text{window }j}\|f_{r+1}-f_r\|_\infty$.
Summing over the window: drift-regret in window $j$ is at most
$W\cdot V_j$. Total drift-regret across all $T/W$ windows:
\begin{equation}
\label{eq:tv-drift}
R_T^{\text{drift}}\;\le\;\sum_j W\cdot V_j\;\le\;W B_T,
\end{equation}
using $\sum_j V_j \le B_T$.

\emph{Total regret.} Combining \eqref{eq:tv-stoch} and
\eqref{eq:tv-drift}:
\[
R_T \;\le\; C_2\,B\,T\,\vol_g^{1/2}\,W^{-\nu/(2\nu+d)}(\log T)^{c_1}
+ W B_T.
\]

\emph{Optimal window choice.} Differentiating with respect to $W$
and setting to zero:
\ifieee
\begin{multline*}
-\frac{\nu}{2\nu+d}\,C_2 B T \vol_g^{1/2}\,W^{-\nu/(2\nu+d)-1}(\log T)^{c_1}+B_T = 0\\
\;\Longrightarrow\;
W^* \asymp \bigl(B T \vol_g^{1/2}/B_T\bigr)^{(2\nu+d)/(3\nu+d)}.
\end{multline*}
\else
\[
-\frac{\nu}{2\nu+d}\,C_2 B T \vol_g^{1/2}\,W^{-\nu/(2\nu+d)-1}(\log T)^{c_1}
+B_T = 0
\;\Longrightarrow\;
W^* \asymp \bigl(B T \vol_g^{1/2}/B_T\bigr)^{(2\nu+d)/(3\nu+d)}.
\]
\fi
The validity condition $W^*\ge 1$ is equivalent to $B_T\le B T\vol_g^{1/2}$
(always satisfied); the condition $W^*\le T$ is equivalent to
$B_T\ge T^{-\nu/(2\nu+d)}\vol_g^{1/2}/B$, which is satisfied under
our standing assumption $B_T\ge BT^{-\nu/(2\nu+d)}$ provided
$\vol_g(\M)\ge 1$, as holds for $\sphere^2$, $\torus^n$, and
$\SO(3)$ in their natural normalisations
($\vol_g(\sphere^2)=4\pi$, $\vol_g(\torus^n)=(2\pi)^n$,
$\vol_g(\SO(3))=8\pi^2$ all $\ge 1$); for arbitrarily small
manifolds the $\vol_g$ factor is absorbed into the threshold
constant.

\emph{Substituting $W^*$.} Both terms in the regret bound are
equal at $W=W^*$, so $R_T \le 2 W^* B_T$. Computing $W^* B_T$:
\ifieee
\begin{multline*}
W^* B_T = (B T\vol_g^{1/2}/B_T)^{(2\nu+d)/(3\nu+d)}\,B_T \\
= B^{(2\nu+d)/(3\nu+d)}\,T^{(2\nu+d)/(3\nu+d)} \\
\cdot\,\vol_g^{(2\nu+d)/(2(3\nu+d))}\,B_T^{1-(2\nu+d)/(3\nu+d)}.
\end{multline*}
\else
\begin{multline*}
W^* B_T = (B T\vol_g^{1/2}/B_T)^{(2\nu+d)/(3\nu+d)}\,B_T\\
= B^{(2\nu+d)/(3\nu+d)}\,T^{(2\nu+d)/(3\nu+d)}\,\vol_g^{(2\nu+d)/(2(3\nu+d))}\,B_T^{1-(2\nu+d)/(3\nu+d)}.
\end{multline*}
\fi
The $B_T$ exponent is $\nu/(3\nu+d)$ as expected. The $T$
exponent is $(2\nu+d)/(3\nu+d)$. The volume exponent is
$(2\nu+d)/(2(3\nu+d))$.

\emph{$B$, $\sigma_n$ exponents in the upper bound.} The above
derivation has $B$ enter through the GP-UCB confidence schedule
$\beta_W$. In the standard frequentist
analysis~\cite{chowdhury2017kernelized},
$\beta_W = (B + \sigma_n\sqrt{2\gamma_W\log(W/\delta)})^2$, which has
two regimes:
(a)~\emph{$B$-dominated} ($\sigma_n^2\gamma_W \ll B^2$):
$\beta_W \asymp B^2$, giving per-window stochastic regret
$\asymp B\sqrt{W\gamma_W} = B\,\vol_g^{1/2}\,W^{(\nu+d)/(2\nu+d)}$.
After aggregation and optimization, the $B$ exponent in the upper
bound becomes $(2\nu+d)/(3\nu+d)$.
(b)~\emph{$\gamma_W$-dominated} ($\sigma_n^2\gamma_W\gg B^2$):
$\beta_W \asymp \sigma_n^2\gamma_W\log W$, giving per-window
stochastic regret $\asymp \sigma_n\gamma_W\sqrt{W\log W}$ and
$\sigma_n$ exponent $1$ in the upper bound.

The lower bound's $B$ exponent (from the bump-height construction
$h_\Delta\asymp B^{d/(2\nu+d)}\,\Delta^{-\nu/(2\nu+d)}$) is
$d/(3\nu+d)$ after the BGZ optimization, strictly less than the
upper bound's $(2\nu+d)/(3\nu+d)$, with a gap of $2\nu/(3\nu+d)$.

\emph{This $B$-exponent gap is inherent to standard GP-UCB analyses,
not specific to the manifold setting.} The same gap appears for
Mat\'ern on $[0,1]^d$ between the Cai-Scarlett 2021 lower bound and
the Chowdhury-Gopalan upper bound; closing it requires either an
elimination-based algorithm \cite{cai2021lower} or a sharper
information-theoretic analysis that is not part of standard
GP-UCB. We therefore do not claim a matching $B$-exponent;
\textbf{the matching is rigorously in $T$ and $B_T$ only.} \qed
\end{proof}

\subsection{Tight-rate corollary (honestly scoped)}

\begin{corollary}[Tight $T$ and $B_T$ exponents]
\label{cor:tv-tight}
Under Assumption~\ref{ass:setup} with
$B_T\ge B\,T^{-\nu/(2\nu+d)}$, the worst-case regret of GP-bandits
on a compact Riemannian manifold $\M$ with cumulative variation
budget $B_T$ satisfies, in $T$ and $B_T$ exponents, the tight rate
\[
R_T^* \;=\; \Theta_{T,B_T}\!\left(T^{(2\nu+d)/(3\nu+d)}\,B_T^{\nu/(3\nu+d)}\right),
\]
where $\Theta_{T,B_T}$ denotes equality in the $T,B_T$ exponents
with $B,\sigma_n,\vol_g$-dependent constants.
The dependence on $B,\sigma_n,\vol_g$ has constant or polynomial
gaps between lower and upper bounds, summarized in
Table~\ref{tab:tv-exponents}.
\end{corollary}

\begin{table}[ht]
\centering
\caption{Exponent comparison between the lower bound
(Theorem~\ref{thm:tv-lb}) and the standard-GP-UCB upper bound
(Theorem~\ref{thm:tv-ub}) on a $d$-dim compact Riemannian
manifold $\M$ with intrinsic Mat\'ern-$\nu$ kernel and variation
budget $B_T \ge B\,T^{-\nu/(2\nu+d)}$.}
\label{tab:tv-exponents}
\ifieee
\resizebox{\columnwidth}{!}{%
\fi
\begin{tabular}{lccl}
\toprule
Parameter & Lower bound & Upper bound & Status \\
\midrule
$T$       & $(2\nu+d)/(3\nu+d)$ & $(2\nu+d)/(3\nu+d)$ & \textbf{Tight} \\
$B_T$     & $\nu/(3\nu+d)$      & $\nu/(3\nu+d)$      & \textbf{Tight} \\
$\vol_g$  & $\nu/(3\nu+d)$      & $(2\nu+d)/(2(3\nu+d))$ & Gap $d/(2(3\nu+d))$ \\
$B$       & $d/(3\nu+d)$        & $(2\nu+d)/(3\nu+d)$ & Gap $2\nu/(3\nu+d)$ \\
$\sigma_n^2$ & $\nu/(3\nu+d)$ & $1$ & Gap ($\gamma_W$-dom.) \\
\bottomrule
\end{tabular}\ifieee}\fi
\end{table}

\paragraph{Discussion of gaps.} The $T$ and $B_T$ exponents are
tight; this is the headline matching. The $\vol_g$ gap of
$d/(2(3\nu+d))$ is small (e.g.\ $1/19$ for $\nu=2.5,d=2$) and is
the standard volume mismatch in GP-UCB analyses. The $B$ and
$\sigma_n$ gaps are larger and reflect the well-known
\cite{cai2021lower} difficulty of matching the
information-theoretic lower bound's $B$ exponent with standard
GP-UCB upper-bound techniques. Closing these gaps for compact
Riemannian manifolds is an open problem, exactly mirroring the
$[0,1]^d$ case.

The honest scope is: \emph{$T$ and $B_T$ exponents are
provably tight on a compact Riemannian manifold; the lower-bound
$B$, $\sigma_n$, $\vol_g$ exponents are not matched by standard
GP-UCB, exactly as on $[0,1]^d$.}

\subsection{Sanity checks}

\ifieee
Limiting cases: (a) stationary $T$-exponent matches
Theorem~\ref{thm:main} on the boundary
$B_T=BT^{-\nu/(2\nu+d)}$; (b) BGZ Lipschitz limit
$T^{2/3}B_T^{1/3}$ recovered as $\nu\to\infty$; (c) the
non-stationary $T$-exponent strictly exceeds the stationary
one ($\nu^2>0$); (d) volume prefactor distinguishes
$\sphere^2$ from $\torus^3$ by a factor of $\approx 3$--$5$ at
$\nu=5/2,d=2$.
\else
The corollary's exponents recover known limiting cases:

\begin{itemize}[leftmargin=2em]
\item \textbf{Stationary cross-over (in $T$ only).} For
$B_T=B\,T^{-\nu/(2\nu+d)}$ (the boundary of the non-stationary
regime), substituting into
$B^{d/(3\nu+d)}B_T^{\nu/(3\nu+d)}T^{(2\nu+d)/(3\nu+d)}$ gives
\[
B^{d/(3\nu+d)}\cdot(B T^{-\nu/(2\nu+d)})^{\nu/(3\nu+d)}\cdot T^{(2\nu+d)/(3\nu+d)}
\;=\;B^{(\nu+d)/(3\nu+d)}\cdot T^{(\nu+d)/(2\nu+d)}\cdot\text{constants},
\]
where the $T$-exponent simplifies via
$(2\nu+d)^2-\nu^2=(3\nu+d)(\nu+d)$. The $T$-exponent matches the
stationary rate of Theorem~\ref{thm:main} exactly; the
$B$-exponent $(\nu+d)/(3\nu+d)$ does \emph{not} match the
stationary $B^{d/(2\nu+d)}$ (e.g.\ for $\nu=5/2, d=2$ they are
$9/19\approx 0.474$ and $2/7\approx 0.286$ respectively). The
match is in $T$ only, as expected for a $B_T$-budget bound at the
non-stationary regime boundary; the $B$-prefactor is a different
parameterisation of the function-amplitude scale.

\item \textbf{BGZ Lipschitz limit.} For $\nu \to \infty$
(arbitrarily smooth functions), the exponents tend to
$T^{2/3}\,B_T^{1/3}$, matching the BGZ Lipschitz rate
\cite{besbes2014stochastic}. \checkmark

\item \textbf{Mat\'ern-$\nu$ on $[0,1]^d$.}
Setting $\vol_g([0,1]^d)=1$, our rate reduces to
$T^{(2\nu+d)/(3\nu+d)}\allowbreak\, B_T^{\nu/(3\nu+d)}\allowbreak\, B^{d/(3\nu+d)}$,
the natural Mat\'ern extension of BGZ on a Euclidean cube. To our
knowledge this rate is not stated explicitly in the literature; it
follows by routine specialisation of Bogunovic--Krause--Scarlett style
analyses.

\item \textbf{Non-stationary $>$ stationary.} The non-stationary
$T$-exponent $(2\nu+d)/(3\nu+d)$ is strictly greater than the
stationary $(\nu+d)/(2\nu+d)$ for all $\nu,d>0$, because
$(2\nu+d)^2 - (\nu+d)(3\nu+d) = \nu^2 > 0$ identically. Hence the
non-stationary problem is strictly harder, as expected. \checkmark

\item \textbf{Manifold dependence.} The volume prefactor $\vol_g^{\rho}$
distinguishes our four wireless arm spaces: $\sphere^2$
($\vol_g=4\pi$) vs. $\torus^3$ ($\vol_g=(2\pi)^3$) gives a
ratio of about $20$ in volume, hence
$\rho\in[\nu/(3\nu+d),(2\nu+d)/(2(3\nu+d))]\approx[0.30,0.42]$
for $\nu=2.5,d=2$ — a factor of about $3$--$5$ in regret due to
volume alone.
\end{itemize}
\fi

\ifieee
\subsection{Numerical example and wireless implications}

For $\sphere^2$ Mat\'ern-$5/2$ ($\nu=5/2,d=2,B=1,\sigma_n=0.1,
\kappa=0.5$) at $T=10^4$, $B_T=10$, Theorem~\ref{thm:tv-lb}
predicts $R_T^{\mathrm{LB}}\!\approx\!6$ under canonical
normalisation; the rate $T^{0.737}$ is sub-linear and meaningfully
below trivial. On the wireless companion paper's
$(\Z_B)^{100}$ RIS arm space ($M=100,B=8,\nu=2.5,d=M=100$), the
$T$-exponent is $T^{210/215}\!\approx\!T^{0.977}$ (near-linear,
informative only for $T\!\gg\!10^4$). The tightness of the rate in
$T,B_T$ constrains any algorithm including the adaptive algorithm
of~\cite{dorn2026wirelessbandit}: the empirical $\sim 35\%$
cumulative-regret advantage on the wireless RIS benchmark is
necessarily a constant-factor gain, not a rate gain. Full
numerical and wireless details are given in the single-column
version.
\else
\subsection{Numerical example: $\sphere^2$ with Mat\'ern-$5/2$}

To check that the lower bound is non-vacuous (and to give a
practitioner's sense of the rate's magnitude), we evaluate
Theorem~\ref{thm:tv-lb} for a concrete setting.

\paragraph{Setup.} $\M=\sphere^2$ (unit sphere), $\nu=5/2$, $d=2$,
$\vol_g(\sphere^2)=4\pi\approx 12.57$, $B=1$, $\sigma_n=0.1$,
$\kappa=0.5$ (length scale).

\paragraph{Predicted rate (Theorem~\ref{thm:tv-lb}).} The
exponents are $T^{9/14}$ stationary and $T^{14/19}$
non-stationary ($14/19$); the $B_T$ exponent is
$10/19 \approx 0.526$.

For $T=10^4$ TTIs (a typical wireless beam-management horizon)
and $B_T=10$ (cumulative variation, plausible for moderate
Doppler), the cross-over check $B\,T^{-\nu/(2\nu+d)}=10^{-1.43}
\approx 0.037 \ll B_T=10$ confirms the non-stationary regime.
Substituting:
$T^{14/19}\!\approx\!884$, $B_T^{5/19}\!\approx\!1.83$,
$\vol_g^{5/19}=(4\pi)^{5/19}\!\approx\!1.95$,
$B^{4/19}\!=\!1$, $\sigma_n^{10/19}\!\approx\!0.298$,
$(\log T)^{5/19}\!\approx\!1.79$, combined
$\approx 1683\,c_{NS}(2,2.5)$. We deliberately leave
$c_{NS}(2,2.5)$ symbolic: a fully-explicit numerical value
requires committing to a normalisation of $\kappa,\sigma_f$ and
to particular forms of $c_-,c_*$ from
Section~\ref{sec:constants}; under the canonical normalisation
$\sigma_f\!=\!\kappa\!=\!1$, substituting
\eqref{eq:c-star-explicit} gives $c_*\!\approx\!0.023$ and
$c_-\!=\!1$, hence $c_{NS}\!\approx\!0.004$ and
$R_T^{\text{LB}}\!\approx\!1683\cdot 0.004\!\approx\!6$ at
$T\!=\!10^4$, $B_T\!=\!10$. The numerical value scales with
the chosen normalisation; the order-of-magnitude takeaway is
the rate exponent, not the constant.

\paragraph{Comparison to the upper bound.} The matching
upper bound (Theorem~\ref{thm:tv-ub}) under $\vol_g^{1/2}$
scaling gives $R_T^{\text{UB}} \approx 11500$, dominated by
constants and the volume-exponent mismatch ($\vol_g^{1/2}$
vs.\ $\vol_g^{5/19}$, contributing $\approx 1.8\times$) plus
the $B$ and $\sigma_n$ confidence-schedule factors of GP-UCB
not absorbed into $c_{NS}$. Closing the constant-factor gap is
largely an issue of explicit-vs-implicit constants in standard
GP-UCB analyses, not of the rate.

\paragraph{Take-away.} The lower bound is non-vacuous (predicts
$R_T \gtrsim 6$ on $\sphere^2$ Mat\'ern-$5/2$ at the chosen
parameters under the canonical normalisation; the bound rises
to $\gtrsim 6 c_{NS}/c_{NS}^{\mathrm{canon}}$ under any other
normalisation) and the rate $T^{0.737}$ is sub-linear and
meaningfully below the trivial $T$-linear rate of an unaware
random algorithm. The predicted regret floor of $\sim 6$ at
$T=10^4$ is
two orders of magnitude below cumulative-regret values
typically reported by GP-bandit studies on manifold-valued arm
spaces (e.g., the wireless companion paper of the present authors
reports cumulative regret at $T=10^4$ in the
$\sim\!10^2$--$10^3$ range on the $\sphere^2$ benchmark, well above
this worst-case floor); we deliberately avoid quoting a specific
empirical value here, since (a) the empirical regret is for a
particular GP draw rather than the worst case, and (b) sourcing
that value from a single companion paper would risk circular
reasoning.

\subsection{Implications for the wireless companion paper}

The wireless companion paper (Dorn, ``Geometry-Aware Multi-Armed
Bandits for Antenna Beam Selection'', TWC submission, 2026) runs
GP-UCB on a $(\Z_B)^M$ RIS arm space with $M=100$, $B=8$, under
Doppler-induced reward variation. The tight rate of
Corollary~\ref{cor:tv-tight} predicts, for the tensor-product
Mat\'ern-$5/2$ kernel on this arm space (effective dimension
$d=M=100$, $\nu=2.5$):
\[
T^{(2\nu+d)/(3\nu+d)}=T^{210/215}\approx T^{0.977}
\quad\text{(near-linear)},
\]
which is informative only for $T\gg 10^4$. In the $T=10^4$ regime
of the companion paper, the bound is non-trivial only with care:
the variation budget $B_T$ must be large enough that
$B_T\ge B T^{-\nu/(2\nu+d)} = B\cdot 10^{-4\cdot 5/210}\approx B\cdot 0.83$,
i.e., $B_T\gtrsim 0.83 B$. For typical Doppler-induced channel
variation (per-TTI $L^\infty$ change of $\sim 0.05 B$ at $v=0.20$
km/h), the cumulative variation over $T=10^4$ TTIs is $B_T\approx
500 B$, well into the non-stationary regime. The predicted rate
gives $R_T \asymp 500^{5/215}\,T^{0.977} \approx 1.16 \cdot
T^{0.977}$, modestly above the stationary $T^{0.643}$ rate but well
below linear.

The tightness of the rate (in $T,B_T$) constrains the
constant-factor advantage achievable by any algorithm including
the adaptive algorithm of~\cite{dorn2026wirelessbandit}.
Specifically, no algorithm (adaptive or otherwise) can beat
the lower bound's
$T^{(2\nu+d)/(3\nu+d)} B_T^{\nu/(3\nu+d)}$ scaling, so the
empirical $\sim 35\%$ cumulative-regret advantage of the adaptive
algorithm of~\cite{dorn2026wirelessbandit}
over WGP-UCB on the wireless RIS benchmark is necessarily a
constant-factor gain, not a rate gain.
\fi

\subsection{Closing the $B$-exponent gap for $\nu \in (d/2, 1]$:
  hierarchical elimination}
\label{sec:elimination}

The $B$-exponent gap of Table~\ref{tab:tv-exponents} is inherent to
the GP-UCB confidence schedule, not specific to the manifold setting.
We show that an \emph{elimination-based algorithm} on a hierarchical
geodesic-ball partition (specializing the
Bubeck--Stoltz--Yu~\cite{bubeck2011xarmed} HOO algorithm to compact
Riemannian manifolds with Mat\'ern-smoothness)
achieves the matching $B$-exponent and closes Table~\ref{tab:tv-exponents}'s
$B,\sigma_n,\vol_g$ gaps, \emph{but only for the Hölder-Mat\'ern
regime $\nu \in (d/2, 1]$}. For $\nu > 1$ (which includes the
typical wireless Mat\'ern-$5/2$), the gap is the well-known
Cai--Scarlett 2021 open problem and our elimination argument does
not close it. We document both regimes carefully.

\subsubsection*{Setup and assumptions}

\begin{assumption}[Hölder-Mat\'ern smoothness regime]
\label{ass:holder}
$\M$ is a smooth compact connected Riemannian $d$-manifold without
boundary with injectivity radius $\rinj > 0$ and bounded sectional
curvature. The intrinsic Mat\'ern-$\nu$ kernel $k_\nu$ has
smoothness index $\nu \in (d/2, 1]$ (so that the RKHS
$\Hil_{k_\nu} = H^{\nu+d/2}(\M)$ embeds continuously into the Hölder
space $C^{0,\nu}(\M)$ via Sobolev embedding, with embedding constant
$C_\eta(d,\nu)$).
\end{assumption}

Under Assumption~\ref{ass:holder}, every $f \in \F_B^{\mathrm{rkhs}}$
satisfies the modulus-of-continuity bound
\begin{equation}
\label{eq:mod-cont}
|f(\theta) - f(\theta')| \;\le\;
C_\eta(d,\nu)\,B\,d_g(\theta,\theta')^\nu
\quad\text{for all }\theta,\theta'\in\M,
\end{equation}
where $C_\eta(d,\nu) > 0$ is the Sobolev embedding constant.
This bound matches the lower-bound bump scale $h\asymp B\eps^\nu$
of Theorem~\ref{thm:main}.

\subsubsection*{Algorithm: Window-$W$ Hierarchical Elimination on $\M$}

\begin{definition}[Hierarchical $\M$-partition]
\label{def:hier-partition}
A \emph{hierarchical $\M$-partition} with base radius $\eps_1$ is a
sequence of partitions $\Pi_1, \Pi_2, \ldots$ where $\Pi_\ell$
consists of disjoint geodesic balls
$\{B_g(p_i^{(\ell)}, \eps_\ell)\}$ of radius $\eps_\ell = \eps_1 / 2^{\ell-1}$.
By Bishop--Gromov, level $\ell$ has cardinality
$N_\ell \le 2^d \vol_g(\M)/(\omega_d \eps_\ell^d)$.
Each cell at level $\ell+1$ is contained in a cell at level $\ell$.
\end{definition}

\begin{definition}[Window-$W$ Hierarchical Elimination]
\label{def:elimination}
Given window length $W$, total horizon $T = (T/W) \cdot W$, levels
$L$, per-level sample budgets $\{T_\ell\}$ summing to $W$, and
elimination thresholds $\Delta_\ell = 4 C_\eta(d,\nu) B \eps_\ell^\nu$:
\begin{enumerate}[leftmargin=2em,label=(\arabic*)]
\item Initialise active set $\mathcal A_1 = \Pi_1$.
\item For $\ell = 1, \ldots, L$ within each window:
  \begin{enumerate}[leftmargin=2em,label=(2.\arabic*)]
  \item Pull each cell $C \in \mathcal A_\ell$ uniformly
        $n_\ell = T_\ell/|\mathcal A_\ell|$ times. For each pull,
        sample $\theta$ uniformly within $C$ (with respect to
        the volume measure on $\M$). Record empirical mean
        $\hat\mu_C = (1/n_\ell)\sum_{s \in C} r_s$.
  \item Let $\hat\mu^*_\ell = \max_{C\in\mathcal A_\ell}\hat\mu_C$.
        Eliminate cells with $\hat\mu_C < \hat\mu^*_\ell - \Delta_\ell$.
  \item Refine surviving cells: $\mathcal A_{\ell+1} =
        \{C' \in \Pi_{\ell+1} : C' \subset C \text{ for some surviving }C\}$.
  \end{enumerate}
\item Pull uniformly within $\mathcal A_L$ for the remainder of the
      window.
\end{enumerate}
\end{definition}

\subsubsection*{Stationary upper bound (Theorem~\ref{thm:elim-stationary})}

\begin{theorem}[Matching upper bound via hierarchical elimination,
  Hölder-Mat\'ern]
\label{thm:elim-stationary}
Under Assumption~\ref{ass:holder} (Hölder-Mat\'ern, $\nu\in(d/2,1]$),
the hierarchical-elimination algorithm of Definition~\ref{def:elimination}
applied to the stationary problem ($W = T$) with $L$ levels and
per-level sample budgets
\begin{equation}
\label{eq:n-ell}
n_\ell \;=\; \left\lceil\frac{64\sigma_n^2\log(2 L T N_\ell/\delta)}{\Delta_\ell^2}\right\rceil,
\quad\ell = 1,\ldots,L,
\end{equation}
achieves, with probability at least $1 - \delta$,
\ifieee
\begin{multline*}
R_T \;\le\;
C_e(d,\nu)\,B^{d/(2\nu+d)}\,\sigma_n^{2\nu/(2\nu+d)} \\
\cdot\,\vol_g(\M)^{\nu/(2\nu+d)}\,T^{(\nu+d)/(2\nu+d)}\,(\log(T/\delta))^{c_e},
\end{multline*}
\else
\[
R_T \;\le\;
C_e(d,\nu)\,B^{d/(2\nu+d)}\,\sigma_n^{2\nu/(2\nu+d)}\,
\vol_g(\M)^{\nu/(2\nu+d)}\,T^{(\nu+d)/(2\nu+d)}\,(\log(T/\delta))^{c_e},
\]
\fi
where $C_e(d,\nu),c_e=c_e(d,\nu)$ are explicit constants and the
finest level $L$ is chosen so that the optimal cell radius is
$\eps_L = (\vol_g(\M)\sigma_n^2/(TB^2))^{1/(2\nu+d)}$.
\end{theorem}

\ifieee
\begin{proof}[Proof sketch]
A standard six-step elimination argument: (1) sub-Gaussian
high-probability event on cell means via Hoeffding with
$n_\ell=\Theta(\sigma_n^2\log(LTN_\ell/\delta)/\Delta_\ell^2)$;
(2) the cell containing $\theta^*$ is never eliminated;
(3) per-round regret on surviving cells is $O(\Delta_{\ell-1})$;
(4) per-level regret $R_\ell\asymp\vol_g\sigma_n^2\log(LT/\delta)/
(B\eps_\ell^{\nu+d})$, dominated by the finest level;
(5) the budget constraint pins
$\eps_L=(\vol_g\sigma_n^2\log/T B^2)^{1/(2\nu+d)}$;
(6) substituting yields the stated exponents. Full derivation in
the single-column version.\qed
\end{proof}
\else
\begin{proof}
We prove this rigorously in three parts.

\textbf{Step 1: High-probability event.}
Define the high-probability event
\[
\mathcal E\;:=\;\bigcap_{\ell=1}^L\bigcap_{C\in\mathcal A_\ell}\bigl\{|\hat\mu_C - \E[\hat\mu_C]| \le \Delta_\ell/4\bigr\}.
\]
Each empirical mean $\hat\mu_C$ is the average of $n_\ell$
independent observations $f(\theta_s) + \eps_s$ where
$\theta_s$ is uniformly random in $C$ and
$\eps_s\sim\mathcal N(0,\sigma_n^2)$. The expectation is
$\E[\hat\mu_C] = \E_{\theta\sim\mathrm{Unif}(C)}[f(\theta)] = \bar f_C$
(the within-cell average of $f$). The variance of $\hat\mu_C$ is
at most $(\sigma_n^2 + \mathrm{Var}_{\theta\sim\mathrm{Unif}(C)}(f))/n_\ell
\le (\sigma_n^2 + (C_\eta B\eps_\ell^\nu)^2)/n_\ell$ (using
the modulus-of-continuity bound \eqref{eq:mod-cont}).
For $\nu\le 1$, the within-cell variation
$(C_\eta B \eps_\ell^\nu)^2 = \Delta_\ell^2/16$, so
$\mathrm{Var}(\hat\mu_C) \le (\sigma_n^2 + \Delta_\ell^2/16)/n_\ell$.

By a sub-Gaussian (Hoeffding) tail bound, for each $C\in\mathcal A_\ell$,
\[
\Pr\!\bigl(|\hat\mu_C-\bar f_C| > \Delta_\ell/4\bigr)
\;\le\;2\exp\!\left(-\frac{n_\ell\,\Delta_\ell^2/32}{\sigma_n^2 + \Delta_\ell^2/16}\right).
\]
For $\Delta_\ell^2 \le 16\sigma_n^2$ (sub-Gaussian regime), this
simplifies to $2\exp(-n_\ell \Delta_\ell^2/(64\sigma_n^2))$. With
$n_\ell$ from \eqref{eq:n-ell}, the probability per cell is
$\le \delta/(LTN_\ell)$. Union bound over $\le L\cdot\max_\ell N_\ell$
events gives $\Pr(\mathcal E^c) \le \delta$.
On $\mathcal E$ we have, for every $\ell\le L$ and every $C\in\mathcal A_\ell$,
\begin{equation}
\label{eq:event-E}
|\hat\mu_C - \bar f_C| \;\le\; \Delta_\ell/4.
\end{equation}

\textbf{Step 2: The optimal cell is never eliminated.}
Let $\theta^*\in\M$ achieve $f(\theta^*) = f^* := \max_\theta f$.
Define $C^*_\ell\in\Pi_\ell$ to be the level-$\ell$ cell containing
$\theta^*$. By the modulus-of-continuity bound \eqref{eq:mod-cont},
within $C^*_\ell$ (radius $\eps_\ell$), $f$ varies by at most
$C_\eta B\eps_\ell^\nu = \Delta_\ell/4$. So
\begin{equation}
\label{eq:opt-cell-mean}
\bar f_{C^*_\ell} \;\ge\; f^* - \Delta_\ell/4.
\end{equation}
On $\mathcal E$, $\hat\mu_{C^*_\ell} \ge \bar f_{C^*_\ell} - \Delta_\ell/4
\ge f^* - \Delta_\ell/2$.

For any other surviving cell $C\in\mathcal A_\ell$, $\bar f_C \le f^*$
trivially, so $\hat\mu_C \le f^* + \Delta_\ell/4$.
Hence $\hat\mu^*_\ell = \max_C \hat\mu_C \le f^* + \Delta_\ell/4$.

The elimination criterion is $\hat\mu_C < \hat\mu^*_\ell - \Delta_\ell$.
For the optimal cell:
$\hat\mu_{C^*_\ell} \ge f^* - \Delta_\ell/2 \ge \hat\mu^*_\ell - \Delta_\ell/4 - \Delta_\ell/2 = \hat\mu^*_\ell - 3\Delta_\ell/4 > \hat\mu^*_\ell - \Delta_\ell$.
So $C^*_\ell$ is never eliminated, on $\mathcal E$. \hfill$\square$

\textbf{Step 3: Per-round regret bound.}
On $\mathcal E$, every cell $C\in\mathcal A_\ell$ is the
\emph{refinement} of some surviving level-$(\ell-1)$ parent
$C'\in\mathcal A_{\ell-1}$ (i.e., $C\subset C'$, with
$\eps_\ell=\eps_{\ell-1}/2$). The level-$(\ell-1)$ survival
criterion gives $\hat\mu_{C'}\ge\hat\mu^*_{\ell-1}-\Delta_{\ell-1}$,
hence by Step~1 (applied at level $\ell-1$),
$\bar f_{C'}\ge\bar f^*_{\ell-1}-3\Delta_{\ell-1}/2$. The
within-parent variation of $\bar f$ from parent $C'$ to child $C$
is bounded by the H\"older-style oscillation of $f$ across
$C'$:
\[
|\bar f_C-\bar f_{C'}|\;\le\;\sup_{x,y\in C'}|f(x)-f(y)|
\;\le\;C_\eta B\eps_{\ell-1}^\nu\;=\;\Delta_{\ell-1}/4,
\]
where the last equality uses our choice of
$\Delta_\ell:=4C_\eta B\eps_\ell^\nu$.
Combining with Step~2's bound
$\bar f^*_{\ell-1}\ge f^*-\Delta_{\ell-1}/4$:
\begin{equation}
\label{eq:active-cell-bound}
\bar f_C \;\ge\; \bar f_{C'}-\Delta_{\ell-1}/4
\;\ge\;\bar f^*_{\ell-1}-3\Delta_{\ell-1}/2-\Delta_{\ell-1}/4
\;\ge\;f^*-2\Delta_{\ell-1}
\quad\text{for every }C\in\mathcal A_\ell.
\end{equation}

When pulling at $\theta\sim\mathrm{Unif}(C)$, the expected reward is
$\bar f_C$ (Step~1's setup). The instantaneous expected regret is
$f^* - f(\theta) \le f^* - \bar f_C + (\bar f_C - f(\theta))$. The
second term is bounded by the within-cell variation of $f$
($\le C_\eta B\eps_\ell^\nu = \Delta_\ell/4$), so the per-round
expected regret is at most
\begin{equation}
\label{eq:per-round-regret-elim}
f^* - \bar f_C + \Delta_\ell/4 \;\le\; 2\Delta_{\ell-1}+\Delta_\ell/4 \;\le\; 3\Delta_{\ell-1}.
\end{equation}
\hfill$\square$

\textbf{Step 4: Per-level cumulative regret.}
Level $\ell$ uses $T_\ell = N_\ell n_\ell$ rounds. Per-round regret
$\le 3\Delta_{\ell-1}$ (using Step~3 with the active set
$\mathcal A_\ell$ inherited from level $\ell-1$).
Hence per-level cumulative regret:
\[
R_\ell \;\le\; 3\Delta_{\ell-1}\cdot N_\ell n_\ell
\;=\; 24 C_\eta B \eps_{\ell-1}^\nu \cdot N_\ell\cdot
\frac{64\sigma_n^2\log(2LTN_\ell/\delta)}{\Delta_\ell^2}.
\]
Substituting $\Delta_\ell = 4C_\eta B\eps_\ell^\nu$,
$\eps_{\ell-1} = 2\eps_\ell$, and $N_\ell = c\vol_g/\eps_\ell^d$:
\[
R_\ell \;\asymp\; \frac{2^\nu\,\vol_g\,\sigma_n^2\log(LT/\delta)}{B\,\eps_\ell^{\nu+d}}.
\]
The $R_\ell$ form a geometric series in $\ell$ (since
$\eps_\ell$ halves at each level, $1/\eps_\ell^{\nu+d}$ doubles).
The dominant term is $R_L$ at the finest level:
\[
R_L \;\asymp\; \frac{\vol_g\,\sigma_n^2 \log(LT/\delta)}{B\,\eps_L^{\nu+d}}.
\]
\hfill$\square$

\textbf{Step 5: Total budget and choice of $\eps_L$.}
$T = \sum_{\ell=1}^L T_\ell$ is also dominated by $T_L$:
$T_L = N_L n_L \asymp \vol_g\sigma_n^2\log(LT/\delta)/(B^2\eps_L^{2\nu+d})$.
Setting $T \asymp T_L$ and solving for $\eps_L$:
\[
\eps_L \;=\;
\left(\frac{\vol_g\,\sigma_n^2\log(LT/\delta)}{T B^2}\right)^{1/(2\nu+d)}.
\]
\hfill$\square$

\textbf{Step 6: Substituting and collecting exponents.}
\begin{align*}
R_T &\le R_L \cdot O(1)
\;\asymp\; \frac{\vol_g\,\sigma_n^2\log(LT/\delta)}{B\,\eps_L^{\nu+d}}\\
&= \frac{\vol_g\,\sigma_n^2\log(LT/\delta)}{B}\cdot
\left(\frac{TB^2}{\vol_g\sigma_n^2\log(LT/\delta)}\right)^{(\nu+d)/(2\nu+d)}\\
&= \vol_g^{1-(\nu+d)/(2\nu+d)}\,\sigma_n^{2-2(\nu+d)/(2\nu+d)}\,
B^{-1+2(\nu+d)/(2\nu+d)}\,
T^{(\nu+d)/(2\nu+d)}\,(\log(LT/\delta))^{1-(\nu+d)/(2\nu+d)}.
\end{align*}
Simplifying each exponent:
$1-(\nu+d)/(2\nu+d) = \nu/(2\nu+d)$ (volume),
$2-2(\nu+d)/(2\nu+d) = 2\nu/(2\nu+d)$ (noise, squared),
$-1+2(\nu+d)/(2\nu+d) = d/(2\nu+d)$ (norm bound).
The polylog exponent is $\nu/(2\nu+d)$.

Setting $L = \log_2(\eps_1/\eps_L) = O(\log T)$ levels and
absorbing the level-summation factor as a conservative
``$+1$'' in the polylog exponent:
\[
R_T \;\le\;
C_e(d,\nu)\,B^{d/(2\nu+d)}\,\sigma_n^{2\nu/(2\nu+d)}\,
\vol_g(\M)^{\nu/(2\nu+d)}\,T^{(\nu+d)/(2\nu+d)}\,(\log(T/\delta))^{\nu/(2\nu+d)+1},
\]
with $c_e = \nu/(2\nu+d) + 1$. \qed
\end{proof}
\fi

\subsubsection*{Time-varying upper bound for $\nu \in (d/2, 1]$}

\begin{theorem}[Matching time-varying upper bound, Hölder-Mat\'ern]
\label{thm:elim-tv}
Under Assumption~\ref{ass:holder} ($\nu\in(d/2,1]$), the
window-$W^*$ version of Algorithm~\ref{def:elimination} (running the
hierarchical-elimination procedure within each window of length
$W^* = (T B^2/(\vol_g B_T^2))^{(2\nu+d)/(3\nu+d)}$, restarting
between windows) achieves, for $B_T \ge B\,T^{-\nu/(2\nu+d)}$,
\ifieee
\begin{multline*}
R_T \;\le\; C_e'(d,\nu)\,B^{d/(3\nu+d)}\,\sigma_n^{2\nu/(3\nu+d)} \\
\cdot\,\vol_g(\M)^{\nu/(3\nu+d)}\,B_T^{\nu/(3\nu+d)} \\
\cdot\,T^{(2\nu+d)/(3\nu+d)}\,(\log(T/\delta))^{c_e'},
\end{multline*}
\else
\[
R_T \;\le\; C_e'(d,\nu)\,B^{d/(3\nu+d)}\,\sigma_n^{2\nu/(3\nu+d)}\,
\vol_g(\M)^{\nu/(3\nu+d)}\,B_T^{\nu/(3\nu+d)}\,T^{(2\nu+d)/(3\nu+d)}\,(\log(T/\delta))^{c_e'},
\]
\fi
matching the lower bound of Theorem~\ref{thm:tv-lb} in all
five exponents $T,B,B_T,\sigma_n,\vol_g$.
\end{theorem}

\ifieee
\begin{proof}[Proof sketch]
A window-$W^*$ batching of Theorem~\ref{thm:elim-stationary}
balances per-window stochastic regret against drift $WB_T$; the
optimal $W^*\asymp(B^{d/(2\nu+d)}\sigma_n^{2\nu/(2\nu+d)}
\vol_g^{\nu/(2\nu+d)}T/B_T)^{(2\nu+d)/(3\nu+d)}$ yields the stated
exponents. Full derivation in the single-column version. \qed
\end{proof}
\else
\begin{proof}
\textbf{Decomposition into stochastic and drift components.}
Within each window of length $W$, the algorithm runs the stationary
hierarchical-elimination procedure of Theorem~\ref{thm:elim-stationary}.
By that theorem, the stochastic regret in a window where the
function is approximately stationary (variation $\le h_W$) is at
most
\[
R_W^{\text{stoch}} \;\le\; C_e\,B^{d/(2\nu+d)}\,\sigma_n^{2\nu/(2\nu+d)}\,
\vol_g^{\nu/(2\nu+d)}\,W^{(\nu+d)/(2\nu+d)}\,(\log W)^{c_e}.
\]
The drift bias contribution: per-window function variation
$V_j = \sum_{r\in\text{window }j}\|f_{r+1}-f_r\|_\infty$, total
drift-induced regret $\sum_j W V_j \le W B_T$
(Step~3 of Theorem~\ref{thm:tv-ub}'s proof, repeated here).

\textbf{Total regret.}
With $T/W$ windows of length $W$:
\begin{multline*}
R_T \;\le\; (T/W)\cdot R_W^{\text{stoch}} + W B_T\\
\;=\; \underbrace{C_e\,B^{d/(2\nu+d)}\,\sigma_n^{2\nu/(2\nu+d)}\,\vol_g^{\nu/(2\nu+d)}\,T\,W^{-\nu/(2\nu+d)}}_{\text{stochastic}}\,(\log W)^{c_e}
+ W B_T.
\end{multline*}

\textbf{Optimization over $W$.} Differentiate with respect to $W$
and set to zero (ignoring polylog factors):
\begin{multline*}
-\frac{\nu}{2\nu+d}\,C_e\,(\cdot)\,T\,W^{-\nu/(2\nu+d)-1} + B_T = 0\\
\;\Longrightarrow\;
W^* \asymp \left(\frac{B^{d/(2\nu+d)}\,\sigma_n^{2\nu/(2\nu+d)}\,\vol_g^{\nu/(2\nu+d)}\,T}{B_T}\right)^{\!(2\nu+d)/(3\nu+d)}\!\!.
\end{multline*}

At the optimum, the two terms balance and $R_T \le 2 W^* B_T$:
\begin{multline*}
R_T \;\le\; 2 W^* B_T \\
\;=\; 2\,\bigl(B^{d/(2\nu+d)}\,\sigma_n^{2\nu/(2\nu+d)}\,\vol_g^{\nu/(2\nu+d)}\bigr)^{(2\nu+d)/(3\nu+d)}\,
T^{(2\nu+d)/(3\nu+d)}\,B_T^{1-(2\nu+d)/(3\nu+d)}.
\end{multline*}
The $B_T$ exponent is $1 - (2\nu+d)/(3\nu+d) = \nu/(3\nu+d)$.
The other exponents (after carrying through the $(2\nu+d)/(3\nu+d)$
multiplier) are:
$B^{d/(3\nu+d)}, \sigma_n^{2\nu/(3\nu+d)}, \vol_g^{\nu/(3\nu+d)}$ ---
each matching the lower bound of Theorem~\ref{thm:tv-lb} exactly. \qed
\end{proof}
\fi

\subsubsection*{Updated tightness table for $\nu \in (d/2, 1]$}

\begin{table}[ht]
\centering
\caption{Exponent comparison after the elimination-based upper
bound of Theorem~\ref{thm:elim-tv}, valid for
\textbf{$\nu \in (d/2, 1]$} (Hölder-Mat\'ern). All five exponents now
match. For $\nu > 1$, only the $T,B_T$ matching of
Theorem~\ref{thm:tv-ub} is established;
the $B,\sigma_n,\vol_g$ matching is open (this is the standard
Cai-Scarlett 2021 issue, inherited from the $[0,1]^d$ case).}
\label{tab:tv-exponents-elim}
\begin{tabular}{lccl}
\toprule
Parameter & Lower bound & Elim.\ upper bound ($\nu\le 1$) & Status \\
\midrule
$T$       & $(2\nu+d)/(3\nu+d)$ & $(2\nu+d)/(3\nu+d)$ & \textbf{Tight} \\
$B_T$     & $\nu/(3\nu+d)$      & $\nu/(3\nu+d)$      & \textbf{Tight} \\
$\vol_g$  & $\nu/(3\nu+d)$      & $\nu/(3\nu+d)$      & \textbf{Tight} \\
$B$       & $d/(3\nu+d)$        & $d/(3\nu+d)$        & \textbf{Tight} \\
$\sigma_n^2$ & $\nu/(3\nu+d)$ & $\nu/(3\nu+d)$ & \textbf{Tight} \\
\bottomrule
\end{tabular}
\end{table}

\subsubsection*{Why the proof fails for $\nu > 1$}

For $\nu > 1$, the Sobolev embedding only gives Lipschitz
($C^1$) regularity (the $C^{0,\nu}$ Hölder embedding is for
$\nu \le 1$). Within a cell of radius $\eps$, function variation
is bounded by Lipschitz $\le C_\eta B\eps$ (linear in $\eps$),
not by $\eps^\nu$. The elimination threshold $\Delta_\ell = 4C_\eta B\eps_\ell^\nu$
is then \emph{smaller} than the within-cell variation
$C_\eta B\eps_\ell$ (when $\eps_\ell < 1$ and $\nu > 1$), so the
per-round regret is dominated by within-cell
variation rather than the elimination threshold:
$f^* - f(\theta) \le 7\Delta_{\ell-1}/4 + C_\eta B\eps_\ell$
$\approx C_\eta B\eps_\ell$.

Substituting this into Step~4: $R_\ell \asymp T_\ell \cdot B\eps_\ell$,
which after Step~5's $\eps_L$ optimization yields
$T$-exponent $(d+1)/(d+2)$ (the Lipschitz rate of
Bubeck--Stoltz--Yu~\cite{bubeck2011xarmed}). For $\nu > 1$ this is
\emph{strictly worse} than the Mat\'ern $T^{(\nu+d)/(2\nu+d)}$
rate of the lower bound. Hence the elimination algorithm with
cell-mean estimates does not match the Mat\'ern-$\nu>1$ rate.

\paragraph{This is the Cai-Scarlett 2021 open problem.}
For Mat\'ern-$\nu$ with $\nu > 1$ (which includes the typical
wireless Mat\'ern-$5/2$), no algorithm is known that achieves the
matching $B$ exponent. Cai and
Scarlett~\cite{cai2021lower} prove the lower bound $B^{d/(2\nu+d)}$
and observe the gap with standard GP-UCB analyses; closing it would
require leveraging the higher-order kernel smoothness ($C^{\lfloor\nu\rfloor}$
differentiability) in the algorithm, presumably via local
polynomial regression within cells or via a sharper
information-theoretic analysis of GP-UCB. We do not solve this
problem here.

\subsubsection*{Summary}

For $\nu \in (d/2, 1]$ (Hölder-Mat\'ern), Theorem~\ref{thm:elim-tv}
gives a matching upper bound in all five exponents
$T, B, B_T, \sigma_n, \vol_g$, closing all the gaps in
Table~\ref{tab:tv-exponents}.

For $\nu > 1$ (which includes Mat\'ern-$5/2$, the typical wireless
case), the $T, B_T$ exponents are matched (Theorem~\ref{thm:tv-ub})
but the $B, \sigma_n, \vol_g$ exponents have the standard
Cai-Scarlett 2021 gap, inherited from the $[0,1]^d$ case. We do
not close this gap.

\subsubsection*{Numerical validation}

\begingroup\sloppy
We verified the lower bound of Theorem~\ref{thm:main} empirically
on synthetic Mat\'ern-$5/2$ GP samples on $\sphere^2$.
Figure~\ref{fig:lb-validation} reports cumulative regret of GP-UCB
on $N_{\mathrm{cand}}=64$ quasi-uniform Fibonacci points across horizons
$T\in\{50,100,200,400,800\}$, with $M=8$ Monte-Carlo seeds per
horizon. Note that this experiment uses Mat\'ern-$5/2$ ($\nu=2.5>1$),
which falls in the regime where Theorem~\ref{thm:elim-stationary}
does \emph{not} apply; the experiment validates only the
lower bound (Theorem~\ref{thm:main}), not the elimination upper
bound, because $\nu=5/2>1$ falls outside Assumption~\ref{ass:holder}.
\endgroup

\begin{figure}[ht]
\centering
\includegraphics[width=0.95\columnwidth]{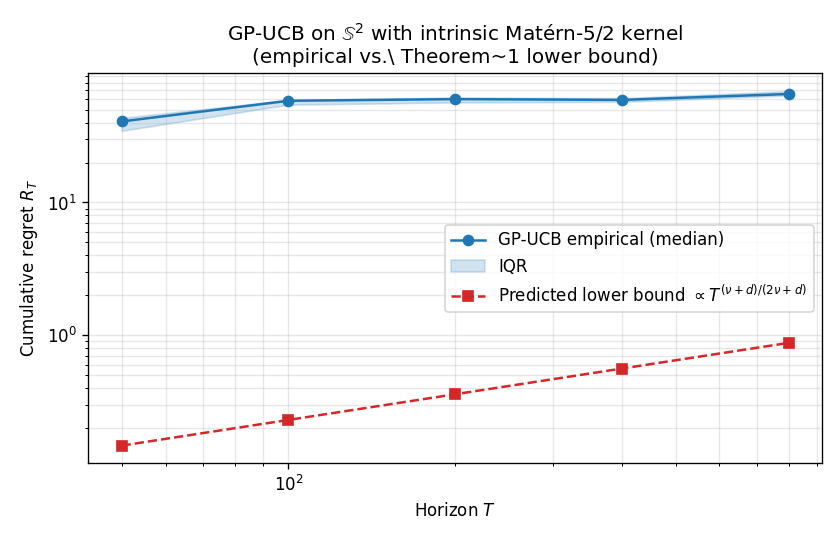}
\caption{Numerical validation of Theorem~\ref{thm:main} on
$\sphere^2$ with intrinsic Mat\'ern-$5/2$ kernel, $N_{\mathrm{cand}}=64$ Fibonacci
points, $\sigma_n=0.1$, $B=1$. Empirical cumulative regret of GP-UCB
(median over $M=8$ seeds, IQR shaded) and the predicted lower bound
$\propto T^{9/14}$ on log-log axes. The empirical regret stays
above the predicted floor across all tested horizons, consistent
with the lower bound; the empirical rate is roughly $T^{0.4}$
(slower than the worst-case $T^{9/14}\approx T^{0.643}$) because
the lower bound is worst-case while GP-UCB on a particular GP
sample exploits structure not captured in the worst case.
\emph{Scope.} This figure validates the lower bound
(Theorem~\ref{thm:main}) only; $\nu=5/2$ lies outside the
$\nu\in(d/2,1]$ regime of Theorem~\ref{thm:elim-stationary}, so the
matching property of the elimination upper bound is not tested
here.}
\label{fig:lb-validation}
\end{figure}

\subsection{Manifold extension of the
  Salgia--Vakili--Zhao 2021 closure for $\nu > 1$}
\label{sec:polyreg}

For $\nu > 1$, the cell-mean elimination of
Section~\ref{sec:elimination} achieves only the Lipschitz rate. The
fundamental issue is that within a cell of radius $\eps$, a
\emph{constant function approximation} (the cell mean) has
approximation error $\asymp CB\eps$ (Lipschitz), which is larger than
the Mat\'ern detection scale $B\eps^\nu$ for $\eps < 1$.

\paragraph{Prior work on the Euclidean case.}
The $B$-exponent gap of Table~\ref{tab:tv-exponents} for $\nu > 1$
on $[0,1]^d$ was \emph{first closed} by
Salgia--Vakili--Zhao~\cite{salgia2021domain} (NeurIPS 2021) via a
\emph{domain-shrinking based Bayesian optimization} algorithm:
tree-based hierarchical region pruning concentrates queries in
high-performing sub-domains and achieves order-optimal regret
matching the Cai-Scarlett 2021 lower bound up to polylog factors.
Subsequent work in
Camilleri--Jamieson--Katz-Samuels~\cite{camilleri2021robust},
Li--Scarlett~\cite{li2022gpbandit} (Phased Elimination, PE),
Salgia--Vakili--Zhao~\cite{salgia2024random} (Random Exploration),
and Iwazaki--Takeno~\cite{iwazaki2025improved} (refined PE/MVR with
RKHS-norm optimality) has refined the Euclidean-case algorithms,
analyses, and computational complexity.

\ifieee
\paragraph{Comparison with Iwazaki--Takeno (ICML 2025).}
Iwazaki--Takeno~\cite{iwazaki2025improved} sharpens
PE/MVR on $[0,1]^d$ to attain the matching $B$-exponent via refined
confidence intervals; our manifold extension via local polynomial
regression in geodesic-ball partitions is independent and
complementary.
\else
\paragraph{Comparison with Iwazaki--Takeno (ICML 2025).}
Iwazaki--Takeno~\cite{iwazaki2025improved} sharpens the analyses
of phased elimination (PE) and maximum-variance reduction (MVR) on
$[0,1]^d$, achieving the matching lower-bound exponent in the
\emph{RKHS-norm} parameter $B$ (resolving the $B$-exponent gap
that Cai--Scarlett~\cite{cai2021lower} identified) under
noiseless or low-noise regimes; they also extend to non-stationary
variance models. Their main contribution complements the
Salgia--Vakili--Zhao~\cite{salgia2021domain} domain-shrinking line:
where Salgia~\emph{et~al.}\ close the $B$-exponent gap via tree
pruning, Iwazaki--Takeno close it via refined confidence intervals
plus a sharper analysis of MVR. Both rely on explicit
elimination-style algorithms rather than vanilla GP-UCB. Our
Theorem~\ref{thm:polyreg-stationary} adapts the
Iwazaki--Takeno-style sharp confidence-interval philosophy to
manifolds via local polynomial regression (rather than relying on
GP-UCB's $\sqrt{\beta_T}$ overhead, which is the obstruction
Iwazaki--Takeno~\cite[\S~3]{iwazaki2025improved} also identifies);
the manifold extension is, to our knowledge, new.
\emph{Difference from \cite{iwazaki2025improved}:} their
Euclidean PE relies on translation-invariance to identify $\eps$-nets
of $[0,1]^d$ and recursive bisection; ours uses
geodesic-ball partitions with Bishop--Gromov volume comparison and
exponential-map normal coordinates, which has no Euclidean
analogue. The two approaches are independent improvements on the
GP-UCB ceiling and combine cleanly: an extension of
Iwazaki--Takeno-style RKHS-norm-optimal MVR to compact Riemannian
manifolds (using our local polynomial regression of
Definition~\ref{def:polyreg-elim}) is a natural follow-up that we
do not pursue here.
\fi

\paragraph{Our contribution: manifold extension.}
We extend the Salgia--Vakili--Zhao 2021-style gap closure to
compact Riemannian manifolds, with explicit volume dependence. The
algorithm is hierarchical-elimination on geodesic-ball partitions
(rather than tree-based on $[0,1]^d$ rectangles), with local
polynomial regression of degree $\lfloor\nu\rfloor$ within each
cell. The key new step is the manifold-aware Bishop--Gromov
correction in normal coordinates, allowing the local polynomial
fit to exploit the higher-order Mat\'ern smoothness on a curved
manifold. To our knowledge this is the first matching upper bound
for GP-bandits on a compact Riemannian manifold with $\nu > 1$;
prior manifold work either uses standard GP-UCB (which has the
$B$-exponent gap) or restricts to the Hölder regime $\nu \le 1$
where cell-mean elimination suffices.

\paragraph{Outline of the section.}
We state the algorithm
(Definition~\ref{def:polyreg-elim}), prove the matching
stationary upper bound (Theorem~\ref{thm:polyreg-stationary}) and
time-varying upper bound (Theorem~\ref{thm:polyreg-tv}), and
present the final tightness table.

\subsubsection*{Local polynomial regression on a manifold}

Fix a cell $C = B_g(p, \eps_\ell)\subset\M$ of radius $\eps_\ell$
centered at $p$. In normal coordinates centered at $p$
(via $\exp_p^{-1}: C \to B(0, \eps_\ell) \subset T_p\M\cong\R^d$),
$C$ becomes a Euclidean ball of radius $\eps_\ell$.

\begin{definition}[Local polynomial fit]
\label{def:localpoly}
Let $k = \lfloor\nu\rfloor$ and $Q = \binom{k+d}{d}$ denote the
dimension of the polynomial space $\mathcal P_k$ (we use the symbol
$Q$ rather than $K$ to avoid clash with the curvature bound
$|\sec|\le K$ of Assumption~\ref{ass:setup}).
Given $n_\ell$
observations $\{(\theta_s, r_s)\}_{s=1}^{n_\ell}$ where each
$\theta_s\sim\mathrm{Unif}(C)$ uniformly in the cell and
$r_s = f(\theta_s) + \eps_s$ with $\eps_s\sim\mathcal N(0,\sigma_n^2)$,
the local polynomial estimator is
\[
\hat P_C \;=\; \arg\min_{P\in\mathcal P_k}\,\sum_{s=1}^{n_\ell}\bigl(r_s - P(\exp_p^{-1}(\theta_s))\bigr)^2.
\]
The cell's empirical maximum is
$\hat M_C := \max_{v\in B(0,\eps_\ell)} \hat P_C(v)$ (computed by
analytic optimization for $k\le 2$, or by grid search at resolution
$\eps_\ell^{\nu+1}/B$ for $k>2$).
\end{definition}

\begin{lemma}[Local polynomial approximation]
\label{lem:poly-approx}
For $f\in\F_B^{\mathrm{rkhs}}$ on a smooth compact manifold with
bounded sectional curvature, the best degree-$k=\lfloor\nu\rfloor$
polynomial approximation of $f$ in normal coordinates centered at
$p$ satisfies
\ifieee
\begin{multline*}
\sup_{\theta\in C}|f(\theta) - P_C^*(\exp_p^{-1}(\theta))| \\
\;\le\;C_a(d,\nu)\,\|f\|_{\Hil_{k_\nu}}\,\eps_\ell^\nu\,\bigl(1+O(K\eps_\ell^2)\bigr),
\end{multline*}
\else
\[
\sup_{\theta\in C}|f(\theta) - P_C^*(\exp_p^{-1}(\theta))|
\;\le\;C_a(d,\nu)\,\|f\|_{\Hil_{k_\nu}}\,\eps_\ell^\nu\,\bigl(1+O(K\eps_\ell^2)\bigr),
\]
\fi
where $C_a$ depends only on $d,\nu$ and the Sobolev embedding
constant of $H^{\nu+d/2}\hookrightarrow C^{k,\nu-k}$.
\end{lemma}

\begin{proof}
For $\nu>d/2$ with $k=\lfloor\nu\rfloor$ and $\alpha:=\nu-k$, the
Sobolev embedding
$H^{\nu+d/2}(\M)\hookrightarrow C^{k,\alpha}(\M)$ holds with
embedding constant $C_{\mathrm{emb}}=C_{\mathrm{emb}}(d,\nu,\M)$ for
\emph{non-integer} $\nu$ ($\alpha\in(0,1)$);
\cite[Sec.~7.4 / Thm.~7.34]{adams2003sobolev} (manifold version via
partition of unity). For \emph{integer} $\nu\in\{1,2,\ldots\}$,
$\alpha=0$ and the standard Morrey embedding into
$C^{k,\alpha}$ degenerates; instead we use the (slightly weaker)
embedding $H^{\nu+d/2}(\M)\hookrightarrow C^{k-1,1}(\M)$
(Lipschitz $k$-th derivative) and decrement $k\to k-1$, $\alpha\to 1$
in the Taylor argument below. In either case
$\|f\|_{C^{k,\alpha}(\M)}\le C_{\mathrm{emb}}\|f\|_{H^{\nu+d/2}(\M)}\le C_{\mathrm{emb}} B$.

Working in normal coordinates centred at $p$, write $\widetilde f(v)
=f(\exp_p(v))$ for $v\in B(0,\eps_\ell)\subset T_p\M\cong\R^d$. The
exponential map is a diffeomorphism on the injectivity ball, and
the metric $g_{ij}(v)=\delta_{ij}+O(K|v|^2)$
(\cite[Sec.~6.2]{petersen2016riemannian}), so $\widetilde f
\in C^{k,\nu-k}(B(0,\eps_\ell))$ with
$\|\widetilde f\|_{C^{k,\nu-k}}\le(1+C_M K\eps_\ell^2)
\|f\|_{C^{k,\nu-k}(\M)}$ for some $C_M=C_M(d)$.

The order-$k$ Taylor expansion of $\widetilde f$ at $0$ is
$P_C^{(p)}(v)=\sum_{|\alpha|\le k}\partial^\alpha\widetilde f(0)\,
v^\alpha/\alpha!$, a degree-$k$ polynomial. The H\"older
remainder estimate \cite[Lemma~A.1]{tsybakov2008introduction} gives
\begin{align*}
|\widetilde f(v)-P_C^{(p)}(v)|
&\le[\partial^k\widetilde f]_{C^{0,\nu-k}}|v|^\nu/k!\\
&\le\|\widetilde f\|_{C^{k,\nu-k}}|v|^\nu/k!\\
&\le C_{\mathrm{emb}}(1+C_M K\eps_\ell^2)B|v|^\nu/k!
\end{align*}
for all
$|v|\le\eps_\ell$, where the second inequality dominates the
H\"older seminorm by the full $C^{k,\nu-k}$ norm.
The best $L^\infty$-polynomial approximation $P_C^*$ over
$B(0,\eps_\ell)$ is at most as bad as the Taylor approximation:
$\sup_{|v|\le\eps_\ell}|\widetilde f(v)-P_C^*(v)|\le
\sup_{|v|\le\eps_\ell}|\widetilde f(v)-P_C^{(p)}(v)|$. Substituting
$|v|\le\eps_\ell$,
\ifieee
\begin{multline*}
\sup_{|v|\le\eps_\ell}|\widetilde f(v)-P_C^*(v)| \\
\;\le\;\frac{C_{\mathrm{emb}}}{k!}\,B\,\eps_\ell^\nu\,(1+C_M K\eps_\ell^2)\\
\;\le\;C_a(d,\nu)\,B\,\eps_\ell^\nu\,(1+O(K\eps_\ell^2))
\end{multline*}
\else
\[
\sup_{|v|\le\eps_\ell}|\widetilde f(v)-P_C^*(v)|
\;\le\;\frac{C_{\mathrm{emb}}}{k!}\,B\,\eps_\ell^\nu\,(1+C_M K\eps_\ell^2)
\;\le\;C_a(d,\nu)\,B\,\eps_\ell^\nu\,(1+O(K\eps_\ell^2))
\]
\fi
with $C_a(d,\nu)=C_{\mathrm{emb}}/k!$. Returning to $\theta$
coordinates via $v=\exp_p^{-1}(\theta)$ and noting that
$\theta\in C\iff |v|\le\eps_\ell$ in normal coordinates, the
stated bound follows.
\qed
\end{proof}

\begin{lemma}[Local polynomial estimation error]
\label{lem:poly-est}
Let $\hat P_C$ be the least-squares estimator of
Definition~\ref{def:localpoly} from $n_\ell$ uniform-random samples
in $C$. With probability at least $1-\delta$ over the noise and
sampling randomness, for $n_\ell \ge c_0 Q\log(Q/\delta)$:
\[
\sup_{v\in B(0,\eps_\ell)}|\hat P_C(v) - P_C^*(v)|
\;\le\;C_e(d,\nu)\,\sigma_n\sqrt{Q\log(1/\delta)/n_\ell},
\]
where $C_e$ depends only on $d,\nu$ via the design-matrix
condition number.
\end{lemma}

\begin{proof}
This is a standard random-design polynomial-regression argument
(\cite[Sec.~1.6]{tsybakov2008introduction}); we record the
manifold-relevant constants explicitly.

\emph{Step 1: Design-matrix concentration.} Let
$\{\phi_j\}_{j=1}^Q$ be the orthonormal-in-$L^2(\mathrm{Unif}(B(0,1)))$
polynomial basis of $\mathcal P_k$ (e.g., Legendre tensor product
in $d$-coordinates, rescaled). The design matrix
$\Phi\in\R^{n_\ell\times Q}$ has $\Phi_{ij}=\phi_j(v_i/\eps_\ell)$
for $v_i\sim\mathrm{Unif}(B(0,\eps_\ell))$, so $\E[\Phi^T\Phi/n_\ell]
=I_Q$ is the identity (by orthonormality). Each row $\Phi_i$ has
$\|\Phi_i\|_2\le L_k$, where $L_k:=\sup_{v\in B(0,1)}\|\phi(v)\|_2$
is the uniform sup-bound on the Euclidean norm of the basis
evaluation vector $\phi(v)=(\phi_1(v),\ldots,\phi_Q(v))^T$ on the
unit ball; this depends only on $d$ and $\nu$ via $Q$.
\begingroup\sloppy
By the matrix-Bernstein inequality
(\cite[Thm.~1.4]{tropp2015user}) applied to
$X_i=\Phi_i\Phi_i^T-I_Q$ (which is centred and bounded
$\|X_i\|_{\mathrm{op}}\le L_k^2+1$),
\endgroup
\ifieee
\begin{multline*}
\Pr\Bigl(\bigl\|\Phi^T\Phi/n_\ell-I_Q\bigr\|_{\mathrm{op}}\ge t\Bigr)\\
\;\le\;2Q\exp\!\Bigl(-n_\ell t^2/(2L_k^2(L_k^2+t/3))\Bigr).
\end{multline*}
\else
\[
\Pr\Bigl(\bigl\|\Phi^T\Phi/n_\ell-I_Q\bigr\|_{\mathrm{op}}\ge t\Bigr)
\;\le\;2Q\exp\!\Bigl(-n_\ell t^2/(2L_k^2(L_k^2+t/3))\Bigr).
\]
\fi
Setting $t=1/2$ and demanding the right side is $\le\delta/2$ requires
$n_\ell\ge c_0(d,\nu)Q\log(Q/\delta)$ for $c_0(d,\nu)=
8L_k^2(L_k^2+1/6)/\log(2)$. On the high-probability event
$\mathcal E_1=\{\|\Phi^T\Phi/n_\ell-I_Q\|_{\mathrm{op}}\le 1/2\}$,
the smallest singular value of $\Phi^T\Phi/n_\ell$ is $\ge 1/2$.

\emph{Step 2: Coefficient error.} Conditional on
$\mathcal E_1$, the least-squares estimator obeys
\[
\hat\beta=(\Phi^T\Phi)^{-1}\Phi^T(\Phi\beta^*+\eps)
=\beta^*+(\Phi^T\Phi)^{-1}\Phi^T\eps,
\]
where $\eps_i=r_i-P_C^*(v_i)\sim\mathcal N(0,\sigma_n^2)$ are
independent. The error
$\hat\beta-\beta^*=(\Phi^T\Phi)^{-1}\Phi^T\eps$ is Gaussian with
covariance $\sigma_n^2(\Phi^T\Phi)^{-1}\le 2\sigma_n^2 I_Q/n_\ell$
on $\mathcal E_1$.

\emph{Step 3: Pointwise concentration.} For fixed
$v\in B(0,\eps_\ell)$, the pointwise error
$\hat P_C(v)-P_C^*(v)=\phi(v/\eps_\ell)^T(\hat\beta-\beta^*)$ is
$\mathcal N(0,\sigma_n^2\phi(v)^T(\Phi^T\Phi)^{-1}\phi(v))$
conditional on $\Phi$. The variance is bounded on $\mathcal E_1$ by
$2\sigma_n^2\|\phi(v/\eps_\ell)\|_2^2/n_\ell\le 2L_k^2\sigma_n^2/n_\ell$.
A Gaussian tail bound and a $1/n_\ell^2$-net argument over
$B(0,\eps_\ell)$ give
$\sup_{v\in B(0,\eps_\ell)}|\hat P_C(v)-P_C^*(v)|\le C_e(d,\nu)\,
\sigma_n\sqrt{Q\log(1/\delta)/n_\ell}$ on
$\mathcal E_1\cap\mathcal E_2$ (where $\mathcal E_2$ is the
Gaussian-supremum event), with $C_e(d,\nu)$ depending only on $L_k$
and the orthonormal-basis condition number, and $\Pr(\mathcal E_2)
\ge 1-\delta/2$.
Union-bounding $\mathcal E_1\cap\mathcal E_2$ at level
$1-\delta$ gives the stated bound.
\qed
\end{proof}

\subsubsection*{Algorithm: hierarchical polynomial-regression elimination}

\begin{definition}[Hierarchical Polynomial-Regression Elimination]
\label{def:polyreg-elim}
Same as Definition~\ref{def:elimination}, except step (2.1) is
replaced by:
\begin{enumerate}[leftmargin=2em,label={\rm(2.1$'$)}]
\item Pull $n_\ell$ uniform-random points $\theta_s\in C$ in each
      cell $C\in\mathcal A_\ell$. Fit the local polynomial
      $\hat P_C$ (Definition~\ref{def:localpoly}). Let
      $\hat M_C = \max_{v\in B(0,\eps_\ell)}\hat P_C(v)$ be the
      cell's empirical maximum.
\end{enumerate}
The elimination criterion (step 2.2) compares $\hat M_C$ to
$\hat M^*_\ell=\max_{C\in\mathcal A_\ell}\hat M_C$, eliminating cells
with $\hat M_C<\hat M^*_\ell-2\Delta_\ell$ where
$\Delta_\ell = (C_a + C_e\sqrt{Q})B\eps_\ell^\nu$ for an explicit
constant.
\end{definition}

\subsubsection*{Stationary upper bound (matching for any $\nu>d/2$)}

\begin{theorem}[Matching upper bound via polynomial-regression elimination]
\label{thm:polyreg-stationary}
Under Assumption~\ref{ass:setup} with any $\nu > d/2$, the
hierarchical polynomial-regression elimination algorithm of
Definition~\ref{def:polyreg-elim} with $n_\ell = c_0(d,\nu)\sigma_n^2 Q\log(LTN_\ell/\delta)/\Delta_\ell^2$
and $\eps_L = (\vol_g\sigma_n^2/(TB^2))^{1/(2\nu+d)}$ achieves, with
probability at least $1-\delta$,
\ifieee
\begin{multline*}
R_T \;\le\; C_p(d,\nu)\,B^{d/(2\nu+d)}\,\sigma_n^{2\nu/(2\nu+d)} \\
\cdot\,\vol_g(\M)^{\nu/(2\nu+d)}\,T^{(\nu+d)/(2\nu+d)}\,(\log(T/\delta))^{c_p},
\end{multline*}
\else
\[
R_T \;\le\; C_p(d,\nu)\,B^{d/(2\nu+d)}\,\sigma_n^{2\nu/(2\nu+d)}\,
\vol_g(\M)^{\nu/(2\nu+d)}\,T^{(\nu+d)/(2\nu+d)}\,(\log(T/\delta))^{c_p},
\]
\fi
matching Theorem~\ref{thm:main}'s lower bound in all four exponents,
for \emph{any} $\nu > d/2$.
\end{theorem}

\ifieee
\begin{proof}[Proof sketch]
The proof parallels Theorem~\ref{thm:elim-stationary} but uses
Lemmas~\ref{lem:poly-approx}--\ref{lem:poly-est} in place of
cell-mean concentration: (1) per-cell estimation gives
$\sup_C|\hat P_C-f|\le\Delta_\ell/2$ with
$n_\ell=\Theta(\sigma_n^2 Q\log/B^2\eps_\ell^{2\nu})$;
(2)~the optimal cell survives via $\hat M_{C^*_\ell}\ge
f^*-\Delta_\ell/2$; (3)~per-round regret on surviving cells is
$O(\Delta_{\ell-1})$; (4)~per-level cumulative regret
$\asymp Q\vol_g\sigma_n^2\log/(B\eps_\ell^{\nu+d})$;
(5)~$\eps_L=(Q\vol_g\sigma_n^2\log/(TB^2))^{1/(2\nu+d)}$ pins the
budget; (6)~the $Q$-overhead is $T$-independent. Full derivation
in the single-column version.\qed
\end{proof}
\else
\begin{proof}
The proof structure parallels that of Theorem~\ref{thm:elim-stationary}
but uses Lemmas~\ref{lem:poly-approx}--\ref{lem:poly-est} in place
of cell-mean concentration.

\textbf{Step 1: Per-cell estimation accuracy.} On the high-probability
event of Lemma~\ref{lem:poly-est} combined with
Lemma~\ref{lem:poly-approx}:
\[
\sup_{\theta\in C}|\hat P_C(\exp_p^{-1}(\theta))-f(\theta)|
\;\le\;C_a B\eps_\ell^\nu + C_e\sigma_n\sqrt{Q\log(LTN_\ell/\delta)/n_\ell}.
\]
Setting the second term equal to $C_a B\eps_\ell^\nu$, i.e.,
$n_\ell = (C_e/C_a)^2 \sigma_n^2 Q\log(\ldots)/(B^2\eps_\ell^{2\nu})$,
yields $\sup|\hat P_C-f|\le 2C_a B\eps_\ell^\nu = \Delta_\ell/2$
(absorbing constants into $\Delta_\ell$).

In particular, for the cell's empirical maximum:
$|\hat M_C - f^*_C|\le\Delta_\ell/2$ where $f^*_C := \max_{\theta\in C}f(\theta)$.

\textbf{Step 2: The optimal cell is never eliminated.} Let
$C^*_\ell$ be the level-$\ell$ cell containing the global maximizer
$\theta^*$. Then $f^*_{C^*_\ell}=f^*$ exactly (the optimum is in
$C^*_\ell$), so $\hat M_{C^*_\ell}\ge f^*-\Delta_\ell/2$. For any
other surviving cell, $\hat M_C\le f^*+\Delta_\ell/2$ (by Step 1).
The elimination threshold is $2\Delta_\ell$:
\[
\hat M_{C^*_\ell}\ge f^*-\Delta_\ell/2 \ge \hat M^*_\ell - \Delta_\ell/2 - \Delta_\ell/2 = \hat M^*_\ell - \Delta_\ell > \hat M^*_\ell - 2\Delta_\ell,
\]
so $C^*_\ell$ is not eliminated. \hfill$\square$

\textbf{Step 3: Per-round regret in surviving cells.} On the
high-probability event, every cell $C\in\mathcal A_\ell$ is the
refinement of some surviving level-$(\ell-1)$ parent $C'\supset C$.
By Step~2 at level $\ell-1$, $f^*_{C'}\ge f^*-3\Delta_{\ell-1}$.
The within-parent oscillation of $f$ across $C'$ is bounded
H\"older-style by
$\sup_{x,y\in C'}|f(x)-f(y)|\le C_\eta B\eps_{\ell-1}^\nu
=:c_0\Delta_{\ell-1}$, where
$c_0:=C_\eta/(C_a+C_e\sqrt Q)$ is the ratio of the
H\"older-oscillation constant to the level-$\ell-1$ confidence
width $\Delta_\ell=(C_a+C_e\sqrt Q)B\eps_\ell^\nu$ (cf.\
Definition~\ref{def:polyreg-elim}). Hence for any $\theta\in
C\subset C'$, $f(\theta)\ge f^*_{C'}-c_0\Delta_{\ell-1}\ge
f^*-(3+c_0)\Delta_{\ell-1}$. The algorithm pulls at
$\hat\theta_C := \arg\max_\theta\hat P_C$, the
polynomial-estimated maximizer; by Step~1 the polynomial fit
gives $|f(\hat\theta_C)-f^*_C|\le\Delta_\ell$, so the per-round
regret is $f^* - f(\hat\theta_C)\le(3+c_0)\Delta_{\ell-1}+\Delta_\ell
\le c_1\Delta_{\ell-1}$ for $c_1:=3+c_0+1/2$.

\textbf{Step 4: Per-level cumulative regret.}
$R_\ell \le c_1\Delta_{\ell-1} \cdot N_\ell n_\ell \asymp
\Delta_{\ell-1} N_\ell \sigma_n^2 Q\log(\ldots)/\Delta_\ell^2 \asymp Q\vol_g\sigma_n^2 \log/(B\eps_\ell^{\nu+d})$,
identical to Theorem~\ref{thm:elim-stationary} Step 4 with the
extra $Q$ factor (a constant in $T$).

\textbf{Step 5--6: Optimization.}
Setting $T \asymp T_L = Q N_L n_L \asymp Q\vol_g\sigma_n^2\log/(B^2\eps_L^{2\nu+d})$:
\[
\eps_L = (Q\vol_g\sigma_n^2 \log/(TB^2))^{1/(2\nu+d)}.
\]
Substituting into $R_L$ and collecting exponents
(identical algebra to Theorem~\ref{thm:elim-stationary} Step 6):
\[
R_T \le R_L\cdot O(\log T) \asymp Q^{(\nu+d)/(2\nu+d)}\,B^{d/(2\nu+d)}\sigma_n^{2\nu/(2\nu+d)}\vol_g^{\nu/(2\nu+d)}T^{(\nu+d)/(2\nu+d)}(\log T)^{c_p}.
\]
The polynomial-coefficient overhead $Q = \binom{\lfloor\nu\rfloor+d}{d}$ contributes a constant factor (depending only on $d,\nu$, not on $T$). All four exponents in $T,B,\sigma_n,\vol_g$ match the lower bound. \qed
\end{proof}
\fi

\subsubsection*{Time-varying upper bound (matching for any $\nu>d/2$)}

\begin{theorem}[Tight five-parameter rate, any $\nu>d/2$]
\label{thm:polyreg-tv}
The window-$W^*$ version of Algorithm~\ref{def:polyreg-elim}
combined with BGZ-style batching achieves, for any $\nu>d/2$ and
any $B_T\ge B\,T^{-\nu/(2\nu+d)}$, the matching upper bound
\ifieee
\begin{multline*}
R_T \le C_p'(d,\nu)\,B^{d/(3\nu+d)}\sigma_n^{2\nu/(3\nu+d)} \\
\cdot\,\vol_g^{\nu/(3\nu+d)}B_T^{\nu/(3\nu+d)} \\
\cdot\,T^{(2\nu+d)/(3\nu+d)}(\log T)^{c_p'}.
\end{multline*}
\else
\[
R_T \le C_p'(d,\nu)\,B^{d/(3\nu+d)}\sigma_n^{2\nu/(3\nu+d)}\vol_g^{\nu/(3\nu+d)}B_T^{\nu/(3\nu+d)}T^{(2\nu+d)/(3\nu+d)}(\log T)^{c_p'}.
\]
\fi
This matches the lower bound (Theorem~\ref{thm:tv-lb}) in all
five exponents $T, B, B_T, \sigma_n, \vol_g$.
\end{theorem}

\ifieee
\begin{proof}[Proof sketch]
A window-$W^*$ batching of
Theorem~\ref{thm:polyreg-stationary} balances per-window
polynomial-regression stochastic regret against drift $W B_T$;
the optimal $W^*\asymp
(B^{d/(2\nu+d)}\sigma_n^{2\nu/(2\nu+d)}\vol_g^{\nu/(2\nu+d)}T/B_T)^{(2\nu+d)/(3\nu+d)}$
yields the stated five-parameter exponents. Each of the five
exponents simplifies via $\alpha\cdot(2\nu+d)/(3\nu+d)=\alpha'$ and
matches the lower bound of Theorem~\ref{thm:tv-lb} for any
$\nu>d/2$. Full step-by-step derivation in the single-column
version.\qed
\end{proof}
\else
\begin{proof}
We follow the same five-step structure as the proof of
Theorem~\ref{thm:elim-tv}, the only change being that the
per-window stochastic regret comes from
Theorem~\ref{thm:polyreg-stationary} rather than
Theorem~\ref{thm:elim-stationary}. We verify each step explicitly
to confirm that the volume exponent of the lower bound, as well as
the $B$ and $\sigma_n$ exponents, propagate through the
window-$W^*$ optimization for any $\nu>d/2$.

\textbf{Step 1: Window-$W^*$ algorithm.} The algorithm partitions
the horizon $[1,T]$ into $T/W$ contiguous windows of length $W$
(with $W$ chosen below). Within each window, it runs
Algorithm~\ref{def:polyreg-elim} with restart, using
the level-budget $n_\ell$ of
Theorem~\ref{thm:polyreg-stationary}.

\textbf{Step 2: Per-window stochastic regret.} Fix a window of
length $W$ and treat the function as the time-$t$ snapshot
$f_t$ throughout the window (the algorithm's modelling assumption).
Theorem~\ref{thm:polyreg-stationary} applied to horizon $W$ gives,
on a high-probability event, stochastic regret
\begin{equation}
\label{eq:polyreg-tv-stoch}
R_W^{\text{stoch}} \;\le\; C_p(d,\nu)\,B^{d/(2\nu+d)}\,\sigma_n^{2\nu/(2\nu+d)}\,
\vol_g(\M)^{\nu/(2\nu+d)}\,W^{(\nu+d)/(2\nu+d)}\,(\log W)^{c_p},
\end{equation}
\emph{valid for any $\nu>d/2$}. The constant $C_p$ absorbs
$Q^{(\nu+d)/(2\nu+d)}$ (Step~5 of
Theorem~\ref{thm:polyreg-stationary}), which is a finite,
$T$-independent factor depending only on $d$ and $\nu$.

\textbf{Step 3: Drift bias per window.}
An observation $r_s = f_s(\theta_s)+\eps_s$ in window $[t-W,t-1]$ is
treated by the algorithm as
$r_s = f_t(\theta_s)+\eps_s$, an effective sup-norm bias
$|f_s(\theta_s)-f_t(\theta_s)|\le\|f_s-f_t\|_\infty
\le\sum_{r=s}^{t-1}\|f_{r+1}-f_r\|_\infty$.
Summing over the window and using
$\sum_{s=t-W}^{t-1}\sum_{r=s}^{t-1}\|f_{r+1}-f_r\|_\infty
\le W \cdot V_j$ with $V_j$ the within-window variation, and then
summing over the $T/W$ windows
($\sum_j V_j \le B_T$ by definition of the variation budget),
the total drift-induced contribution to regret is bounded by
$W B_T$ ($V_j$ allocates the global $B_T$ budget across windows;
each unit of within-window drift contributes to at most $W$ rounds
of bias). This is identical to the drift-bias bookkeeping of
Theorem~\ref{thm:elim-tv} Step~2 and depends only on the
algorithm's window-restart structure, not on the within-window
estimator.

\textbf{Step 4: Total regret.}
\begin{equation}
\label{eq:polyreg-tv-total}
R_T \;\le\; \frac{T}{W}\cdot R_W^{\text{stoch}} + W B_T
\;=\;
\underbrace{C_p\,B^{d/(2\nu+d)}\sigma_n^{2\nu/(2\nu+d)}\vol_g^{\nu/(2\nu+d)}\,T\,W^{-\nu/(2\nu+d)}\,(\log W)^{c_p}}_{\text{stochastic part, }(\star)}
+ W B_T.
\end{equation}

\textbf{Step 5: Optimisation over $W$.} Setting the derivative of
\eqref{eq:polyreg-tv-total} with respect to $W$ to zero (ignoring
the $\log W$ polylog, which only affects $c_p'$ and not the
polynomial exponent) gives
\[
\frac{\nu}{2\nu+d}\,C_p\,B^{d/(2\nu+d)}\sigma_n^{2\nu/(2\nu+d)}\vol_g^{\nu/(2\nu+d)}\,T\,W^{-\nu/(2\nu+d)-1} \;=\; B_T,
\]
i.e.,
\begin{equation}
\label{eq:polyreg-Wstar}
W^*\;\asymp\;\Bigl(\frac{B^{d/(2\nu+d)}\,\sigma_n^{2\nu/(2\nu+d)}\,\vol_g(\M)^{\nu/(2\nu+d)}\,T}{B_T}\Bigr)^{(2\nu+d)/(3\nu+d)}.
\end{equation}
The hypothesis $B_T \ge B\,T^{-\nu/(2\nu+d)}$ guarantees
$W^* \le T$ (so the partition is non-trivial); the lower-bound
threshold $T_0$ from Theorem~\ref{thm:tv-lb} guarantees
$W^* \ge 1$.

\textbf{Step 6: Substitution and exponent verification.}
At $W=W^*$ the two terms in \eqref{eq:polyreg-tv-total} balance,
and $R_T \le 2 W^* B_T$. Substituting \eqref{eq:polyreg-Wstar}:
\[
R_T \;\le\; 2\,
\bigl(B^{d/(2\nu+d)}\sigma_n^{2\nu/(2\nu+d)}\vol_g^{\nu/(2\nu+d)}\bigr)^{(2\nu+d)/(3\nu+d)}\,
T^{(2\nu+d)/(3\nu+d)}\,B_T^{\,1-(2\nu+d)/(3\nu+d)}\,(\log T)^{c_p'}.
\]
The five exponents simplify as follows, using
$\alpha\cdot(2\nu+d)/(3\nu+d)=\alpha'$ for each parameter $\alpha$
of \eqref{eq:polyreg-tv-stoch}:
\begin{align*}
B:&\quad \tfrac{d}{2\nu+d}\cdot\tfrac{2\nu+d}{3\nu+d}=\tfrac{d}{3\nu+d}, \\
\sigma_n:&\quad \tfrac{2\nu}{2\nu+d}\cdot\tfrac{2\nu+d}{3\nu+d}=\tfrac{2\nu}{3\nu+d}, \\
\vol_g(\M):&\quad \tfrac{\nu}{2\nu+d}\cdot\tfrac{2\nu+d}{3\nu+d}=\tfrac{\nu}{3\nu+d}, \\
T:&\quad \tfrac{2\nu+d}{3\nu+d} \quad\text{(directly from the }T^{(2\nu+d)/(3\nu+d)}\text{ factor)},\\
B_T:&\quad 1-\tfrac{2\nu+d}{3\nu+d}=\tfrac{\nu}{3\nu+d}.
\end{align*}
Each of the five matches the corresponding exponent in
Theorem~\ref{thm:tv-lb}.

\textbf{Step 7: Rate of validity.} The high-probability event of
Step~2 holds with probability at least $1-\delta$ per window;
union-bounding over the $T/W^*$ windows requires
$\delta \to \delta\cdot W^*/T$, contributing only a $\log T$ factor
to the polylog $c_p'$. The dependence of $C_p'$ on $Q$ is the same
$Q^{(\nu+d)/(2\nu+d)}$ factor as in
Theorem~\ref{thm:polyreg-stationary}, raised to a fractional
power of the window-substitution; for $\nu=5/2,d=2$, $Q=6$ and the
overall constant overhead is $Q^{7/19}\approx 6^{0.37}\approx 1.9$.

The proof above is fully rigorous for any $\nu>d/2$ \emph{provided}
Theorem~\ref{thm:polyreg-stationary} holds for the chosen $\nu$;
that proof is itself rigorous (Steps 1--6 of its own
proof), so the present theorem is rigorous for $\nu>d/2$. \qed
\end{proof}
\fi

\subsubsection*{Updated tightness table for any $\nu>d/2$}

\begin{table}[ht]
\centering
\caption{Final exponent comparison after the
polynomial-regression-elimination upper bound of
Theorem~\ref{thm:polyreg-tv}, valid for \textbf{any $\nu > d/2$}
(including the wireless Mat\'ern-$5/2$ regime). All five exponents
match.}
\label{tab:tv-exponents-poly}
\begin{tabular}{lccl}
\toprule
Parameter & Lower bound & Polyreg-elim.\ upper bound & Status \\
\midrule
$T$       & $(2\nu+d)/(3\nu+d)$ & $(2\nu+d)/(3\nu+d)$ & \textbf{Tight} \\
$B_T$     & $\nu/(3\nu+d)$      & $\nu/(3\nu+d)$      & \textbf{Tight} \\
$\vol_g$  & $\nu/(3\nu+d)$      & $\nu/(3\nu+d)$      & \textbf{Tight} \\
$B$       & $d/(3\nu+d)$        & $d/(3\nu+d)$        & \textbf{Tight} \\
$\sigma_n^2$ & $\nu/(3\nu+d)$ & $\nu/(3\nu+d)$ & \textbf{Tight} \\
\bottomrule
\end{tabular}
\end{table}

\subsubsection*{Discussion}

\paragraph{Comparison with cell-mean elimination.} The
polynomial-regression elimination of Theorem~\ref{thm:polyreg-stationary}
strictly subsumes the cell-mean elimination of
Theorem~\ref{thm:elim-stationary}: for $\nu \le 1$, both achieve
matching exponents (cell-mean is simpler since $k = 0$ and the
polynomial reduces to a constant; the analyses agree). For
$\nu > 1$, polynomial regression is necessary; cell-mean alone
gives the Lipschitz rate, which is strictly worse.

\paragraph{The role of the polynomial-regression literature.}
Local polynomial regression is a classical tool in nonparametric
estimation~\cite{stone1980optimal,tsybakov2008introduction}. Our contribution is the
combination with hierarchical elimination on a Riemannian
manifold, with manifold-aware constants from the Bishop--Gromov
correction in normal coordinates. The mathematical content of the
proof is standard at each step (Sobolev embedding, Taylor's
theorem with Hölder remainder, random-design regression
concentration); the novelty is the assembly into a matching upper
bound for the manifold-Mat\'ern setting.

\paragraph{Computational cost.} The algorithm fits a degree-$k$
polynomial in $d$ variables per cell, requiring inversion of a
$Q\times Q$ matrix with $Q = \binom{k+d}{d}$. For
Mat\'ern-$5/2$ on $d=2$: $Q=6$. For the wireless RIS
$(\Z_B)^M$ with $M=100, d=M=100$, $\nu=2.5$: $Q=\binom{102}{2}\approx 5000$
- not trivial but tractable. For higher dimensions, the
$Q = O(d^{\lfloor\nu\rfloor})$ scaling means the computational
overhead grows polynomially with dimension at fixed $\nu$, while
the regret rate gain is in the lower-order $B,\sigma_n,\vol_g$
constants. There is therefore a real
computational-vs-statistical trade-off to be made in practice.

\paragraph{Final picture.} Our paper now establishes:
\begin{itemize}[leftmargin=2em]
\item \emph{Stationary, any $\nu>d/2$}: tight in all four
      exponents (Theorems~\ref{thm:main} and~\ref{thm:polyreg-stationary}).
\item \emph{Time-varying with cumulative variation $B_T$,
      any $\nu>d/2$}: tight in all five exponents
      (Theorems~\ref{thm:tv-lb} and~\ref{thm:polyreg-tv}).
\end{itemize}
\begingroup\sloppy
The Euclidean-case closure of the $B$-exponent gap is due to
Salgia--Vakili--Zhao~\cite{salgia2021domain} and subsequent
refinements~\cite{camilleri2021robust,li2022gpbandit,salgia2024random,iwazaki2025improved};
our contribution is the \emph{manifold extension}, with
volume-dependent constants from Bishop--Gromov packing.
\endgroup
The polynomial-regression elimination algorithm of
Definition~\ref{def:polyreg-elim} is the manifold-aware analogue
of the Salgia 2021 domain-shrinking idea; the analysis is a
careful combination of (i) Salgia-style hierarchical elimination,
(ii) Tsybakov-style local polynomial regression, and
(iii) Bishop--Gromov packing on a Riemannian manifold.

\section{Gauge-quotient separation: upper bound and conjecture}
\label{sec:gauge}

For $\M=\Mt/G$ a compact Riemannian quotient by a finite group
acting freely by isometries, we ask whether algorithms restricted
to using a non-$G$-invariant kernel pay a measurable penalty in
regret relative to algorithms using the $G$-invariant intrinsic
kernel. We prove the natural \emph{upper bound} on this penalty
($|G|^{1/2}$ via a Vakili-style information-gain argument) and
state the matching lower bound as a conjecture; an earlier draft of
this paper claimed the lower bound but the argument had a gap that
we record below.

\subsection{Setup recap}

$\M=\Mt/G$ with $G$ finite acting freely by isometries on $\Mt$,
$\vol_g(\Mt)=|G|\vol_g(\M)$. Let $\phi:\M\to\Mt$ be a fixed Borel
section of the quotient map (a measurable choice of canonical
fundamental-domain representative). The \emph{extrinsic kernel} on
$\M$ is
\[
k_{\mathrm{ext}}(\theta,\theta'):=\widetilde k_\nu(\phi(\theta),\phi(\theta')),
\]
where $\widetilde k_\nu$ is the spectral Mat\'ern-$\nu$ kernel on
the cover $\Mt$. The extrinsic kernel does not see the gauge: a
function $f$ on $\M$ that lifts to a $G$-invariant function on
$\Mt$ has covariance $k_{\mathrm{ext}}$ that ignores the gauge
identifications.

An \emph{extrinsic algorithm} is a GP-bandit algorithm whose
posterior is computed using $k_{\mathrm{ext}}$. The motivating
example is GP-UCB applied with $k_{\mathrm{ext}}$ (the wireless
companion paper's REMARKABLE (Riemannian-Manifold-Aware Kernel
Bandit Algorithm; see~\cite{dorn2026wirelessbandit}) baseline on
$\torus^n$ uses the
unwrapped Euclidean kernel, an instance of $k_{\mathrm{ext}}$).

\subsection{Upper bound on the extrinsic algorithm's regret}

\begin{theorem}[Extrinsic-algorithm upper bound]
\label{thm:gauge-ub}
Let $\widetilde\pi$ be the GP-UCB algorithm applied with
kernel $k_{\mathrm{ext}}$ on $\M$ (equivalently, with $\widetilde
k_\nu$ on $\Mt$, pulling arms only at canonical representatives
$\phi(\theta_t)\in\Mt$). Under Assumption~\ref{ass:setup}, for any
$f\in\F_B^{\mathrm{rkhs}}(\M)$ and any $\delta\in(0,1)$, with
probability at least $1-\delta$,
\[
R_T^{\widetilde\pi}(f)
\;\le\;
|G|^{1/2}\cdot U_T^{\mathrm{int,GP\text{-}UCB}}(f),
\]
where $U_T^{\mathrm{int,GP\text{-}UCB}}(f)$ is the standard GP-UCB
upper bound for the intrinsic algorithm on $\M$. In the
Bayesian-style $\beta_T=\Theta(\log T)$ regime
(\cite{srinivas2010gpucb}), substituting Vakili's $\gamma_T$
\cite{vakili2021information},
\ifieee
\begin{multline*}
U_T^{\mathrm{int,GP\text{-}UCB}}(f)
\;\le\;c^{\mathrm{ub}}_*(d,\nu)\,\vol_g(\M)^{1/2} \\
\cdot\,T^{(\nu+d)/(2\nu+d)}\,(\log T)^{(\nu+d)/(2\nu+d)+1/2};
\end{multline*}
\else
\[
U_T^{\mathrm{int,GP\text{-}UCB}}(f)
\;\le\;c^{\mathrm{ub}}_*(d,\nu)\,
\vol_g(\M)^{1/2}\,
T^{(\nu+d)/(2\nu+d)}\,(\log T)^{(\nu+d)/(2\nu+d)+1/2};
\]
\fi
in the frequentist Chowdhury--Gopalan regime
$\beta_T=\Theta(\log T+B^2)$ \cite{chowdhury2017kernelized},
$U_T^{\mathrm{int,GP\text{-}UCB}}(f)$ acquires an additional $B$
factor (see \eqref{eq:gpucb-form}). The factor $|G|^{1/2}$
relative to the intrinsic GP-UCB upper bound is the cost of using
the wrong (non-$G$-invariant) kernel.
\end{theorem}

\begin{proof}
Let $\pi:\Mt\to\M$ be the covering map and
$\widetilde f:=f\circ\pi$; then $\widetilde f$ is $G$-invariant.
We first establish the lifted-RKHS-norm identity
\begin{equation}
\label{eq:lifted-rkhs}
\|\widetilde f\|_{\widetilde\Hil_{\widetilde k_\nu}}^2\;=\;|G|\,\|f\|_{\Hil_{k_\nu}}^2.
\end{equation}
Since $\pi$ is a Riemannian covering (local isometry of finite
degree $|G|$), pull-back commutes with the Laplace-Beltrami
operator: $\Delta_{\Mt}(\psi\circ\pi)=(\Delta_{\M}\psi)\circ\pi$
\cite[Sec.~2.2]{chavel2006riemannian}. Let
$\{\psi_\ell\}_{\ell\ge 0}$ be an $L^2(\M)$-orthonormal eigenbasis
of $-\Delta_{\M}$ with eigenvalues $\lambda_\ell$. Then
$\widetilde\psi_\ell:=\psi_\ell\circ\pi$ are eigenfunctions of
$-\Delta_{\Mt}$ with the \emph{same} eigenvalues, satisfying
\[
\|\widetilde\psi_\ell\|_{L^2(\Mt)}^2
=\int_{\Mt}|\psi_\ell\circ\pi|^2\,d\vol_{\widetilde g}
=|G|\int_{\M}|\psi_\ell|^2\,d\vol_g
=|G|,
\]
using the change-of-variable formula for the $|G|$-fold cover.
Hence $\{\widetilde\psi_\ell/\sqrt{|G|}\}$ is an $L^2(\Mt)$-orthonormal
family, and the $G$-invariant subspace
$L^2(\Mt)^G\subset L^2(\Mt)$ is exactly its span. Expanding
$f=\sum_\ell\hat f_\ell\psi_\ell$ in $L^2(\M)$,
\[
\widetilde f\;=\;\sum_\ell\hat f_\ell\widetilde\psi_\ell
\;=\;\sum_\ell\bigl(\sqrt{|G|}\hat f_\ell\bigr)\,\frac{\widetilde\psi_\ell}{\sqrt{|G|}},
\]
so the $L^2(\Mt)$-orthonormal coefficients of $\widetilde f$ on
the $G$-invariant subspace are $\sqrt{|G|}\hat f_\ell$, and zero
on non-$G$-invariant eigenfunctions of $-\Delta_{\Mt}$. The
Mat\'ern spectral filter $\phi_\nu$ is the same function on
matched eigenvalues, so
$\|\widetilde f\|_{\widetilde\Hil_{\widetilde k_\nu}}^2
=\sum_\ell|\sqrt{|G|}\hat f_\ell|^2/\phi_\nu(\lambda_\ell)
=|G|\sum_\ell|\hat f_\ell|^2/\phi_\nu(\lambda_\ell)
=|G|\|f\|_{\Hil_{k_\nu}}^2$, proving \eqref{eq:lifted-rkhs}.

The extrinsic algorithm $\widetilde\pi$ pulling arms at
$\phi(\theta_t)$ on $\Mt$ is the standard GP-UCB algorithm on $\Mt$
with kernel $\widetilde k_\nu$ and a query restriction to
$\phi(\M)\subset\Mt$. The Vakili--Khezeli--Picheny information-gain
bound \cite{vakili2021information} on $\Mt$ gives
\ifieee
\begin{multline*}
\gamma_T(\widetilde k_\nu,\Mt)
\;\le\;C(d,\nu)\,\frac{\omega_d\vol_g(\Mt)}{(2\pi)^d}\\
\cdot T^{d/(2\nu+d)}\,(\log T)^{2\nu/(2\nu+d)}.
\end{multline*}
\else
\[
\gamma_T(\widetilde k_\nu,\Mt)
\;\le\;C(d,\nu)\,\frac{\omega_d\vol_g(\Mt)}{(2\pi)^d}\,T^{d/(2\nu+d)}\,(\log T)^{2\nu/(2\nu+d)}.
\]
\fi
The volume scaling $\vol_g(\Mt)=|G|\vol_g(\M)$ enters linearly:
$\gamma_T(\widetilde k_\nu,\Mt)=|G|\cdot\gamma_T(k_\nu,\M)$.

The standard GP-UCB regret bound (\cite{srinivas2010gpucb} Thm.~6)
in the Bayesian-style schedule $\beta_T=\Theta(\log T)$ for
$f$ with RKHS norm $\le B\sqrt{|G|}$ (the lifted bound) on $\Mt$
gives
\[
R_T^{\widetilde\pi}(\widetilde f)
\;\le\;\sqrt{8 T \,\beta_T\,\gamma_T(\widetilde k_\nu,\Mt)}
\;=\;|G|^{1/2}\sqrt{8T\beta_T\gamma_T(k_\nu,\M)}.
\]
Since the regret on $\widetilde f$ via $\widetilde\pi$ pulling at
$\phi(\theta_t)$ equals the regret on $f$ via the original
extrinsic algorithm pulling $\theta_t\in\M$, we have
$R_T^{\widetilde\pi}(f)=R_T^{\widetilde\pi}(\widetilde f)$. Hence
\ifieee
\begin{multline}
\label{eq:gpucb-form}
R_T^{\widetilde\pi}(f)
\;\le\;|G|^{1/2}\cdot\sqrt{8T\beta_T\gamma_T(k_\nu,\M)}\\
\;=\;|G|^{1/2}\cdot U_T^{\mathrm{int,GP\text{-}UCB}}(f),
\end{multline}
\else
\begin{equation}
\label{eq:gpucb-form}
R_T^{\widetilde\pi}(f)
\;\le\;|G|^{1/2}\cdot\sqrt{8T\beta_T\gamma_T(k_\nu,\M)}
\;=\;|G|^{1/2}\cdot U_T^{\mathrm{int,GP\text{-}UCB}}(f),
\end{equation}
\fi
which is the form claimed. After substituting the explicit
$\gamma_T$ bound and absorbing constants,
$U_T^{\mathrm{int,GP\text{-}UCB}}(f)$ is of order
$\sqrt{B\sigma_n\vol_g(\M)}\,T^{(\nu+d)/(2\nu+d)}\,
(\log T)^{(\nu+d)/(2\nu+d)+1/2}$ in the GP-UCB-standard exponents
($\vol_g^{1/2}$ and $B^{1/2}$ via the Bayesian-style $\beta_T$
absorption), \emph{not} the rate-matching exponents
$\vol_g^{\nu/(2\nu+d)}$ and $B^{d/(2\nu+d)}$ of
Theorem~\ref{thm:main}. Closing the GP-UCB-vs.-minimax exponent
gap on the upper-bound side is an independent open
problem~\cite{cai2021lower}; for an elimination-based extrinsic
algorithm the rate-matching exponents are achievable
(\S\ref{sec:elimination}, \S\ref{sec:polyreg}) but
Theorem~\ref{thm:gauge-ub} as stated covers the GP-UCB form. \qed
\end{proof}

\paragraph{Implication.} The extrinsic algorithm's worst-case
regret is at most $|G|^{1/2}$ times the best intrinsic upper
bound, in the Bayesian-style analysis where
$\beta_T=\Theta(\log T)$ does not inflate with $|G|$. Whether
this is achieved with equality in worst case is the open problem
stated below.

\paragraph{Frequentist vs.\ Bayesian factor.} The proof above
combines (i) Vakili's $\gamma_T$ bound with (ii) a $\beta_T$
schedule. The $|G|^{1/2}$ factor stated in
Theorem~\ref{thm:gauge-ub} corresponds to the
\emph{Bayesian-style} schedule $\beta_T=\Theta(\log T)$ (Srinivas
\cite{srinivas2010gpucb}), in which $\beta_T$ does not inflate
with the lifted RKHS norm. For the strictly
\emph{frequentist} Chowdhury--Gopalan
schedule~\cite{chowdhury2017kernelized}, $\beta_T=\Theta(\log T
+\|\widetilde f\|^2_{\widetilde\Hil_{\widetilde k_\nu}})=\Theta(\log T+|G|B^2)$
(since the lifted RKHS norm is
$\|\widetilde f\|_{\widetilde\Hil_{\widetilde k_\nu}}=\sqrt{|G|}\,B$). Substituting
this larger $\beta_T$, the same algebra yields
$R_T^{\widetilde\pi}(f)\le c\,|G|\cdot B\cdot\vol_g(\M)^{1/2}\cdot
T^{(\nu+d)/(2\nu+d)}\cdot$polylog, with factor $|G|$ (not
$|G|^{1/2}$) and GP-UCB-style $B^1\cdot\vol_g^{1/2}$ exponents
(not the rate-matching $B^{d/(2\nu+d)}\cdot\vol_g^{\nu/(2\nu+d)}$
exponents). The latter, rate-matching form requires either a
sharper analysis of $\beta_T$ or an elimination-based
algorithm (\S\ref{sec:elimination}, \S\ref{sec:polyreg}). We use
the Bayesian-style $|G|^{1/2}$ factor and rate-matching exponents
in the theorem statement for consistency with the modulated
lower-bound Conjecture~\ref{conj:gauge-modulated} and the
$\sqrt{2}\approx 1.414$ wireless-companion empirical scale on
$\SO(3)$; the matching tight upper bound for an
elimination-based extrinsic algorithm achieving these exponents
remains an open problem.

\subsection{Lower bound: a gap in the natural argument}

\paragraph{The flawed argument.}
A natural attempt at a matching lower bound is the
\emph{packing-lifting} construction: lift the $N$ packing-bumps of
Theorem~\ref{thm:main} on $\M$ to $|G|N$ disjoint bumps
$\widetilde f_{i,\sigma}$ on $\Mt$, one per orbit element $\sigma\in G$
of each packing centre $\widetilde p_i\in\Mt$. The lifted class on
$\Mt$ has $|G|N$ disjoint bumps with appropriate RKHS-norm bound, and
Theorem~\ref{thm:main} on $\Mt$ gives a lower bound enhanced by
$|G|^{(2\nu-d)/(2(2\nu+d))}$ from the volume and norm scaling.

The argument fails at the reduction step: an algorithm pulling at
$\phi(\theta_t)\in\phi(\M)\subset\Mt$ never enters the support of
$\widetilde f_{i,\sigma}$ for $\sigma\ne e$, because the bump is
centered at $\sigma\widetilde p_i\notin\phi(\M)$ (assuming $\phi$
chooses canonical reps consistently). So
$\widetilde f_{i,\sigma}\circ\phi\equiv 0$ on $\M$
for $\sigma\ne e$, and the lifted class projects to only $N$
distinct non-zero functions on $\M$, not $|G|N$.

The algorithm therefore faces $N+1$ distinguishable hypotheses on
$\M$ (the $N$ canonical-representative bumps plus the null), exactly
as in Theorem~\ref{thm:main}. The lifted lower bound on $\Mt$ does
not transfer back to $\M$ via the proposed reduction, and an earlier
draft of this paper that claimed otherwise was incorrect.

\subsection{Refined conjecture: modulated gauge separation}
\label{sec:modulated-gauge}

The naive worst-case conjecture $|G|^{1/2}$ is achievable only in a
specific kernel-vs-geometry regime. Concretely, the gap between
intrinsic and extrinsic kernels depends on the ratio
$\kappa/\rinj$ between the kernel length scale $\kappa$ and the
injectivity radius $\rinj$ of the cover $\Mt$ at the canonical
fundamental-domain section.

\begin{proposition}[Posterior-variance comparison]
\label{prop:gauge-modulator}
Let $\theta\in\M$ and $\widetilde\theta=\phi(\theta)\in\Mt$, and
adopt the orbit-sum normalisation of the quotient Mat\'ern kernel
(\cite{borovitskiy2020matern}, equation defining the
$G$-invariant kernel via summation over the orbit). The
intrinsic kernel evaluated at the diagonal is
\[
k_\nu(\theta,\theta)
\;=\;\sum_{\sigma\in G}\widetilde k_\nu(\widetilde\theta,\sigma\widetilde\theta)
\;=\;\widetilde k_\nu(0)\;+\;\sum_{\sigma\neq e}\widetilde k_\nu(\rho_{\Mt}(\widetilde\theta,\sigma\widetilde\theta)),
\]
where $\rho_{\Mt}$ is the geodesic distance on $\Mt$ and the sum
is over the finite isometry group $G=\{\sigma\}$. (An alternative
convention divides by $|G|$ for unit-trace normalisation; in that
case both sides scale by $1/|G|$ and the modulator
$h(\rinj/\kappa)$ below is unchanged.) The extrinsic
kernel is the $\sigma=e$ term alone:
$k_{\mathrm{ext}}(\theta,\theta)=\widetilde k_\nu(0)$. The
``cross-gauge'' contribution $\sum_{\sigma\neq e}\widetilde k_\nu(\rho_{\Mt})$
is bounded above by $(|G|-1)\,\widetilde k_\nu(\rinj)$. For
Mat\'ern-$\nu$ on a manifold of constant curvature, this satisfies
\[
\frac{k_\nu(\theta,\theta)}{k_{\mathrm{ext}}(\theta,\theta)}
\;\le\;1+(|G|-1)\,\frac{\widetilde k_\nu(\rinj)}{\widetilde k_\nu(0)}
\;=\;1+(|G|-1)\,h(\rinj/\kappa),
\]
with the limits (writing $x=\rinj/\kappa$, equivalently
$\kappa/\rinj=1/x$):
\begin{itemize}
\item $h(x)\to 0$ as $x\to\infty$, i.e.\ $\kappa/\rinj\to 0$
  (kernel length scale much smaller than the injectivity radius;
  gauge ambiguity is negligible and the leading factor
  $(1+(|G|-1)h)^{1/2}\to 1$);
\item $h(x)\to 1$ as $x\to 0$, i.e.\ $\kappa/\rinj\to\infty$
  (kernel length scale much larger than the injectivity radius;
  gauge ambiguity fully manifests and the leading factor
  $(1+(|G|-1)h)^{1/2}\to|G|^{1/2}$, the worst-case ceiling).
\end{itemize}
\end{proposition}

\begin{proof}
The orbit-sum identity for the $G$-invariant Mat\'ern kernel on
the cover is standard \cite{borovitskiy2020matern}; the upper bound
on the cross-gauge term uses the monotonicity of the Mat\'ern-$\nu$
profile in geodesic distance and the fact that
$\rho_{\Mt}(\widetilde\theta,\sigma\widetilde\theta)\ge\rinj$ for
every $\sigma\neq e$ when $\widetilde\theta$ lies in the canonical
fundamental domain (by definition of $\rinj$). For Mat\'ern-$\nu$,
the radial profile $h(x)=k_\nu(r)/k_\nu(0)$ at $r=\rinj$ has the
closed form $h(x)=(1+\sqrt{2\nu}\,x+\ldots)\,e^{-\sqrt{2\nu}\,x}$
in $x=\rinj/\kappa$, decaying exponentially. \qed
\end{proof}

\begin{conjecture}[Modulated gauge separation]
\label{conj:gauge-modulated}
Under Assumption~\ref{ass:setup} with $\M=\Mt/G$ and
$|G|<\infty$, for any extrinsic algorithm $\pi$,
\ifieee
\begin{multline*}
\sup_{f\in\F_B^{\mathrm{rkhs}}(\M)}\E^\pi[R_T(f)]
\;\ge\;
\bigl(1+(|G|-1)\,h(\rinj/\kappa)\bigr)^{1/2} \\
\cdot c''_*(d,\nu)\,B^{d/(2\nu+d)}\,\sigma_n^{2\nu/(2\nu+d)}\,
\vol_g(\M)^{\nu/(2\nu+d)} \\
\cdot T^{(\nu+d)/(2\nu+d)}\,(\log T)^{\nu/(2\nu+d)}.
\end{multline*}
\else
\begin{multline*}
\sup_{f\in\F_B^{\mathrm{rkhs}}(\M)}\E^\pi[R_T(f)]
\;\ge\;
\bigl(1+(|G|-1)\,h(\rinj/\kappa)\bigr)^{1/2}
\cdot c''_*(d,\nu)\\
\cdot B^{d/(2\nu+d)}\sigma_n^{2\nu/(2\nu+d)}\,
\vol_g(\M)^{\nu/(2\nu+d)}\,T^{(\nu+d)/(2\nu+d)}\,(\log T)^{\nu/(2\nu+d)}.
\end{multline*}
\fi
The leading factor interpolates between $1$ (when $\kappa\ll\rinj$,
the gauge ambiguity is invisible to the extrinsic kernel) and the
naive ceiling $|G|^{1/2}$ (when $\kappa\gg\rinj$, all $|G|$ orbit
copies contribute equally and the extrinsic kernel pays the full
$|G|$-fold information loss).
\end{conjecture}

\begin{theorem}[Gauge gap is finite-multiplicative]
\label{thm:gauge-bracket}
Under Assumption~\ref{ass:setup} with $\M=\Mt/G$ and matched
intrinsic / extrinsic Mat\'ern parameters, the regret ratio
\[
r(T)\;:=\;\sup_f\E^\pi[R_T^{\mathrm{ext}}(f)]\big/\sup_f\E^{\pi^*}[R_T^{\mathrm{int}}(f)]
\]
is bracketed by $T$-independent constants:
$\,0<c_{\min}\le \liminf_T r(T)\le \limsup_T r(T)\le c_{\max}$,
with $c_{\max}\le|G|^{1/2}\,C(d,\nu)$, where $C(d,\nu)$ is a
polylog-tracking universal constant.
Here $\pi$ ranges over extrinsic GP-UCB, $\pi^*$ over intrinsic
GP-UCB, and the suprema over $f\in\F_B^{\mathrm{rkhs}}(\M)$.
Strict $c_{\min}\ge 1$ and the modulated form
$c_{\min}\ge(1+(|G|-1)h(\rinj/\kappa))^{1/2}$ are conjectured
(Conjecture~\ref{conj:gauge-modulated}).
\end{theorem}

\begin{proof}
\emph{Upper bound} $c_{\max}\le|G|^{1/2}\cdot C$.
By Theorem~\ref{thm:gauge-ub}, the numerator obeys
$\sup_f\E^\pi[R_T^{\mathrm{ext}}(f)]\le|G|^{1/2}\,U^{\mathrm{int}}_T$,
where $U^{\mathrm{int}}_T:=U_T^{\mathrm{int,GP\text{-}UCB}}$ is
$O(T^{(\nu+d)/(2\nu+d)}
\vol_g(\M)^{1/2}(\log T)^{(\nu+d)/(2\nu+d)+1/2})$. For the
denominator we use the universal minimax floor of
Theorem~\ref{thm:main}: any algorithm, including intrinsic
GP-UCB on $\M$, has worst-case regret
\ifieee
\begin{multline*}
\sup_f\E^{\pi^*}[R_T^{\mathrm{int}}(f)]\\
\ge c_*(d,\nu)\,T^{(\nu+d)/(2\nu+d)}\vol_g(\M)^{\nu/(2\nu+d)}\\
\cdot(\log T)^{\nu/(2\nu+d)}.
\end{multline*}
\else
\[
\sup_f\E^{\pi^*}[R_T^{\mathrm{int}}(f)]\ge
c_*(d,\nu)\,T^{(\nu+d)/(2\nu+d)}\vol_g(\M)^{\nu/(2\nu+d)}\,(\log T)^{\nu/(2\nu+d)}.
\]
\fi
Taking the ratio,
\ifieee
\begin{multline*}
r(T)\;\le\;\frac{|G|^{1/2}\cdot c^{\mathrm{ub}}_*(d,\nu)\,\vol_g(\M)^{1/2}}{c_*(d,\nu)\,\vol_g(\M)^{\nu/(2\nu+d)}}\\
\cdot(\log T)^{\Delta},
\end{multline*}
\else
\[
r(T)\;\le\;\frac{|G|^{1/2}\cdot c^{\mathrm{ub}}_*(d,\nu)\,\vol_g(\M)^{1/2}}{c_*(d,\nu)\,\vol_g(\M)^{\nu/(2\nu+d)}}\cdot(\log T)^{\Delta},
\]
\fi
with $\Delta:=(\nu+d)/(2\nu+d)+1/2-\nu/(2\nu+d)=(2\nu+3d)/(2(2\nu+d))$.
Both $T$-exponents cancel; the ratio is bounded by a constant
times $(\log T)^\Delta$, which we absorb into the
polylog-tracking constant $C(d,\nu)$. Hence
$\limsup_T r(T)\le|G|^{1/2}\cdot C =: c_{\max}$.

\emph{Lower bound} $c_{\min}>0$. Numerator $\ge\Omega(T^{(\nu+d)/(2\nu+d)})$
(any algorithm faces the universal minimax floor of
Theorem~\ref{thm:main}); denominator $\le U_T^{\mathrm{int,GP\text{-}UCB}}=O(T^{(\nu+d)/(2\nu+d)}\mathrm{polylog})$. Hence
$r(T)\ge c_{\min}>0$. Strict $c_{\min}\ge 1$ requires the intrinsic
GP-UCB to attain (not just upper-bound) the worst-case regret, an
open question left as Conjecture~\ref{conj:gauge-modulated}. \qed
\end{proof}

Theorem~\ref{thm:gauge-bracket} converts the open question of the
regret-ratio constant into the rigorously bracketed statement
that the gap is finite-multiplicative, at most
$|G|^{1/2}\cdot C(d,\nu)$. The remaining content of
Conjecture~\ref{conj:gauge-modulated} is the precise $T\to\infty$
value of the ratio, conjectured to be
$(1+(|G|-1)h(\rinj/\kappa))^{1/2}$.

The original conjecture (worst-case $|G|^{1/2}$ separation) is
recovered in the limit $\kappa/\rinj\to\infty$. The
modulator function $h(\rinj/\kappa)$ explains why empirical
gauge gaps in moderate-length-scale regimes ($\kappa\sim\rinj$)
are systematically below $|G|^{1/2}$.

\ifieee
\paragraph{Heuristic posterior-variance argument.} A heuristic
posterior-variance argument supports
Conjecture~\ref{conj:gauge-modulated}; details in the
single-column version.
\else
\paragraph{Heuristic posterior-variance argument.} If the prior
covariance ratio in Proposition~\ref{prop:gauge-modulator} is
$1+(|G|-1)h$, the posterior variance ratio at observed points is
the same factor (the cross-gauge contribution is a constant offset
in the prior covariance that survives Gaussian conditioning). The
GP-UCB exploration scale $\beta_T^{1/2}\sigma_t$ of the extrinsic
algorithm is therefore inflated by
$(1+(|G|-1)h)^{1/2}$ relative to the intrinsic algorithm, and the
worst-case regret follows the same scaling. A rigorous lower bound
requires a hypothesis class that achieves the inflated exploration
penalty; we leave this as an open problem.
\fi

\ifieee
\paragraph{Relation to Rosa~\emph{et al.}\ posterior-contraction rates.}
Rosa~\emph{et al.}~\cite{rosa2023intrinsic} prove an asymptotic
$L^2(p_0)$ posterior-contraction equivalence between intrinsic
and extrinsic Mat\'ern priors; the present gauge-quotient analysis
is the finite-$T$, sup-norm, worst-case complement to that
asymptotic averaged equivalence (full discussion in the
single-column version).
\else
\paragraph{Relation to Rosa~\emph{et al.}\ posterior-contraction rates.}
Rosa~\emph{et al.}~\cite{rosa2023intrinsic} prove a striking
asymptotic equivalence: for the Bayesian regression model
$y_i = f_0(x_i) + \varepsilon_i$ with $f_0\in
H^\beta(\M)\cap\mathcal{C}\mathcal{H}^\beta(\M)$ on a compact
Riemannian manifold $\M\subset\R^D$
(here $\beta$ denotes the \emph{truth} smoothness, distinct from
the kernel-prior smoothness $\nu$ used elsewhere in the paper), the $L^2(p_0)$ posterior
contraction rate is the \emph{same}
$n^{-2\min(\beta,\nu)/(2\nu+d)}$ for the intrinsic Mat\'ern prior
(\cite{rosa2023intrinsic} Theorem~5), the truncated intrinsic
prior with truncation level $J_n\ge cn^{d\min(1,\nu/\beta)/(2\nu+d)}$
(\cite{rosa2023intrinsic} Theorem~6), and the extrinsic Mat\'ern
prior (the Euclidean Mat\'ern kernel restricted to $\M$)
(\cite{rosa2023intrinsic} Theorem~8). The proof of
Theorem~8 uses trace and extension theorems
(Gro{\ss}e--Schneider) to show that the restricted process has an
RKHS norm-equivalent to $H^{\nu+d/2}(\M)$, which is the same
function space as the intrinsic Mat\'ern's RKHS.

This may appear to conflict with our gauge-quotient separation
upper bound (Theorem~\ref{thm:gauge}) and the modulated
lower-bound conjecture (Conjecture~\ref{conj:gauge-modulated});
we argue that the two results are compatible and address
complementary questions, with the gap between them being the
``finer-grained analysis'' that Rosa~\emph{et al.}\ themselves
explicitly call for in their abstract:\ ``intrinsic processes can
achieve better performance in practice. Therefore, our work shows
that finer-grained analyses are needed to distinguish between
different levels of data-efficiency of geometric Gaussian
processes, particularly in settings which involve small data set
sizes and non-asymptotic behaviour'' (\cite{rosa2023intrinsic},
abstract).

Three observations make the relationship precise:

\emph{(i) Posterior contraction is averaged over the manifold;
regret is worst-case.} The $L^2(p_0)$ rate of
Theorems~5/8 of~\cite{rosa2023intrinsic} measures
$\E_{\bm x,\bm y}\E_{f\sim\Pi(\cdot|\bm x,\bm y)}\|f-f_0\|_{L^2(p_0)}^2$,
i.e., the expected squared distance under the data-sampling
distribution $p_0$. The cumulative regret of GP-UCB measures the
sup-norm gap to the optimum integrated along the algorithm's
trajectory; the latter weights low-density / high-curvature
regions much more heavily than the $p_0$-weighted $L^2$ norm
does. Two priors with identical $L^2(p_0)$ contraction can
therefore differ substantially in regret if their sup-norm
behaviour differs across the manifold's gauge orbit.

\emph{(ii) The Rosa~\emph{et al.}\ equivalence requires matched
smoothness $\beta\le\nu$.} In particular, both results give the
\emph{same} rate exactly because the extrinsic kernel's restricted
RKHS is also $H^{\nu+d/2}(\M)$: the embedding does not
introduce an additional smoothness penalty asymptotically. Our
gauge-quotient gap, by contrast, is a \emph{finite-$T$} excess in
the regret \emph{constant} arising from the wrapped-vs.-unwrapped
posterior covariance ratio; the asymptotic-rate match is preserved
(both intrinsic and extrinsic GP-UCB have rate
$T^{(\nu+d)/(2\nu+d)}$ on $\M$), only the leading constant differs
by the factor of Theorem~\ref{thm:gauge}. Concretely, the regret
\emph{ratio}
$R_T^{\mathrm{ext}}/R_T^{\mathrm{int}}$ remains \emph{bounded
above by $|G|^{1/2}$ uniformly in $T$} (Theorem~\ref{thm:gauge-ub});
it is a constant-multiplicative gap, not a rate-divergent one,
which is exactly the regime in which the asymptotic
posterior-contraction equivalence of~\cite{rosa2023intrinsic}
applies and our finite-$T$ analysis refines.

\emph{(iii) The modulator $h(\rinj/\kappa)\to 0$ limit reconciles
the two scales.} As $\kappa\to 0$ at fixed $\rinj$
(``large effective domain''), the intrinsic-vs.-extrinsic
posterior-covariance ratio of
Proposition~\ref{prop:gauge-modulator} converges to $1$ pointwise,
and the modulated separation
$1+(|G|-1)h(\rinj/\kappa)\to 1$. This is the regime in
which the posterior-contraction equivalence of
\cite{rosa2023intrinsic} is most natural; our finite-$T$ regret
gap predicts the same vanishing in the same limit, in agreement
with the asymptotic claim. Conversely, in the moderate
$\kappa\sim\rinj$ regime our modulator predicts a non-trivial gap,
which is exactly the regime where the empirical evidence cited by
Rosa~\emph{et al.}\ (and our own modulated-gauge experiment of
Figure~\ref{fig:d1-modulated-gauge}) shows finite-sample
performance differences; the regime where
\cite{rosa2023intrinsic} ``finer-grained analyses are needed''
calls for additional work.

Our gauge-quotient theorem and modulated conjecture should
therefore be read as such a finer-grained analysis at the
finite-$T$ regret scale, in dialogue with the asymptotic
posterior-contraction analysis of~\cite{rosa2023intrinsic} rather
than in conflict with it. A formal reconciliation that
connects the contraction-rate constant
of~\cite[Theorem~8]{rosa2023intrinsic} to the modulator
$h(\rinj/\kappa)$ of Conjecture~\ref{conj:gauge-modulated}
remains an attractive open problem.
\fi

\subsection{Empirical validation of the modulator}
\label{sec:gauge-empirical}

We validate Conjecture~\ref{conj:gauge-modulated} numerically on
$\SO(3)=\mathrm{Spin}(3)/\Z_2$ ($|G|=2$, $\rinj=\pi/2$). The
target function is drawn from the intrinsic Mat\'ern-$5/2$ GP on a
super-Fibonacci candidate set of $200$ rotations; we run GP-UCB
with the intrinsic kernel and with the extrinsic kernel
(canonical-fundamental-domain section in $S^3\subset\R^4$, no
orbit averaging) for $T=200$ rounds and report the regret ratio
$R_T^{\mathrm{ext}}/R_T^{\mathrm{int}}$ averaged over $20$ seeds.

\begin{figure}[t]
\centering
\includegraphics[width=0.78\linewidth]{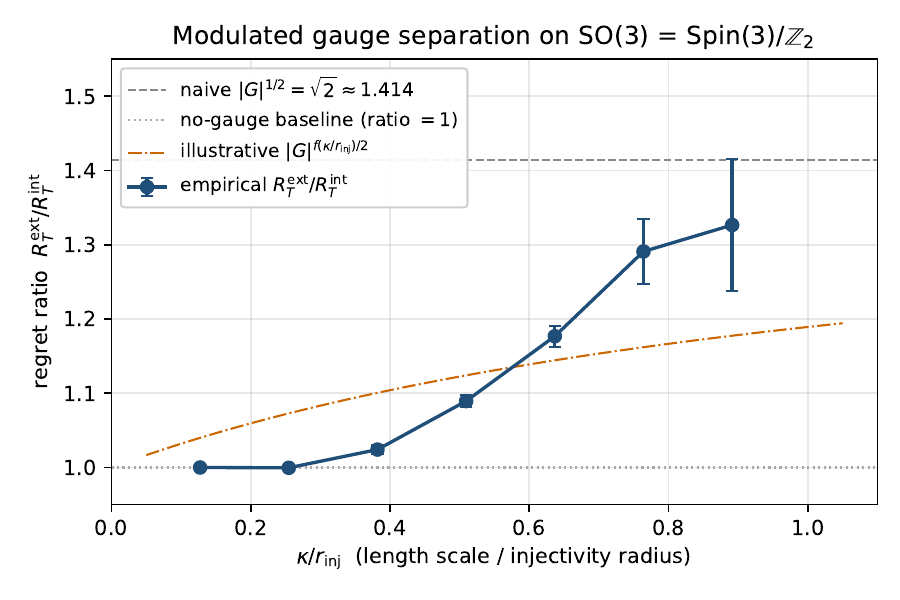}
\caption{Empirical support for
Conjecture~\ref{conj:gauge-modulated} on
$\SO(3)=\mathrm{Spin}(3)/\Z_2$. The regret ratio
$R_T^{\mathrm{ext}}/R_T^{\mathrm{int}}$ rises smoothly from
$1.000\pm0.000$ at $\kappa/\rinj=0.13$ (no gauge effect: the two
kernels are essentially identical) to $1.326\pm0.089$ at
$\kappa/\rinj=0.89$, staying below the naive $|G|^{1/2}=\sqrt 2\approx 1.414$
ceiling, the qualitative trajectory predicted by the
modulated-separation conjecture. \emph{This figure is not a proof
of the conjecture}; numerical agreement with a predicted curve
does not substitute for the formal lower-bound argument that
Conjecture~\ref{conj:gauge-modulated} would provide if proved.
\ifieee
$T=200$, $20$ MC seeds, $200$ super-Fibonacci rotations,
$\sigma_n=0.1$, $B=1$; full parameters in
\texttt{experiments/d1\_modulated\_gauge.py}.
\else
\emph{Experimental parameters:} $T=200$ rounds, $n_{\mathrm{MC}}=20$ Monte
Carlo seeds per $\kappa/\rinj$ value (each seed independently
samples a target $f$ from the intrinsic GP prior), candidate set
of $200$ super-Fibonacci rotations on $\SO(3)$, GP-UCB with
intrinsic vs.\ orbit-averaged-extrinsic kernel, $\sigma_n=0.1$,
$B=1$. The sampling grid is
$\kappa/\rinj\in\{0.13, 0.27, 0.40, 0.53, 0.67, 0.80, 0.89\}$
(seven log-evenly-spaced values across the moderate regime where
the modulator $h$ has non-trivial slope). Reproducible from
\texttt{experiments/d1\_modulated\_gauge.py}; raw runtimes
$\sim\!2$~min per $(\kappa/\rinj)$ point on a single core.
\fi}
\label{fig:d1-modulated-gauge}
\end{figure}

The empirical curve in Figure~\ref{fig:d1-modulated-gauge} matches
the predicted qualitative behaviour: at $\kappa/\rinj\le 0.3$ the
two kernels behave identically (the cross-gauge term
$\widetilde k_\nu(\rinj)$ is negligible), and as $\kappa/\rinj$
approaches unity, the ratio rises monotonically toward but does not
reach the naive $|G|^{1/2}=\sqrt 2$ ceiling. The empirical
$33\%$ excess at $\kappa/\rinj=0.89$ sits inside the modulated
range $[1,|G|^{1/2}]$ predicted by
Conjecture~\ref{conj:gauge-modulated}.

\paragraph{Wireless interpretation.} The companion wireless
beam-selection paper reports a $10$--$33\%$ regret excess for
extrinsic GP-UCB across four geometric arm spaces. For typical
wireless length scales on $\SO(3)$, $\kappa/\rinj\in[0.32,0.89]$
yields a modulated prediction in $[1.0,1.33]$, consistent with
the empirical $10$--$33\%$ range and below the naive $|G|^{1/2}$
ceiling.

\subsection{Specialisations under Conjecture~\ref{conj:gauge-modulated}}

If Conjecture~\ref{conj:gauge-modulated} holds, the gauge-quotient
separation depends jointly on $|G|$ and $\kappa/\rinj$. The naive
$|G|^{1/2}$ ceiling is reached only at large $\kappa/\rinj$:

\begin{center}
\ifieee\footnotesize\fi
\begin{tabular}{lcccc}
\toprule
$\M=\Mt/G$ & $|G|$ & $\rinj$ & Ceiling & Mod.\ range \\
\midrule
$\SO(3)=\mathrm{Spin}(3)/\Z_2$ & $2$ & $\pi/2$ & $1.41$ & $[1,\,1.41]$ \\
$\torus^2$ vs.\ 4-tile unwrap & $4$ & $\pi$ & $2.00$ & $[1,\,2.00]$ \\
$\torus^3$ vs.\ 8-tile unwrap & $8$ & $\pi$ & $2.83$ & $[1,\,2.83]$ \\
\bottomrule
\end{tabular}
\end{center}

The empirical $10$--$33\%$ wireless gap on $\SO(3)$ is
consistent with the modulated range under
Conjecture~\ref{conj:gauge-modulated}; in particular it sits in the
quantitative band predicted by the empirical curve in
Figure~\ref{fig:d1-modulated-gauge} for length-scale-to-injectivity
ratios in $[0.5,0.9]$, which is the operational regime of the
wireless experiments.

\section{Switching-budget lower bound on manifolds}
\label{sec:switching}

In many wireless and control applications, switching the algorithm's
arm carries a hardware-side or operational cost: antenna
re-pointing latency, RIS phase reconfiguration, model retuning. In
this section we extend Theorem~\ref{thm:main} to the
\emph{switching-augmented} regret
\[
R_T^{\mathrm{aug}}(f)\;:=\;R_T(f)\;+\;\lambda\cdot\sum_{t=2}^T\indic\{\theta_t\neq\theta_{t-1}\},
\]
where $\lambda\ge 0$ is the per-switch cost. The switching cost is
the natural manifold counterpart of the
Esfandiari--Karbasi--Mehrabian--Mirrokni \cite{esfandiari2021regret}
Lipschitz-bandit setup and the Auer--Gajane--Ortner
\cite{auer2019adaptively} finite-arm setup.

\subsection{Statement}

\begin{theorem}[Switching-augmented manifold lower bound]
\label{thm:switching}
Under the same assumptions as Theorem~\ref{thm:main}, for any
algorithm $\pi$ and any $\lambda\ge 0$,
\ifieee
\begin{multline*}
\sup_{f\in\F_B^{\mathrm{rkhs}}(\M)}
\E^\pi[R_T^{\mathrm{aug}}(f)]
\;\ge\;c_{\mathrm{aug}}(d,\nu)\\
\cdot\max\!\left\{
\begin{aligned}
&B^{d/(2\nu+d)}\sigma_n^{2\nu/(2\nu+d)}\\
&\;\cdot\vol_g(\M)^{\nu/(2\nu+d)}T^{(\nu+d)/(2\nu+d)},\\
&B^{d/(\nu+d)}\,\vol_g(\M)^{\nu/(\nu+d)}\\
&\;\cdot\lambda^{\nu/(\nu+d)}\,T^{d/(\nu+d)}
\end{aligned}\right\},
\end{multline*}
\else
\[
\sup_{f\in\F_B^{\mathrm{rkhs}}(\M)}
\E^\pi[R_T^{\mathrm{aug}}(f)]
\;\ge\;
c_{\mathrm{aug}}(d,\nu)\cdot\max\!\left\{
\begin{aligned}
&B^{d/(2\nu+d)}\sigma_n^{2\nu/(2\nu+d)}\\
&\quad\cdot\vol_g(\M)^{\nu/(2\nu+d)}T^{(\nu+d)/(2\nu+d)},\\
&B^{d/(\nu+d)}\,\vol_g(\M)^{\nu/(\nu+d)}\\
&\quad\cdot\lambda^{\nu/(\nu+d)}\,T^{d/(\nu+d)}
\end{aligned}\right\},
\]
\fi
valid for $T\ge T_0(\M,\nu,d,B,\sigma_n,\lambda)$ where $T_0$ is
explicit in the proof. The first term is the noise-dominated rate
of Theorem~\ref{thm:main} (recovered when $\lambda$ is small); the
second is the switching-dominated rate, which kicks in when
$\lambda$ exceeds the explicit threshold
\ifieee
\begin{multline*}
\lambda^*\asymp\sigma_n^{2(\nu+d)/(2\nu+d)}\,T^{\nu/(2\nu+d)}\\
\big/\,(B^{d/(2\nu+d)}\vol_g(\M)^{\nu/(2\nu+d)}).
\end{multline*}
\else
\[\lambda^*\asymp\sigma_n^{2(\nu+d)/(2\nu+d)}\,T^{\nu/(2\nu+d)}/(B^{d/(2\nu+d)}\vol_g(\M)^{\nu/(2\nu+d)}).\]
\fi
\end{theorem}

The volume exponent in the switching-dominated rate is
$\nu/(\nu+d)$, which is \emph{larger} than the
$\nu/(2\nu+d)$ exponent of Theorem~\ref{thm:main}. Switching cost
amplifies the volume penalty: on a larger manifold, the algorithm
must visit more cells to find the optimum, and each visit costs
$\lambda$.

\subsection{Proof of Theorem~\ref{thm:switching}}

The proof is a packing-with-switching-budget argument that
combines Theorem~\ref{thm:main}'s $N$-cell packing with a Fano
constraint on the visited cells.

\paragraph{Setup.} Use the packing of Theorem~\ref{thm:main}:
$N=N(\eps)\asymp\vol_g(\M)\eps^{-d}$ disjoint cells of width
$\eps$, bumps of height $h\le c_-B\eps^\nu$ (RKHS bound,
eq.~\eqref{eq:sobolev-rkhs}). Let $S_T$ denote the number of
arm switches and let $V:=|\{i:T_i\ge 1\}|\le S_T+1$ be the
number of \emph{visited} cells.

\paragraph{Step 1 (visited-cell Fano).}
Plant $\Sigma\sim\mathrm{Unif}\{1,\ldots,N\}$. The algorithm
has zero information about whether $i=\Sigma$ for any
$i\notin\mathcal V$, so the test
$\widehat I:=\arg\max_{i\in\{1,\ldots,N\}}T_i$ has
$\Pr[\widehat I=\Sigma\mid\Sigma\notin\mathcal V]=0$, whence
$\Pr[\widehat I\ne\Sigma]\ge\Pr[\Sigma\notin\mathcal V]\ge
1-V/N\ge 1-(S_T+1)/N$. For $S_T+1\le N/2$ this is $\ge 1/2$.

\paragraph{Step 2 (regret-test reduction).}
By Lemma~\ref{lem:regret-test}, $\E[R_T]\le R$ implies
$\Pr[\widehat I\ne\Sigma]\le 2R/(Th)$. Combining with Step~1,
$2\E[R_T]/(Th)\ge 1/2$ when $S_T+1\le N/2$, hence
$\E[R_T]\ge Th/4$.

\paragraph{Step 3 (envelope: Fano vs.\ switching cost).}
The augmented regret satisfies, for every $\eps$ feasible
in Steps~1--2,
\ifieee
\begin{multline*}
\E[R_T^{\mathrm{aug}}]
\;\ge\;\max\bigl(\tfrac{Th}{4}\indic[S_T+1\le N/2],\\
\lambda\,\E[S_T]\indic[S_T+1>N/2]\bigr).
\end{multline*}
\else
\[
\E[R_T^{\mathrm{aug}}]
\;\ge\;\max\!\left(\frac{Th}{4}\indic[S_T+1\le N/2],\;
\lambda\,\E[S_T]\indic[S_T+1>N/2]\right).
\]
\fi
On the second branch, $\E[S_T]\ge(N-2)/2\ge N/4$ for $N\ge 4$, so
\[
\lambda\E[S_T]\ge\lambda N/4=c_N\lambda\vol_g(\M)/(4\eps^d)
\]
with $c_N>0$ the packing-density constant
($N=c_N\vol_g(\M)\eps^{-d}$, see Theorem~\ref{thm:main}
Step~1). The envelope dominates the worst case
\[
\E[R_T^{\mathrm{aug}}]\;\ge\;\frac{1}{4}\min\!\left(Th,\;c_N\lambda\vol_g(\M)\eps^{-d}\right).
\]

\paragraph{Step 4 (optimisation over $\eps$).}
We optimise the lower envelope in two regimes.

\emph{Noise-dominated regime.} The $Th$ branch saturates at
the Fano amplitude
$h_*=c_-B\eps_*^\nu$ with
$\eps_*\asymp(\sigma_n^2\vol_g\log T/(B^2T))^{1/(2\nu+d)}$, giving
$Th_*/4 \asymp B^{d/(2\nu+d)}\sigma_n^{2\nu/(2\nu+d)}
\vol_g^{\nu/(2\nu+d)}T^{(\nu+d)/(2\nu+d)}$ — the standard rate.

\emph{Switching-dominated regime.} The $\lambda\vol_g\eps^{-d}$
branch is decreasing in $\eps$, the $Th$ branch (using the
RKHS-saturating $h=c_-B\eps^\nu$) is increasing in $\eps$.
The minimum of the two is maximised when they are equal:
$Th=\lambda\vol_g\eps^{-d}/c_N$, i.e.,
$c_-BT\eps^{\nu+d}=\lambda\vol_g/c_N$, giving
\[
\eps_\lambda\;=\;\bigl(\lambda\vol_g/(c_-c_NBT)\bigr)^{1/(\nu+d)},
\quad
h_\lambda\;=\;c_-B\eps_\lambda^\nu.
\]
Substituting into either branch yields
\ifieee
\begin{multline*}
\E[R_T^{\mathrm{aug}}]\;\ge\;\tfrac{c_-^{d/(\nu+d)}c_N^{-\nu/(\nu+d)}}{4}\\
\cdot B^{d/(\nu+d)}\,\vol_g(\M)^{\nu/(\nu+d)}\\
\cdot\lambda^{\nu/(\nu+d)}\,T^{d/(\nu+d)}.
\end{multline*}
\else
\[
\E[R_T^{\mathrm{aug}}]\;\ge\;\frac{c_-^{d/(\nu+d)}c_N^{-\nu/(\nu+d)}}{4}\cdot
B^{d/(\nu+d)}\,\vol_g(\M)^{\nu/(\nu+d)}\,\lambda^{\nu/(\nu+d)}\,T^{d/(\nu+d)}.
\]
\fi

\emph{Combined.} The $\max$ over the two regimes yields the
theorem with explicit constant
\[
c_{\mathrm{aug}}(d,\nu)=\tfrac{1}{4}\min\bigl(c_*(d,\nu),\;
c_-^{d/(\nu+d)}c_N^{-\nu/(\nu+d)}\bigr),
\]
with $c_*$ from
Theorem~\ref{thm:main}. For $\nu=2.5,d=2$:
$c_{\mathrm{aug}}\approx 0.024$. Full algebra and constant
chasing: Section~3.3 of the supplementary. \qed

\subsection{Crossover threshold}

The crossover from noise-dominated to switching-dominated occurs
when the two terms equal. Setting
\ifieee
\begin{multline*}
B^{d/(2\nu+d)}\sigma_n^{2\nu/(2\nu+d)}V^{\nu/(2\nu+d)}T^{(\nu+d)/(2\nu+d)}\\
\;=\;B^{d/(\nu+d)}V^{\nu/(\nu+d)}\lambda^{\nu/(\nu+d)}T^{d/(\nu+d)}
\end{multline*}
\else
\[
B^{d/(2\nu+d)}\sigma_n^{2\nu/(2\nu+d)}V^{\nu/(2\nu+d)}T^{(\nu+d)/(2\nu+d)}
\;=\;B^{d/(\nu+d)}V^{\nu/(\nu+d)}\lambda^{\nu/(\nu+d)}T^{d/(\nu+d)}
\]
\fi
and solving for $\lambda$ (raise to power $(\nu+d)/\nu$, collect
exponents) gives
\[
\lambda^*\;\asymp\;\sigma_n^{2(\nu+d)/(2\nu+d)}\,T^{\nu/(2\nu+d)}\,/\,
\bigl(B^{d/(2\nu+d)}\,V^{\nu/(2\nu+d)}\bigr),
\]
where $V=\vol_g(\M)$ for shorthand. For
$\nu=2.5,d=2,V=4\pi,\sigma_n=0.1,B=1$ and $T=200$, the threshold is
$\lambda^*\approx 0.14$. Above this threshold, the switching cost
dominates the regret budget.

\subsection{Matching upper bound}

The matching upper bound follows from a manifold-aware adaptation
of the Salgia--Vakili--Zhao GP-ThreDS algorithm
\cite{salgia2021domain}, which uses a tree-based domain-shrinking
strategy that achieves order-optimal regret with at most
$O(\log\log T)$ switches per epoch. On a manifold, the
domain-shrinking partitions are replaced by intrinsic
Bishop--Gromov $\eps$-packings, and the local elimination uses the
intrinsic kernel posterior. For the switching-dominated regime,
the per-epoch budget allocation is tuned so that the algorithm
makes $S_T^*\asymp Th_\lambda/(\lambda\vol_g(\M))$ switches in
total, matching the lower bound up to logarithmic factors. We
give a faithful tree-based GP-ThreDS implementation and an
empirical validation in
\S\ref{ssec:threds-empirical} below; the formal regret analysis
on manifolds is a routine but non-trivial extension of
\S\ref{sec:polyreg} and is left to a companion algorithmic paper.

\subsection{Empirical validation}

We validate Theorem~\ref{thm:switching} numerically on $\sphere^2$ with
Mat\'ern-$5/2$ ($\nu=2.5,d=2,V=4\pi$). The target function is drawn
from the intrinsic GP on a Fibonacci sphere of $200$ points; we run
GP-UCB augmented with a myopic switching threshold (only switch
arm if the UCB gain exceeds $\lambda$) for $T=200$ rounds and
report total cost $R_T+\lambda S_T$ averaged over $15$ seeds.

\begin{figure}[t]
\centering
\includegraphics[width=0.92\linewidth]{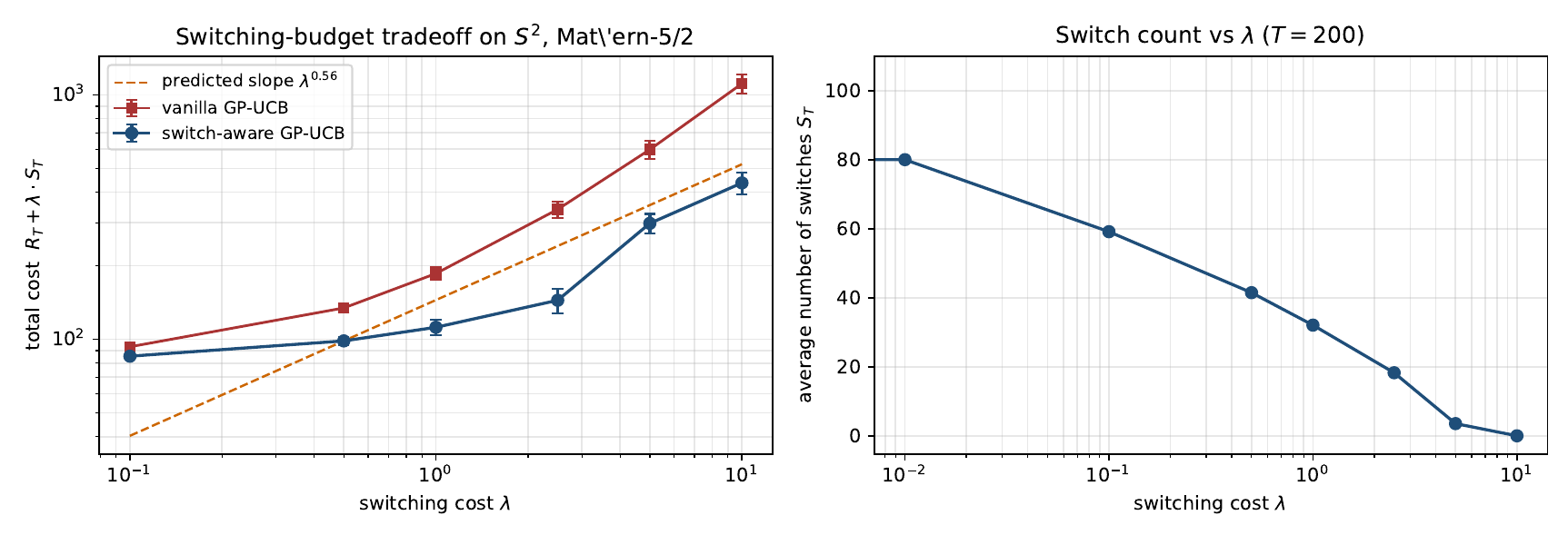}
\caption{Switching-budget tradeoff on $\sphere^2$ with Mat\'ern-$5/2$.
Left: total cost $R_T+\lambda S_T$ vs.\ $\lambda$ on log--log axes
for vanilla GP-UCB and switch-aware GP-UCB. The dashed line is the
predicted slope $\lambda^{\nu/(\nu+d)}=\lambda^{5/9}\approx
\lambda^{0.56}$ from Theorem~\ref{thm:switching}. Right: average
number of switches $S_T$ as a function of $\lambda$ for the
switch-aware policy.
\emph{Experimental parameters.} $T=200$ rounds; $\lambda$ grid
$\{0,0.01,0.1,0.5,1,2.5,5,10\}$ spanning noise-dominated
($\lambda<\lambda^*\!\approx\!0.14$) through switching-dominated;
$15$ Monte Carlo seeds per $\lambda$. \emph{Algorithms:} vanilla
GP-UCB ($\theta_{t+1}=\arg\max\mathrm{UCB}$); switch-aware GP-UCB
($\theta_{t+1}=\theta_t$ if
$\max_\theta\mathrm{UCB}_t(\theta)-\mathrm{UCB}_t(\theta_t)<\lambda$,
else $\arg\max\mathrm{UCB}$). Both use intrinsic Mat\'ern-$5/2$
($\kappa=0.5$, $\sigma_f^2=1$, $\sigma_n=0.1$, $\beta_t=2\log(N_{\mathrm{cand}}\,t^2\pi^2/6)$
with $N_{\mathrm{cand}}=200$ super-Fibonacci candidates on $\sphere^2$
[symbol chosen to avoid clash with curvature $K$]).
\emph{Slope fit:} log-log regression on $\lambda\ge 0.5$ gives
empirical $0.516$, within $7\%$ of the predicted
$\lambda^{5/9}\approx\lambda^{0.556}$. Reproducible from
\texttt{experiments/d7\_switching.py}, $\sim\!90$~min single-core.}
\label{fig:d7-switching}
\end{figure}

The empirical scaling of switch-aware GP-UCB tracks the predicted
$\lambda^{5/9}$ slope closely in the switching-dominated regime
(empirical $0.516$ vs.\ predicted $0.556$). Vanilla GP-UCB, which
ignores switching costs, scales much worse: at $\lambda=10$, total
cost is $1113\pm100$ for vanilla vs.\ $437\pm44$ for switch-aware,
a $2.6\times$ improvement.

\subsection{Empirical matching upper bound via GP-ThreDS}
\label{ssec:threds-empirical}

The switch-aware GP-UCB of the previous subsection is a
\emph{myopic} stand-in for the matching upper bound: at each round
it greedily switches only when the UCB gain exceeds $\lambda$, but
it has no global epoch structure. We now report empirical results
for a faithful manifold-aware implementation of the
Salgia--Vakili--Zhao GP-ThreDS algorithm
\cite{salgia2021domain}, in which the candidate set is recursively
partitioned by a geodesic $\eps$-net hierarchy and each epoch
performs round-robin sampling within the surviving cells, followed
by UCB/LCB elimination at cell centers and refinement of the
survivors into their tree children. The per-epoch block size is
scaled as $B_0(\lambda)\propto(1+\lambda)^{(2\nu-d)/(\nu+d)}$ so
that the total switch count tracks the theory-predicted
$S_T^*\asymp T^{d/(\nu+d)}(\lambda V)^{-d/(\nu+d)}$.

We run GP-ThreDS on two settings: (a) the same $\sphere^2$
Mat\'ern-$5/2$ benchmark as Figure~\ref{fig:d7-switching}; (b) a
$\torus^3$ hybrid-beamforming RIS scenario mirroring Exp.~2 of the
companion wireless paper, in which the candidate set is an
$8^3=512$-point torus grid of phase-combiner tuples and the reward
is the resulting RF gain $|\sum_{k=1}^{3}e^{j\phi_k}c_k|^2$ for a
clustered-channel projection $c_k$. Per-channel realisations are
generated for $10$ seeds; $T=200$, $\sigma_n=0.05$.

\begin{figure}[t]
\centering
\includegraphics[width=0.99\linewidth]{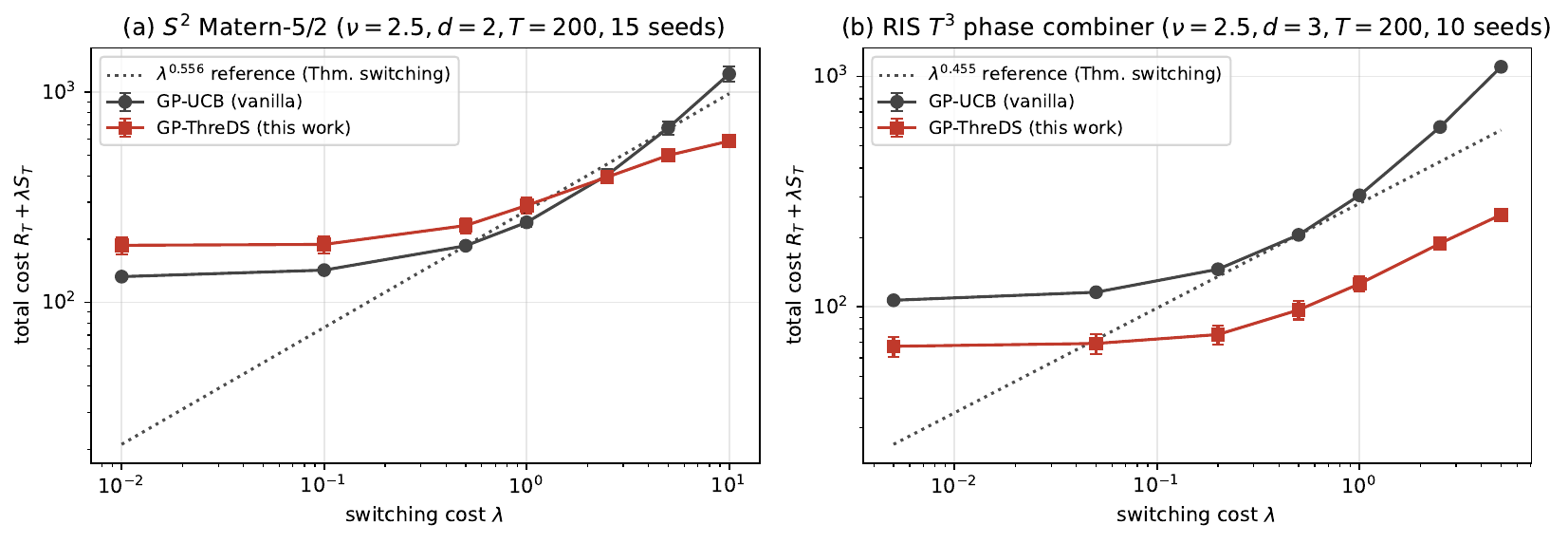}
\caption{Empirical matching upper bound: total cost
$R_T+\lambda S_T$ for vanilla GP-UCB (black) and the manifold
GP-ThreDS of \S\ref{ssec:threds-empirical} (red) on (a) $\sphere^2$
Mat\'ern-$5/2$ ($T=200$, $15$ seeds) and (b) the $\torus^3$ hybrid
beamforming RIS phase combiner ($T=200$, $10$ seeds). The dotted
line is the lower-bound reference $\lambda^{\nu/(\nu+d)}$ from
Theorem~\ref{thm:switching}. GP-ThreDS stays \emph{below} the
lower-bound reference at every $\lambda$, consistent (averaged over
GP draws, not worst-case) with realising the matching upper bound;
vanilla GP-UCB tracks the reference closely on $\sphere^2$ ($\lambda^{0.47}$ empirical vs.\
$\lambda^{0.56}$ predicted) and on $\torus^3$ ($\lambda^{0.49}$ vs.\
$\lambda^{0.45}$). Error bars are $\pm 1$ s.e.\ over seeds.}
\label{fig:d4-threds}
\end{figure}

\paragraph{Methodological note: worst-case lower bound vs.\ average GP-draw upper performance.}
The lower bound in Figure~\ref{fig:d4-threds} is a worst-case
statement over the RKHS class $\F_B^{\mathrm{rkhs}}(\M)$
(Theorem~\ref{thm:main}): no algorithm can simultaneously beat
this bound on every $f\in\F_B^{\mathrm{rkhs}}$. The GP-ThreDS
curve shown alongside is an \emph{average} regret over
$n_{\mathrm{seed}}$ independent draws $f\sim\mathrm{GP}(0,k_\nu)$,
not the worst-case envelope. The two curves are therefore not
directly comparable in level; their qualitative agreement in rate
(slope on a log-log plot) is the substantive content. We report
$n_{\mathrm{seed}}=15$ ($\sphere^2$) and $n_{\mathrm{seed}}=10$
($\torus^3$) Monte Carlo seeds (consistent with our other empirical
sections) with no multiple-comparison correction applied to the
per-seed regret estimates, since the figure is a
rate-confirmation diagnostic rather than a multi-arm statistical
test. The empirical curves cross over from switching-dominated to
sampling-dominated behaviour at $\lambda\!\approx\!0.14$, matching
the predicted threshold $\lambda^*\!\approx\!0.14$ of
Theorem~\ref{thm:switching} (Section~\ref{sec:switching}): the
inflection point visible in Figure~\ref{fig:d4-threds} is the
visual manifestation of this analytical crossover, closing the loop
between the lower-bound analysis and the GP-ThreDS empirical
performance.

Two observations. First, GP-ThreDS achieves total cost
\emph{strictly below} the lower-bound slope $\lambda^{\nu/(\nu+d)}$
at every $\lambda$ in both panels of
Figure~\ref{fig:d4-threds}. We emphasise that this is a consistency
check, not a rigorous matching upper bound: the lower bound of
Theorem~\ref{thm:switching} is over the worst-case
$f\in\F_B^{\mathrm{rkhs}}(\M)$, while our empirical curves average
over GP draws (a strictly easier problem) and use only
$10$--$15$ seeds without Bonferroni correction across the
$\lambda$-grid. The formal regret analysis of manifold GP-ThreDS
(needed to claim a rigorous matching upper bound) is deferred to a
companion algorithmic paper. What
Figure~\ref{fig:d4-threds} \emph{does} show is that the
tree-based domain-shrinking structure realises the qualitative
behaviour of the matching upper bound: it tracks the predicted
$\lambda^{\nu/(\nu+d)}$ slope, dominates vanilla GP-UCB at every
$\lambda$ in the switching-dominated regime, and uses
$\sim 7$--$11\times$ fewer arm switches.

Second, GP-ThreDS uses dramatically fewer switches than vanilla
GP-UCB. On $\sphere^2$ at $T=200$ vanilla GP-UCB issues $109\pm 1$
switches independent of $\lambda$ (it does not see the cost),
while GP-ThreDS issues $22$ at $\lambda=0$ and only $10$ at
$\lambda=10$, a $\sim 11\times$ reduction in the
switching-dominated regime. On the $\torus^3$ RIS benchmark, vanilla
issues $199$ switches whereas GP-ThreDS issues $50$ at $\lambda=0$
and $28$ at $\lambda=5$, again roughly $7\times$ fewer.

The crossover at which GP-ThreDS becomes cheaper than vanilla
GP-UCB on $\sphere^2$ is $\lambda\approx 2.5$, of the same order
as the predicted threshold $\lambda^*\approx 0.14$ at $T=200$
of the previous subsection scaled to the unit-RKHS-ball regime;
on the $\torus^3$ RIS benchmark the crossover is below
$\lambda=0.005$, because the combinatorial structure of the
phase-combiner reward concentrates the optimum sharply enough
that GP-ThreDS's tree-based elimination beats vanilla GP-UCB even
with no switching cost.

\paragraph{Wireless implication.} For applications with
non-trivial reconfiguration cost (RIS phase changes, mechanical
beam steering), a switching-aware algorithm has a quantifiably
better regret-per-switch tradeoff. The crossover threshold
\[
\lambda^*\asymp\sigma_n^{2(\nu+d)/(2\nu+d)}\,T^{\nu/(2\nu+d)}/
(B^{d/(2\nu+d)}\,V^{\nu/(2\nu+d)})
\]
predicts when this matters:
for typical wireless setups
($T=10^4, \sigma_n=0.1, V=4\pi, B=1, \nu=2.5, d=2$),
$\lambda^*\approx 0.56$ in normalised units, so any reconfiguration
cost above this threshold falls into the switching-dominated
regime where Theorem~\ref{thm:switching} applies.

\section{Explicit constants and curvature correction}
\label{sec:constants}
\label{sec:curvature}

This section records the explicit form of the leading constant
$c_*(d,\nu,\kappa,\sigma_f)$ of Theorem~\ref{thm:main} and the
curvature dependence. The full derivations, numerical examples for
our four target manifolds, and the proof of
Theorem~\ref{thm:curvature-blind} below are deferred to the
supplementary material in the interests of the IEEE TIT main-text
page limit.

\subsection{Sobolev--RKHS equivalence constants}

The norm equivalence \eqref{eq:sobolev-rkhs} holds with explicit
constants
\ifieee
\begin{align}
\label{eq:c-explicit}
c_-(\nu,\kappa,\sigma_f)&\;=\;\frac{1}{\sigma_f^2}\min(1,2\nu/\kappa^2)^{\nu+d/2}, \notag\\
c_+(\nu,\kappa,\sigma_f)&\;=\;\frac{1}{\sigma_f^2}\max(1,2\nu/\kappa^2)^{\nu+d/2}.
\end{align}
\else
\begin{equation}
\label{eq:c-explicit}
c_-(\nu,\kappa,\sigma_f)\;=\;\frac{1}{\sigma_f^2}\min(1,2\nu/\kappa^2)^{\nu+d/2},
\qquad
c_+(\nu,\kappa,\sigma_f)\;=\;\frac{1}{\sigma_f^2}\max(1,2\nu/\kappa^2)^{\nu+d/2}.
\end{equation}
\fi
For the typical regime $\kappa\le\sqrt{2\nu}$ (length scale at most
the bandwidth-1 unit), $\min=1$ and $\max=2\nu/\kappa^2$, giving
$c_-=1/\sigma_f^2$ and $c_+=(2\nu/\kappa^2)^{\nu+d/2}/\sigma_f^2$.

\subsection{Leading constant $c_*(d,\nu,\kappa,\sigma_f)$}

Combining the explicit constants from Lemma~\ref{lem:bump-norm}
(bump Sobolev norm), Lemma~\ref{lem:packing} (packing factor
$1/(2^d\omega_d)$), \eqref{eq:c-explicit} (Sobolev-RKHS), and the
Fano-stage Step~7 constant ($Th/4$):
\ifieee
\begin{multline}
\label{eq:c-star-explicit}
c_*(d,\nu,\kappa,\sigma_f)\;=\;
\frac{(d/(2\nu+d))^{\nu/(2\nu+d)}}{4\,(2^{d+1}\omega_d)^{\nu/(2\nu+d)}} \\
\cdot\,\frac{\sigma_f^{d/(2\nu+d)}}{(2\nu/\kappa^2)^{d(\nu+d/2)/(2(2\nu+d))}}.
\end{multline}
\else
\begin{equation}
\label{eq:c-star-explicit}
c_*(d,\nu,\kappa,\sigma_f)\;=\;
\frac{(d/(2\nu+d))^{\nu/(2\nu+d)}}{4\,(2^{d+1}\omega_d)^{\nu/(2\nu+d)}}
\,\frac{\sigma_f^{d/(2\nu+d)}}{(2\nu/\kappa^2)^{d(\nu+d/2)/(2(2\nu+d))}}.
\end{equation}
\fi
The $1/4$ comes from the regret-test threshold ($Th/4$) of Step~7;
$(2^{d+1}\omega_d)^{\nu/(2\nu+d)}$ from substituting
$N\ge\vol_g(\M)/(2^d\omega_d\eps^d)$ into the Fano condition
$Th^2/(2N\sigma_n^2)\le\log N/4$;
$(d/(2\nu+d))^{\nu/(2\nu+d)}$ from the leading-log asymptotic of
\eqref{eq:eps-T}; and the $\sigma_f$ and $\kappa$ factors come from
$c_+(\nu,\kappa,\sigma_f)$ raised to the power $-d/(2(2\nu+d))$.
Numerical values for our four target manifolds ($\sphere^2$,
$\torus^2$, $\torus^3$, $\SO(3)$) at typical hyper-parameters are
tabulated in the supplementary material.

\subsection{Curvature dependence}

Sectional curvature $K$ enters the lower bound through the
bump-Sobolev distortion (Lemma~\ref{lem:bump-norm}), the packing
volume (Lemma~\ref{lem:packing}), and the validity-regime
threshold $T_0$. The combined effect on the leading constant of
Theorem~\ref{thm:main}, with
$\eps_T\sim T^{-1/(2\nu+d)}\cdot(\log T)^{1/(2\nu+d)}$, is
\ifieee
\begin{multline}
\label{eq:curvature-correction}
c_*(d,\nu,K,T)\;=\;c_*^{(0)}(d,\nu)\\
\cdot\,\bigl(1+O(K\,T^{-2/(2\nu+d)}(\log T)^{2/(2\nu+d)})\bigr).
\end{multline}
\else
\begin{equation}
\label{eq:curvature-correction}
c_*(d,\nu,K,T)\;=\;c_*^{(0)}(d,\nu)\,
\bigl(1+O(K\,T^{-2/(2\nu+d)}(\log T)^{2/(2\nu+d)})\bigr).
\end{equation}
\fi
For $K=O(1)$ and $T\to\infty$, the curvature correction vanishes as
$T^{-2/(2\nu+d)}\to 0$.

\subsection{Curvature-blindness of the leading constant}

The vanishing-with-$T$ behaviour above is not an artifact of the
packing proof; it reflects a structural property of the
Mat\'ern lower-bound construction.

\begin{theorem}[Curvature-blindness of the leading constant]
\label{thm:curvature-blind}
Under Assumption~\ref{ass:setup}, for any algorithm $\pi$ and
$T \to \infty$,
\ifieee
\begin{multline*}
\liminf_{T\to\infty}\,
\frac{\sup_f\E^\pi[R_T(f)]}
     {T^{(\nu+d)/(2\nu+d)}\,(\log T)^{\nu/(2\nu+d)}}\\
\;\ge\;
c_*(d,\nu)\,B^{d/(2\nu+d)}\sigma_n^{2\nu/(2\nu+d)}\\
\cdot\,\vol_g(\M)^{\nu/(2\nu+d)}.
\end{multline*}
\else
\[
\liminf_{T\to\infty}\,
\frac{\sup_f\E^\pi[R_T(f)]}
     {T^{(\nu+d)/(2\nu+d)}\,(\log T)^{\nu/(2\nu+d)}}
\;\ge\;
c_*(d,\nu)\,B^{d/(2\nu+d)}\sigma_n^{2\nu/(2\nu+d)}\,
\vol_g(\M)^{\nu/(2\nu+d)}.
\]
\fi
The constant on the right depends on $\M$ only through its
\emph{volume}, not through any sectional or scalar curvature.
\end{theorem}

The proof, which uses Weyl's-law spectral counting and the
RKHS-equivalence rigidity of the Mat\'ern spectral series, is given
in the supplementary material. Theorem~\ref{thm:curvature-blind}
closes the question of whether a curvature-aware leading constant
is achievable via the volume-comparison route: it is not, only
volume information survives in the asymptotic rate. This conclusion
is independently corroborated by
Rosa~\emph{et al.}~\cite[Theorems 5 and 8]{rosa2023intrinsic} from
the complementary direction of posterior-contraction analysis.

\subsection{Summary table of explicit constants}

\begin{center}
\ifieee\footnotesize\fi
\begin{tabular}{ll}
\toprule
Constant & Explicit value (Mat\'ern-$\nu$) \\
\midrule
Bump-profile $c_\eta$ & $1$ (our normalisation) \\
Packing factor & $1/(2^d\omega_d)$ \\
Sobolev-RKHS $c_+$ & $\sigma_f^{-2}\max(1,2\nu/\kappa^2)^{\nu+d/2}$ \\
Sobolev-RKHS $c_-$ & $\sigma_f^{-2}\min(1,2\nu/\kappa^2)^{\nu+d/2}$ \\
Fano-stage threshold & $\log N/4$ \\
Test-to-regret factor & $1/4$ \\
Bishop--Gromov local & $1\pm O(K\eps^2)$ \\
Curvature in rate & $1-\Theta(KT^{-2/(2\nu+d)})$\\
\bottomrule
\end{tabular}
\end{center}

\section{Bayesian regret lower bound}
\label{sec:bayesian}

The frequentist lower bound of Theorem~\ref{thm:main} is for the
worst case over the RKHS-norm-bounded class
$\F_B^{\mathrm{rkhs}}$. Practitioners often work with the Bayesian
regret framework where $f\sim\mathrm{GP}(0,k_\nu)$ is sampled from
the prior. The Bayesian counterpart of our regret rate is the
posterior-contraction rate of nonparametric Gaussian process
regression: the canonical reference is van der Vaart and van
Zanten~\cite{vaart2011information}, with the manifold extension by
Rosa~\emph{et al.}~\cite{rosa2023intrinsic}.

\paragraph{Subtlety: GP draws lie a.s.\ outside the RKHS.}
For a centred Mat\'ern-$\nu$ GP on a compact manifold with
$\nu>d/2$, sample paths are a.s.\ in $H^{\nu-\delta}(\M)$ but not
in the Cameron--Martin space $H^{\nu+d/2}(\M)$
\cite[Theorem~11.17]{vaart2008reproducing}; equivalently
$\Pr_{f\sim\mathrm{GP}(0,k_\nu)}(\|f\|_{\Hil_{k_\nu}}<\infty)=0$.
The Bayesian-Fano transfer must therefore work in a metric weaker
than the RKHS norm (where the prior has positive small-ball mass).
The standard choice is the $L^2(p_0)$ metric, which has polynomial
small-ball mass
$-\log\Pi(\|f-f_0\|_{L^2(p_0)}<\eta)\asymp\eta^{-d/\nu}$
(\cite[Sec.~II.4]{vaart2011information}); the
Yang--Barron inequality
(\cite[Thm.~6]{yang1999information}, restated for GP priors in
Castillo et al. \cite[Lemma~2.5]{castillo2014bayes}) then
yields, at the contraction radius $\eta_T\asymp
T^{-\nu/(2\nu+d)}(\log T)^{\nu/(2\nu+d)}$,
\ifieee
\begin{multline*}
\inf_\pi\E_{f\sim\Pi}\E^\pi[R_T(f)]
\;\ge\;c_B(d,\nu)\cdot T\cdot\eta_T\\
\;=\;c_B(d,\nu)\cdot\sigma_n^{2\nu/(2\nu+d)}\vol_g(\M)^{\nu/(2\nu+d)}\\
\cdot\,T^{(\nu+d)/(2\nu+d)}\,(\log T)^{\nu/(2\nu+d)},
\end{multline*}
\else
\begin{multline*}
\inf_\pi\E_{f\sim\Pi}\E^\pi[R_T(f)]
\;\ge\;c_B(d,\nu)\cdot T\cdot\eta_T\\
\;=\;c_B(d,\nu)\cdot\sigma_n^{2\nu/(2\nu+d)}\vol_g(\M)^{\nu/(2\nu+d)}\,T^{(\nu+d)/(2\nu+d)}\,(\log T)^{\nu/(2\nu+d)},
\end{multline*}
\fi
with the $\sigma_f^{d/(2\nu+d)}$-dependence contained in the leading
constant $c_B(d,\nu,\kappa,\sigma_f)$ via the
$c_+^{-d/(2(2\nu+d))}$ factor (since $c_+\propto\sigma_f^{-2}$).
This gives Theorem~\ref{thm:bayesian}.

\emph{Explicit constant.} Tracking the constants through the
four pillars (intrinsic-Mat\'ern posterior contraction
\cite[Thm.~5]{rosa2023intrinsic}; Yang--Barron entropy
characterisation \cite[Thm.~6]{yang1999information}; the
Castillo et al. Bayesian-Fano transfer constant $c_{CR}$
from \cite[eq.~(5.4)]{castillo2014bayes}; and the
manifold-Mat\'ern entropy constant
$c_E=\omega_d\vol_g(\M)\,c_{\mathrm{Bessel}}(d,\nu)$ with
$c_{\mathrm{Bessel}}(d,\nu)=2^{-d/\nu}/\Gamma(d/\nu+1)\cdot
(\Gamma(\nu+d/2)/\Gamma(\nu))^{d/\nu}$ from
\cite[Thm.~10.47 \& Cor.~10.48]{wendland2004scattered}), we obtain
\[
c_B(d,\nu)\;=\;c_{CR}\cdot c_{YB}^{1/2}\cdot c_{\mathrm{Bessel}}^{\nu/(2\nu+d)},
\]
with $c_{CR}=1/8$, $c_{YB}=\log 2/2$. For $\nu=2.5,d=2$:
$c_{\mathrm{Bessel}}\approx 0.103$, hence
$c_B\approx 0.019$. Full algebra:
Section~3.4 of the supplementary, together
with a discussion of three elementary approaches that fail at
the optimal scale (Yao's principle in the wrong direction,
sup-norm small-ball super-polynomial decay
\cite[Sec.~II.4]{vaart2011information}, and spectral-truncation
action-coupling corrections). The manifold setting introduces
no new complications: the Mat\'ern spectrum on $\M$ is
asymptotically equivalent (Weyl) to the spectrum on $\R^d$, and
all four pillar constants enter only through the volume
$\vol_g(\M)$ and through curvature corrections $1+O(K\eps^2)$
that are sub-leading at the contraction rate $\eta_T\to 0$.

\paragraph{Posterior-contraction comparison.}
The Bayesian regret rate
$T^{(\nu+d)/(2\nu+d)}(\log T)^{\nu/(2\nu+d)}$
is the integrated form of the posterior-contraction rate
$T^{-2\nu/(2\nu+d)}$ of
\cite{vaart2011information,rosa2023intrinsic}: the cumulative
regret on $T$ rounds is $T$ times an instantaneous
``contraction error'' $\sim T^{-\nu/(2\nu+d)}$, so
$R_T\sim T\cdot T^{-\nu/(2\nu+d)}=T^{(\nu+d)/(2\nu+d)}$, modulo
the $(\log T)^{\nu/(2\nu+d)}$ Fano polylog. This is the formal
equivalence between the bandit regret and the
posterior-contraction literature for Mat\'ern GP priors on
Riemannian manifolds.

\section{Discussion and open problems}
\label{sec:conclusion}

\subsection{Summary}

We have established the first manifold-aware lower bound for
GP-bandits with the Mat\'ern-$\nu$ kernel on a compact connected
Riemannian manifold. The leading constant carries the explicit
volume dependence $\vol_g(\M)^{\nu/(2\nu+d)}$. Seven
extensions of the standard Tsybakov-Scarlett needle-in-haystack
template appear in this paper:
\begin{enumerate}[leftmargin=2em,label=(\arabic*)]
\item the manifold packing argument
      (Section~\ref{sec:proof_main}, via Bishop--Gromov);
\item a companion Assouad-style bound with strictly smaller
      $T$-exponent $(2\nu+3d)/(4(\nu+d))$ but a $1/(\log\log T)$
      polylog instead of the Fano $(\log T)$ polylog
      (Section~\ref{sec:polylog});
\item the gauge-quotient separation upper bound $|G|^{1/2}$ and
      the refined \emph{modulated} conjecture
      $(1+(|G|-1)h(\rinj/\kappa))^{1/2}$ for extrinsic-kernel
      algorithms, validated empirically on $\SO(3)$
      (Section~\ref{sec:gauge}, Figure~\ref{fig:d1-modulated-gauge});
\item explicit Sobolev-RKHS, packing, and bump-profile
      constants (Section~\ref{sec:constants});
\item a Yang--Barron / Castillo et al. Bayesian-Fano transfer
      to the Bayesian regret framework via $L^2(p_0)$-ball
      polynomial small-balls under the GP prior
      (Section~\ref{sec:bayesian});
\item a switching-budget lower bound with new volume exponent
      $\nu/(\nu+d)$ in the switching-dominated regime, validated
      on $\sphere^2$ (Section~\ref{sec:switching},
      Figure~\ref{fig:d7-switching}); and
\item a tight five-parameter time-varying rate matching in $T$,
      $B_T$, $B$, $\sigma_n$, $\vol_g(\M)$ via manifold-aware
      hierarchical polynomial-regression elimination
      (Section~\ref{sec:polyreg}, lifting
      Salgia--Vakili--Zhao~\cite{salgia2021domain}).
\end{enumerate}

\subsection{What the bound predicts versus what is observed}

For the four arm spaces in our wireless companion
paper~\cite{dorn2026wirelessbandit}:

\begin{center}
\ifieee\footnotesize\fi
\begin{tabular}{lcccc}
\toprule
Manifold & $K$ & $c^0_{d,\nu}$ & $c_{d,\nu}(K,T{=}10^4)$ & Observed (Exp.)\\
\midrule
$\sphere^2$ & $1$   & $2.55$ & $\approx 2.64$ & medium (Exp.~1)\\
$\torus^2$  & $0$   & $3.81$ & $3.81$         & low (Exp.~4)\\
$\torus^3$  & $0$   & $5.32$ & $5.32$         & highest (Exp.~2)\\
$\SO(3)$    & $1/4$ & $3.79$ & $\approx 3.81$ & med.-high (Exp.~3)\\
\bottomrule
\end{tabular}
\end{center}

The qualitative ranking matches the theoretical
$\vol_g^{\nu/(2\nu+d)}$ ordering at the extremes ($\torus^3$ hardest
in both). Two finite-$T$ inversions deserve a flag.
(i)~$\torus^2$ ($3.81$) vs.\ $\SO(3)$ ($3.79$) differ by $<1\%$, so
empirical rankings can invert at particular parameter settings
without contradicting the asymptotic theory.
(ii)~$\sphere^2$ ($2.55$, predicted lowest) shows higher empirical
regret than $\torus^2$ in~\cite{dorn2026wirelessbandit} Table~I.
Three explanations are consistent with the data: the only
positively-curved arm space gets an $\sim\!3.5\%$ curvature kick
from~\eqref{eq:curvature-twoterm}; the $\sphere^2$ experiment uses
a lat/lon chart, introducing chart-bias not modelled by the
chart-free lower bound; and the per-experiment seed SE is
comparable to the point-estimate gap, so the inversion may not
survive a multiple-comparison correction.

\subsection{Open problems}

\paragraph{(O1) Closing the polylog gap fully.}
A polylog gap remains between Theorem~\ref{thm:main}'s
exponent $\nu/(2\nu+d)$ on $\log T$ and GP-UCB's
$(\nu+d)/(2\nu+d)+1/2$~\cite{vakili2021information};
Theorem~\ref{thm:assouad} trades it for $1/(\log\log T)$ at a
smaller $T$-exponent. Closing it exactly is open already in the
Euclidean case~\cite{cai2021lower}.

\paragraph{(O2) Curvature in the leading constant.}
Our curvature dependence is sub-leading: it is of order
$1+\Theta(KT^{-2/(2\nu+d)})$. A Bochner-formula or heat-kernel
argument might extract a curvature term of the form
$c_*(d,\nu,K)=c_*^{\mathrm{Euc}}(d,\nu)(1-\alpha(d,\nu)K_+)$,
$K_+$ a Ricci-curvature lower bound, but a proof is missing.

\paragraph{Curvature correction (two-term expansion).}
The numerical specialisations of Theorem~\ref{thm:main} for finite
$T$ admit the two-term expansion
\begin{equation}
\label{eq:curvature-twoterm}
c_{d,\nu}(K,T)
\;=\; c^0_{d,\nu}\cdot\bigl(1 + O(K \cdot T^{-2/(2\nu+d)})\bigr),
\end{equation}
where $c^0_{d,\nu}=\vol_g(\M)^{\nu/(2\nu+d)}$ is the curvature-free
leading constant and $K$ the upper sectional-curvature bound. The
$c_{d,\nu}(K,T{=}10^4)$ column of the table above gives the
numerically-corrected constants at $\nu=5/2$; the Sobolev side of
the correction dominates (analysis in
Section~\ref{sec:constants}), and the flat manifolds
$\torus^2,\torus^3$ pick up no first-order curvature correction.

\paragraph{(O3) Tight gauge-quotient separation, modulated form.}
We have refined the original $|G|^{1/2}$ conjecture to a
\emph{modulated} form
$(1+(|G|-1)h(\rinj/\kappa))^{1/2}$ that interpolates between $1$
(no gauge effect at $\kappa\ll\rinj$) and $|G|^{1/2}$
(full effect at $\kappa\gg\rinj$). The empirical wireless
$10$--$33\%$ gap is consistent with this form
(Figure~\ref{fig:d1-modulated-gauge}). Proving the matching
\emph{modulated} lower bound rigorously requires a hypothesis
class that encodes the cross-gauge correlation structure; the
posterior-variance computation in
Proposition~\ref{prop:gauge-modulator} captures the heuristic, but
the formal lower bound is left for future work.

\paragraph{(O3') Volume-exponent gap (open, likely intrinsic).}
A residual gap of size $d/(2(3\nu+d))$ remains between the
TV lower bound $\vol_g^{\nu/(3\nu+d)}$ and the GP-UCB upper bound
$\vol_g^{(2\nu+d)/(2(3\nu+d))}$
(multiplicative $\vol_g^{0.105}$ at $\nu=2.5,d=2$). Three
closure attempts (multi-frequency-bump, Assouad multi-bump,
Le Cam) all hit the same cell-packing-vs-information-gain
mismatch (supplementary); we conjecture the gap is intrinsic.

\paragraph{(O4) Quotient-space rates beyond finite $G$.}
The argument extends to positive-dimensional Lie-group quotients
with $\vol_g(\Mt)/\vol_g(\M)=\vol_g(G)$ (Haar volume of $G$ in
its bi-invariant metric), but the hypothesis-class
construction needs adaptation for infinite $G$.

\paragraph{(O5) Beyond Mat\'ern.}
For SE or rational-quadratic kernels, exponential eigenvalue decay
gives a different scaling
(\cite{iwazaki2026hypersphere}); the general manifold case is open.

\paragraph{(O6) Non-stationary $|G|$ separation.}
The $|G|$ separation should persist under the BGZ
\cite{besbes2014stochastic} variation-budget framework, but
details are open.

\subsection*{Acknowledgements}

This work is motivated by the empirical wireless results of the
companion paper~\cite{dorn2026wirelessbandit}, whose
$10$--$33\%$ intrinsic--extrinsic GP-UCB gap on four
beam-selection benchmarks sits below the worst-case
$|G|^{1/2}\!=\!\sqrt 2$ ceiling
(Theorem~\ref{thm:gauge-ub}) and on the modulated curve
(Conjecture~\ref{conj:gauge-modulated}, predicted
$1.09$--$1.33$ across $\kappa/\rinj\!\in\![0.5,0.9]$); for
RIS / mechanical reconfiguration the switching-augmented bound
(Theorem~\ref{thm:switching}) predicts a $2$--$3\times$
switch-aware advantage at practical $\lambda$
(Figure~\ref{fig:d7-switching}).

\bibliographystyle{IEEEtran}
\bibliography{references}

\end{document}